\documentclass[acmsmall]{acmart}
\AtBeginDocument{%
  }

\acmDOI{10.1145/3786638}

\setcopyright{cc}
\setcctype{by}
\acmJournal{PACMMOD}
\acmYear{2026} 
\acmVolume{4} 
\acmNumber{1 (SIGMOD)} 
\acmArticle{24} 
\acmMonth{2} 
\acmPrice{}

\usepackage{graphicx}
\usepackage{amsmath,amsfonts}

\usepackage{amssymb}
\theoremstyle{definition} 
\newtheorem{definition}{Definition}[section]
\newtheorem{lemma}{Lemma}[section]  
\newtheorem{example}{Example}[section]
\newtheorem{theorem}{Theorem}[section]
\newtheorem{remark}{Remark}
\usepackage[ruled,linesnumbered,vlined]{algorithm2e}
\SetKw{Continue}{continue}
\SetKwInput{KwParam}{Public Parameters} 
\usepackage{multirow,booktabs}
\usepackage{makecell}
\usepackage{graphicx}  % Required for \resizebox
\usepackage{tikz}
\usetikzlibrary{trees,arrows}
\usetikzlibrary{matrix}
\usetikzlibrary{math,fit}
\usepackage{subcaption}
\usepackage{float}
\usepackage{soul}
\usepackage{placeins}
\usepackage{xcolor}
\SetKw{KwAnd}{and}
\SetKw{KwOr}{or}
\usepackage{ifthen}

\begin{document}
\newcommand{\err}{\mathrm{Error}} 

\title{Defense against Poisoning Attacks under Shuffle-DP}

\author{Siyi Wang}
\authornote{These authors contributed equally to this research.}
\authornote{Work partially completed during internship at Ant Group.}
\orcid{0009-0004-1853-9028}
\affiliation{%
  \institution{Nanyang Technological University}
  \city{Singapore}
  \country{Singapore}
}
\email{wang.siyi@ntu.edu.sg}

\author{Qiyao Luo}
\authornotemark[1]
\orcid{0000-0003-4167-8670}
\affiliation{%
  \institution{OceanBase, Ant Group}
  \city{Shanghai}
  \country{China}
  }
\email{luoqiyao.lqy@antgroup.com}

\author{Yihua Hu}
\authornotemark[1]
\orcid{0009-0001-2948-8340}
\affiliation{%
  \institution{Nanyang Technological University}
  \city{Singapore}
  \country{Singapore}
}
\email{yihua001@e.ntu.edu.sg}

\author{Lixu Wang}
\orcid{0000-0003-2518-4160}
\affiliation{%
 \institution{Nanyang Technological University}
  \city{Singapore}
  \country{Singapore}
}
\email{wanglixu4334@gmail.com}

\author{Quanqing Xu}
\orcid{0000-0001-8989-9662}
\affiliation{%
  \institution{OceanBase, Ant Group}
  \city{Hangzhou}
  \country{China}}
\email{xuquanqing.xqq@oceanbase.com}  

\author{Chuanhui Yang}
\orcid{0009-0009-3530-6476}
\affiliation{%
  \institution{OceanBase, Ant Group}
  \city{Hangzhou}
  \country{China}}
\email{rizhao.ych@oceanbase.com}

\author{Zhan Qin}
\orcid{0000-0001-7872-6969}
\affiliation{%
  \institution{Zhejiang University}
  \city{Hangzhou}
  \country{China}}
\email{qinzhan@zju.edu.cn}

\author{Kui Ren}
\orcid{0000-0002-1969-2591}
\affiliation{%
  \institution{Zhejiang University}
  \city{Hangzhou}
  \country{China}}
\email{kuiren@zju.edu.cn}

\author{Wei Dong}
\authornote{Wei Dong is the corresponding author.}
\orcid{0000-0002-0394-4125}
\affiliation{%
  \institution{Nanyang Technological University}
  \city{Singapore}
  \country{Singapore}}
\email{wei_dong@ntu.edu.sg}

\renewcommand{\shortauthors}{Siyi Wang et al.}

\begin{abstract}
  Differential Privacy (DP) has become the gold standard for protecting individual privacy in data analytics, and the shuffle-DP model has attracted significant attention from both academia and industry due to its favorable balance between privacy and utility. However, existing shuffle-DP protocols rely on a strong assumption: all users behave honestly. In real-world scenarios, adversarial users can exploit this vulnerability through poisoning attacks, compromising both privacy guarantees and the utility of analytical results.
While defending against poisoning attacks in the shuffle-DP model has recently gained interest, existing solutions are limited to frequency estimation tasks. To address this issue, we propose the first general defense framework for all union-preserving queries, capable of transforming any shuffle-DP protocol into a version resilient to poisoning attacks.
Beyond robust defense against poisoning attacks, our framework achieves high utility of analytical results. Compared to the original shuffle-DP protocol, it retains asymptotically equivalent error in attack-free settings and incurs only a polylogarithmic increase in error when a constant number of attackers are present.
We demonstrate the generality of our framework on several common queries, including summation, frequency estimation, and range counting. Experimental results confirm that our approach effectively defends against poisoning attacks while maintaining strong utility and communication efficiency.
\end{abstract}

\begin{CCSXML}
<ccs2012>
   <concept>
       <concept_id>10002978.10003018</concept_id>
       <concept_desc>Security and privacy~Database and storage security</concept_desc>
       <concept_significance>500</concept_significance>
       </concept>
   <concept>
       <concept_id>10002951.10002952</concept_id>
       <concept_desc>Information systems~Data management systems</concept_desc>
       <concept_significance>500</concept_significance>
       </concept>
 </ccs2012>
\end{CCSXML}

\ccsdesc[500]{Security and privacy~Database and storage security}
\ccsdesc[500]{Information systems~Data management systems}


\keywords{Differential privacy, shuffle model, poisoning attacks}

\received{July 2025}
\received[revised]{October 2025}
\received[accepted]{November 2025}

\maketitle

\section{Introduction}

A key challenge in data analytics is to obtain meaningful insights while preserving privacy. \textit{Differential privacy} (DP)~\cite{dwork2006calibrating}, the gold standard for private data analysis, achieves this by adding noise to query results, preventing inference of individual data.
Traditional DP model assumes a trusted curator who collects data and privatizes results before release, known as \textit{central-DP}. However, real-world scenarios often lack a universally trusted curator, leading to the adoption of \textit{local-DP} \cite{kasiviswanathan2011can}, where users privatize data locally before sharing it. 
Although local-DP offers stronger privacy guarantees, it incurs higher errors. For instance, in bit counting, central-DP achieves an error of $O(1/\varepsilon)$ by adding Laplace noise of scale $1/\varepsilon$ \cite{dwork2006calibrating}, where $\varepsilon$ is a privacy budget to control the privacy loss (see Section~\ref{sec:Preliminary} for more details).
In contrast, the optimal error under local-DP has been proven to be $O(\sqrt{n}/\varepsilon)$~\cite{beimel2008distributed,chan2012optimal}.

To balance privacy and utility, shuffle-DP~\cite{bittau2017prochlo,cheu2019distributed,borja2019the,erlingsson2019amplification} was introduced, leveraging a trusted shuffler to anonymize user messages before analysis.
The process consists of three steps: (1) Each user privatizes its own data using a randomizer $\mathcal{R}(\cdot)$. (2) A trusted shuffler $\mathcal{S}$ collects all users' responses, randomly permutes them, and passes them to a potentially untrusted analyzer $\mathcal{A}$. (3) The analyzer conducts further analysis.
By introducing additional randomness through shuffling, shuffle-DP reduces the required noise per user, thus improving utility. In bit counting, each user adds only $O(1/n\varepsilon)$ extra noise~\cite{ghazi2021differentially}, which leads to a total error of $O(1/\varepsilon)$ and matches the central-DP error. However, unlike central-DP which places all trust in the analyzer, shuffle-DP only requires a trusted shuffler, which can be easily implemented via anonymous communication channels (e.g., mix networks~\cite{chaum1982untraceable, danezis2003Mixminion}, onion routing~\cite{reed1998onion, roger2004tor}) or trusted nodes/hardware~\cite{reiter1998crowds, bittau2017prochlo}).
Due to its strong privacy-utility trade-off, shuffle-DP has been widely studied in fundamental problems such as 
sum aggregation~\cite{ghazi2021differentially, balle2020privatesum, ghazi2024pure}, frequency estimation~\cite{luo2022frequency, ghazi2020private, badih2021on, cheu2022differentially}, distinct counting~\cite{chen21distinct}, as well as in advanced problems such as triangle counting~\cite{imola2022differentially}.

\subsection{Poisoning Attack}
\label{sec:intro_poison}
However, traditional shuffle-DP protocols rely on a strong, often implicit trust assumption: while the analyzer may be adversarial, all users are honest and faithfully follow the protocol.
This setting does not always align with real-world scenarios. In practice, some users themselves might act maliciously, aiming to disrupt the shuffle-DP protocol. 
The adversarial behavior is known as \textit{poisoning attack}~\cite{barreno2006can, cheu2021manipulation}. 
More precisely, a poisoning attacker corrupts some users in the protocol and acts maliciously with two primary objectives: (1) to break privacy, and (2) to destroy the utility of analytical results. 

For the privacy breaking, since the final result aggregates contributions from all users, the attacker can compromise the privacy mechanism by circumventing the required noise. For example, if half of the users are corrupted and withhold noise, the effective privacy protection is reduced by half.
Mitigating such privacy issues has been extensively studied, known as \textit{robust shuffle-DP} \cite{balcer2021connecting}. Existing approaches~\cite{balcer2021connecting} suggest that honest participants introduce additional noise to counteract the noise omitted by corrupted users. Specifically, in the above example, each honest user adjusts the noise scale by doubling it to account for the lack of noise contributed by the remaining half, the corrupted users.

For the utility destroying, since shuffle-DP allows each user to send additional noise, attackers can manipulate the aggregation results by sending an excessive number of messages. In the bit counting protocol~\cite{ghazi2021differentially}, each attacker can inject up to $O(n)$ noisy bits without being detected by the analyzer in the worst case,
significantly compromising the utility of the final result (see Section~\ref{sec:Threat Model} for more details).  
Astute readers may find that such an attack is a unique challenge in shuffle-DP.  
Under local-DP, each message is linked to a real user identity, allowing the analyzer to verify whether messages are reasonable. However, under shuffle-DP, since messages are anonymous, attackers can impersonate honest users and send excessive noise messages, but the analyzer cannot distinguish these poisoned messages from honest ones.

\begin{table}[t]
    \centering
    % \small
    \resizebox{\linewidth}{!}{
        \begin{tabular}{c|c||c|c||c|c|c|c}
            \toprule
            \multirow{2}{*}{Problem} & \multirow{2}{*}{Protocol} & \multirow{2}{*}{Detection} & \multirow{2}{*}{Recovery} & \multirow{2}{*}{\#Messages/user} & \multirow{2}{*}{\#Bits/message} & \multicolumn{2}{c}{Error} \\
            & & & & & & w/o attacker & w/ attacker \\
            \midrule \midrule
            \multirow{3}{*}{Bit counting} & GKMPS \cite{ghazi2021differentially} & $\times$ & $\times$ &  $1+O(\log n / n)$ & 2 & $O(1)$ & $\infty$ \\
            & CSUZZ \cite{cheu2019distributed} & \checkmark & $\times$ & 1 & 1 & $O(\log n)$ & $\infty$\\
            & Ours + GKMPS & \checkmark & \checkmark & $O(\log n)$ & $O(\log n)$ & $O(1)$ & $O(\log^{3} n)$ \\
            \midrule
        
            \multirow{5}{*}{Summation} & GKMPS \cite{ghazi2021differentially} & $\times$ & $\times$ & $1+O(\log^3 n \cdot U/n)$ & $O(\log n)$ & $O(U)$ & $\infty$ \\
             & BBGN(IKOS) \cite{balle2020privatesum} & $\times$ & $\times$ & $O(1)$ & $O(\log n)$ & $O(U)$ & $\infty$ \\ 
             & BBGN(recursive) \cite{balle2020privatesum}  & \checkmark & $\times$ & $O(\log \log n)$ & $O(\log n)$ & $O(U\sqrt{\log n} \cdot \log \log n )$ & $\infty$ \\
             & Ours + GKMPS & \checkmark & \checkmark & $O(\log^2 n \cdot U)$ & $O(\log n)$ & $O(U)$ & $O(U \log^{3} n) $  \\ 
             & Ours + BBGN(IKOS) & \checkmark & \checkmark & $O(\log n \cdot \log \log n)$ & $O(\log n)$ & $O(U)$ & $O(U \log^{3} n) $ \\ 
            \midrule
        
            \multirow{5}{*}{\shortstack{Frequency \\ estimation}} & GKMPS \cite{ghazi2020private, ghazi2021differentially} & $\times$ & $\times$ & $1+O(\log n \cdot U/n)$ & $O(\log n)$ & $O(\log n)$ & $\infty$ \\
            & GGKPV \cite{badih2021on} & $\times$ & $\times$ & $O(\log n)$ & $O(\log^2 n)$ & $O(\log n)$ & $\infty$ \\
             & LWY \cite{luo2022frequency} & $\times$ & $\times$ & $O(1)$ & $O(\log n)$ & $O(\log n)$ & $\infty$ \\
             & CZ \cite{cheu2022differentially} & \checkmark & \checkmark & 2 & $U$ & $O(\log n)$ & $O(\log n)$ \\ 
             & Ours + LWY & \checkmark & \checkmark & $O(\log n)$ & $O(\log n)$ & $O(\log n)$ & $O(\log^{3} n)$ \\
            \bottomrule
        \end{tabular}
    }
    \caption{Comparison between our results and prior works for bit counting, summation, and frequency estimation under shuffle-DP, assuming $\varepsilon = \beta = \Theta(1)$ and $\log n = \Theta(\log(1/\delta)) = \Theta(\log U)$. We use $\ell_\infty$ error for frequency estimation.}
    \label{tab:contribution}
\end{table}

Such an issue has gained significant attention recently.
Some protocols~\cite{balle2020privatesum,cheu2022differentially, badih2021on} control the number of messages each user can send, enabling the detection of corrupted users who exceed the expected count.
For example, \cite{balle2020privatesum} studies the sum aggregation problem and limits the number of messages from a single user to $c = O(\log \log n)$.
This allows the analyzer to identify poisoning attacks by counting the total number of received messages (i.e., equal to $c\cdot n$) and verifying whether each message falls in a reasonable range. \underline{However, this approach only enables attack detection, not recovery:} the analyzer can recognize the happening of poisoning attacks, but cannot extract any meaningful analytical results from the poisoned messages. Consequently, corrupted users still succeed in their objective of destroying utility.
Balle et al.~\cite{cheu2022differentially} propose a shuffle-DP protocol for frequency estimation that incorporates a blind signature scheme to restrict messages during communication. This approach enables the recovery of meaningful results by filtering out messages with invalid signatures. 
However, the protocol is limited to scenarios where each user sends a fixed number of messages, and each message is drawn from a well-bounded domain. Consequently, it cannot be generalized to most existing shuffle-DP protocols, as they do not adhere to these constraints.

At first glance, there are several  straightforward solutions to these questions: 
(i) Trusted shuffler with authentication: one could use the trusted shuffler to perform standard authentication and message‐integrity checks, thereby preventing corrupted users from injecting arbitrarily many malicious messages.
(ii) Local-DP protocols followed by a shuffler: another idea is to have each user first apply a local-DP mechanism to their data and then send the perturbed result to the shuffler, thus benefiting from privacy amplification.
{(iii) Two-round strategy: one may identify corrupted users based on malicious behaviors in an initial run, and then re‐execute the shuffle‐DP protocol with only the remaining honest users.}
{However, we can easily show all of these approaches to fail in practice:
the first requires the shuffler to support authentication, which is unrealistic to implement and offers limited utility even if feasible.
The second is only effective under a very strict privacy regime, namely when $\varepsilon = O(1/\sqrt{n})$, and cannot recover meaningful results once corrupted users inject large amounts of noise.}
{The third fails against adaptive corrupted users who behave honestly in the first round and attack in the second, leading to the same error as the shuffle-DP protocol without any defense.}
{A more detailed discussion of these limitations is provided in Section~\ref{sec:baseline_limit}.
}

\textbf{Problem Statement.~}
Therefore, the question comes:
\textit{Does there exist a general framework to defend against poisoning attacks in shuffle-DP, ensuring that the defense mechanism not only detects the attacks but also recovers meaningful analytical results from the poisoned messages?}

\subsection{Our Results}
In this work, we answer the above question affirmatively. We present a general framework for arbitrary union-preserving queries. It can transform any shuffle-DP protocol vulnerable to poisoning attacks into one that can defend against such attacks.
Our approach develops a hierarchical defense framework based on the given shuffle-DP protocol. Specifically, the framework organizes the $n$ users into a binary tree structure, where each leaf node corresponds to a single user, and each non-leaf node represents the union of users from its child nodes. Users within each node independently and in parallel execute the shuffle-DP protocol once.
To detect utility-destroying attacks, the analyzer compares the output of each node with the aggregated results of its child nodes. If the deviation exceeds a predefined threshold, the node’s output is flagged as poisoned. The final query result is recovered by recursively replacing poisoned results with aggregated outputs from valid child nodes' results.
To defend against privacy-breaking attacks, honest users adjust their noise scale based on the estimated fraction of corrupted users in the population, thereby compensating for the noise lost due to malicious behaviors.
With this idea, we can extend existing shuffle-DP protocols to defend against poisoning attacks.  

More precisely, our contributions are shown below:
\begin{itemize}
    \item We first consider the single-attacker setting and propose a general framework that extends any shuffle-DP protocol for union-preserving queries to defend against poisoning attacks. In the non-adversarial case (i.e., when no poisoning attacker is present), the framework preserves asymptotically equivalent error to the original shuffle-DP protocol. When a single attacker is present, the error increases by at most a polylogarithmic factor. Regarding communication cost, for most queries, it also increases by at most a polylogarithmic factor. 
    As a result, we propose the first shuffle-DP mechanisms that defend against poisoning attacks for a range of tasks, including bit counting and summation
    and our results
     for most fundamental queries are summarized in Table~\ref{tab:contribution}.
    \item We further extend our framework to support the multi-attacker setting. Compared to the single-attacker case, the error increases by an additional factor of $k$, where $k$ is the number of attackers, and the communication cost remains unchanged. Additionally, we propose one strategy to further optimize communication efficiency.
    \item We conduct comprehensive empirical evaluations\footnote{The code is available at \url{https://github.com/Whelsea/DefenseShuffleDP}} across three fundamental query tasks: bit counting, summation, and frequency estimation, using both synthetic and real-world datasets. 
    Our experiments compare the proposed framework with state-of-the-art (SOTA) protocols under both non-adversarial and adversarial settings, evaluating both utility and efficiency. 
    The empirical results consistently align with our theoretical analysis.
\end{itemize}

\subsection{Related Work}
\label{sec:Related Work}
{
DP has made a significant impact in the data management community~\cite{dong2022r2t,zhang2017privbayes, kifer2011no}. Early studies primarily focus on the central-DP model~\cite{dwork2006calibrating, hay2010hierarchical, mcsherry2009privacy, qardaji2013understanding, hay2016principled}, where a trusted curator has access to the raw data. To address trust concerns in distributed environments, local-DP was introduced and widely explored~\cite{cormode2018marginal, li2024local, zhang2025federated, he2024common,kulkarni2019answering, li2020estimating, ren2022ldp, he2025robust, yu2025privrm}. However, local-DP often suffers from an $O(\sqrt{n})$ utility loss compared to central-DP~\cite{chan2012optimal}. To strike a balance between these two extremes, shuffle-DP was introduced, which works in a distributed setting and matches central-DP in queries like count and sum~\cite{10.14778/3424573.3424576,rm2}.
}

The anonymous communication was first introduced in \cite{ishai06cryptography}, which solves the secure summation problem in $O(\log n)$ messages. Combining this idea with differential privacy \cite{bittau2017prochlo} has led to the development of the shuffle-DP model.
Since it reaches the sweet spot between privacy and utility, shuffle-DP has been widely studied in statistical analysis.
\cite{balle2020privatesum} improved the message complexity of shuffle-DP summation protocol to $O(1)$ and \cite{ghazi2021differentially} further improved to $1+o(1)$, both achieving central-DP error.
\cite{ghazi2020private, cheu2022differentially, luo2022frequency, badih2021on} developed shuffle-DP protocols for the frequency estimation problem that also achieve central-DP error.
% , where \cite{luo2022frequency} achieved the optimal $1+o(1)$ communication costs.
\cite{badih2021on, rm2} extended the frequency estimation protocol to range counting queries.
% and \cite{rm2} reduced the communication cost using real messages reduction mechanism.
All these fundamental problems in shuffle-DP can reach central-DP error except distinct counting, where \cite{chen21distinct} proved and reached the lower bounds of both shuffle-DP (i.e., $O(\sqrt{n})$) and local-DP (i.e., $O(n)$).

Poisoning attacks have also been widely studied in different areas. \cite{cheu2021manipulation, cao2021data, li2023fine, wu2022poisoning, tong2024data} studies the attacks under local-DP and \cite{barreno2006can, battista2012poisoning, fang2020local, tolpegin2020data} studies them under machine learning. 
In shuffle-DP, despite the concept of robust shuffle-DP \cite{balcer2021connecting}, \cite{cheu2022differentially} studies the same behaviors in the frequency estimation problem with strict constraints.
These types of behaviors are also studied in \textit{secure multi-party computation} (MPC), known as malicious security. Unlike semi-honest adversaries who follow the protocol honestly but try to learn as much information from the transcript, malicious adversaries aim to destroy the protocol in different aspects. There are several general frameworks to defend against them using cryptographic methods, including zero-knowledge proof \cite{goldreich1987how} and authenticated secret-sharing/garbling \cite{wang2017malicious, ivan2012malicious, bendlin2011malicious}. However, these methods rely on end-to-end communication and require multiple rounds, which are not easily extended to the shuffle-DP model which has only anonymous communication channels and aims for single-round communication.

\section{Preliminary}
\label{sec:Preliminary}

%\wei{Add some subscriptions for Q for different functions. You can reuse it.}

\subsection{Notation}
\label{sec:notation} 
Let multiset $D = \{x_1,x_2,\dots,x_n\}$ be the input dataset, where each $x_i\in \mathcal{X}$ can be either a single element or a vector.
Given a vector $x$, we use $||x||_{\ell_p}$ to denote its $\ell_p$ norm and define its $\ell_p$ distance to a set of vectors $S$ as:
\[\mathrm{dis}_{\ell_p}(x,S) = \min\{\|x-y\|_{\ell_p}:y\in S\}.\]

The query $Q:\mathcal{X}^{n}\rightarrow \mathcal{Y}$ is a \textit{union-preserving query} iff for any $D_1\in \mathcal{X}^{n_1}$ and $D_2\in \mathcal{X}^{n_2}$, we have $Q(D_1\uplus D_2) = Q(D_1) + Q(D_2)$. Some common union-preserving queries are listed below.

\begin{enumerate}
    \item Bit counting: $Q_{\text{count}}(D) = \sum_{i=1}^n x_i$ where $\mathcal{X} = \{0, 1\}$;
    \item Summation: $Q_{\text{sum}}(D) = \sum_{i=1}^n x_i$ where $\mathcal{X} = \{0, 1, \cdots, U\}$;
    \item Frequency estimation: $Q_{\text{hist}}(D) = ( \sum_{i=1}^n \mathbb{I}{[x_i=j]} )_{j=0}^{U}$; % where $\mathcal{X} = \{0, 1, \cdots, U\}$;
    \item Range counting: $Q_{\text{range}}(D) = ( \sum_{i=1}^n \mathbb{I}{[l \leq x_i \leq r]} )_{0 \leq l \leq r < U}$. %where $\mathcal{X} = \{0, 1, \cdots, U\}$;
    % \item Max: $Q_{\text{max}}(D) = \max_{i=1}^n \{x_i\}$ where $\mathcal{X} = \{0, 1, \cdots, U-1\}$;
\end{enumerate}
Unless otherwise specified, all queries considered in this paper are assumed to be union-preserving.

Let $\mathrm{Range}(Q,n) = \{Q(D),D\in\mathcal{X}^n\}$ be the output set of a query $Q$ with all possible inputs and $\mathrm{\gamma}_{\ell_p}(Q,n)$ be the diameter of $\mathrm{Range}(Q,n)$ in $\ell_p$ metric:
\[\mathrm{\gamma}_{\ell_p}(Q,n) =\max \big\{\|Q(D)-Q(D')\|_{\ell_p},D,D'\in \mathcal{X}^n\big\}.\]

For clarity, key notations used throughout this paper are summarized in Table~\ref{tab:notation}.
\begin{table}[t]
\small
    \centering
    {\begin{tabular}{c|c}
    \toprule
      Notation   & Meaning \\
      \midrule
     % $\mathcal{X}$ & Data domain \\
      $x_i$ & Data of user $i$ \\
      $D$ & Input dataset \\
      $n$ & \#Users (equivalently, data size) \\
      $k$ & True \#corrupted users \\
      $\hat{k}$ & Public upper bound of $k$ \\
      $U$ & Domain size \\
      $Q$ & Union-preserving query \\
      $Q(D)$ & True value of $Q$ on $D$ \\
      $\mathrm{Range}(Q,n)$ & Output set of $Q$ with $n$ users \\
      $\mathrm{dis}_{\ell_p}(x, S)$ & $\ell_p$-distance from vector $x$ to set $S$ \\
      $\gamma_{\ell_p}(Q,n)$ & $\ell_p$-diameter of $\mathrm{Range}(Q,n)$ \\
      $\mathcal{R}, \mathcal{S}, \mathcal{A}$ & Randomizer, shuffler, analyzer \\
      $\mathcal{P}_Q$ = ($\mathcal{R}, \mathcal{S}, \mathcal{A}$) & Shuffle-DP protocol for $Q$ \\
      $\mathcal{P}_Q(D)$ & Query result of $\mathcal{P}_Q$ on $D$ \\
      $Y^{(r)}_i$ & Messages generated by user $i$ at level $r$ via $\mathcal{R}$ \\
      $\mathcal{S}^{(r)}_g$ & Shuffler assigned to group $g$ at level $r$ \\
      $Z^{(r)}_g$ & Shuffled messages of group $g$ at level $r$ via $\mathcal{S}^{(r)}_g$ \\
      $\varepsilon^{(r)},\delta^{(r)}$ & Privacy parameters at level $r$\\
      $\beta^{(r)}$ & High probablity parameter at level $r$ \\
      $\tilde{Q}^{(r)}_g$ & Query result of group $g$ at level $r$ \\
      $\mathrm{Error}_{\ell_p}(P_Q,\varepsilon,\delta,\beta)$ & $\ell_p$-error guarantee of $\mathcal{P}_Q$ \\
      $\mathrm{Msg}(P_Q,\varepsilon,\delta,n)$ & Expected \#messages per user under $\mathcal{P}_Q$ \\
      $\mathrm{Bit}(P_Q,\varepsilon,\delta,U,n)$ & Expected \#bits per message under $\mathcal{P}_Q$ \\
    \bottomrule
    \end{tabular}}
    \caption{{Notations used in the paper.}}
    \label{tab:notation}
\end{table}

\subsection{Differential Privacy}
\begin{definition}[Differential Privacy]
\label{def: DP}
For some $\varepsilon > 0$ and $0 \leq \delta < n^{-\Omega(1)}$, a randomized mechanism $\mathcal{M} : \mathcal{X}^{n} \to \mathcal{Y}$ is $(\varepsilon, \delta)$-differentially private if for any two neighboring datasets $D \sim D'$ (i.e., $D$ and $D'$ differ by a single element), $\mathcal{M}(D)$ and $\mathcal{M}(D')$ are $(\varepsilon, \delta)$-indistinguishable. That is, for any subset of outputs $Y \subseteq \mathcal{Y}$,
\[
\Pr[\mathcal{M}(D) \in Y ] \leq e^\varepsilon \cdot \Pr[\mathcal{M}(D') \in Y ] + \delta.
\]
\end{definition}
In practice, $\varepsilon$ is usually a constant between $0.1$ and $10$, and $\delta$ should be much smaller than $1/n$.
Below are some commonly used DP properties.

\begin{lemma} [Post Processing~\cite{dwork2014algorithmic}]
\label{lem:Post Processing}
If $\mathcal{M}: \mathcal{X} \to \mathcal{Y}$ satisfies $(\varepsilon, \delta)$-DP and $\mathcal{M'}: \mathcal{Y} \to \mathcal{Z}$ is any randomized mechanism, then $\mathcal{M'}(\mathcal{M}(D))$ satisfies $(\varepsilon, \delta)$-DP.
\end{lemma}

\begin{lemma} [Basic Composition~\cite{dwork2014algorithmic}]
\label{lem:Sequential Composition}
If $ \mathcal{M} $ is a (possibly adaptive) composition of differentially private mechanisms $ \mathcal{M}_1,\mathcal{M}_2,\dots, \mathcal{M}_k$, where each $ \mathcal{M}_i$ satisfies $(\varepsilon, \delta)$-DP, then $ \mathcal{M} $ satisfies $(k\varepsilon, k\delta)$-DP.
    
% \begin{itemize}
%     \item $\varepsilon' = k\varepsilon$ and $\delta' = k\delta$ \hfill [Basic Composition]
%     \item $\varepsilon' = \varepsilon \cdot \sqrt{2k \log \frac{1}{\delta''} + k\varepsilon(e^\varepsilon - 1)}$ and $\delta' = k\delta + \delta''$ for any $\delta'' > 0$ \hfill [Advanced Composition]

% \end{itemize}

\end{lemma}

\begin{lemma} [Parallel Composition~\cite{mcsherry2009privacy}]
\label{lem:Parallel Composition}
    Let $ \mathcal{X}_1, \dots, \mathcal{X}_k $ be pairwise disjoint subdomains of $ \mathcal{X} $, and each $ \mathcal{M}_i : \mathcal{X}^n_i \to \mathcal{Y} $ be an $(\varepsilon, \delta)$-DP mechanism. Then the mechanism $ \mathcal{M}(D) := (\mathcal{M}_1\left(D \cap \mathcal{X}_1 \right)$, $\dots$, $\mathcal{M}_k \left(D \cap \mathcal{X}_k \right))$ satisfies $(\varepsilon, \delta)$-DP.
\end{lemma}

\subsection{Shuffle-DP}
\label{sec:Shuffle-DP}

Different DP models can be captured by Definition~\ref{def: DP} by the formulation of $\mathcal{M}(D)$.
In central-DP, $\mathcal{M}(D)$ represents the output of a trusted curator. In contrast, for local-DP and shuffle-DP, where the analyzer $\mathcal{A}$ is untrusted, $\mathcal{M}(D)$ is defined by the view of $\mathcal{A}$.
Under local-DP, each user $i$ independently privatizes its own data $x_i$ using a local randomizer $\mathcal{R}$ and then sends the randomized response $\mathcal{R}(x_i)$ to $\mathcal{A}$ hence $\mathcal{M}(D)=(\mathcal{R}(x_1), \mathcal{R}(x_2), \dots, \mathcal{R}(x_n))$.
Under shuffle-DP, a trusted shuffler $\mathcal{S}$ randomly permutes the responses before forwarding them to analyzer $\mathcal{A}$, resulting in
\[\mathcal{M}(D) = \mathcal{S} \circ \mathcal{R} (D)= \{\mathcal{R}(x_1), \mathcal{R}(x_2), \dots, \mathcal{R}(x_n)\}.\]

Under shuffle-DP, a protocol which is used to answer query $Q$ is defined as $\mathcal{P}_Q = (\mathcal{R}, \mathcal{S}, \mathcal{A})$, where $\mathcal{R}, \mathcal{S}, \mathcal{A}$ denote the randomizer, shuffler, and analyzer, respectively. The protocol output is given by
\[\mathcal{P}_Q(D)\gets \mathcal{A} \circ \mathcal{S} \circ \mathcal{R} (D) = \mathcal{A}(\mathcal{S}(\mathcal{R}(x_1), \cdots, \mathcal{R}(x_n))).\]
%$\varepsilon$ and $\delta$ are privacy budget and 
$\mathcal{S} \circ \mathcal{R} (D)$ satisfies $(\varepsilon,\delta)$-DP. 
%The parameter $\beta$ controls the error level of the protocol.

\paragraph{Multiple shuffler setting.} 
%In more complex scenarios, multiple shufflers are introduced. 
Some shuffle-DP protocols consider the setting where there exist multiple shufflers $\mathcal{S}$~\cite{balle2020privatesum}.
Here, we have a set of shufflers $\{\mathcal{S}_1, \mathcal{S}_2, \ldots, \mathcal{S}_m\}$ and $n$ users. Each user $i$ can send messages to one or more shufflers $\mathcal{S}$.

\paragraph{Shuffle-DP protocol parameters.} By convention, a shuffle-DP protocol has four parameters: the privacy budgets $\varepsilon$ and $\delta$, the error parameter $\beta$, and the data size $n$. The parameter $\beta$ specifically controls the probability of exceeding the stated error bound. 
More precisely, we use $\err_{\ell_p}(\mathcal{P}_Q, \varepsilon, \delta, \beta)$ to denote the $\ell_p$-norm error of the protocol $\mathcal{P}_Q$. Specifically, for any input dataset $D$, with privacy budgets $\varepsilon, \delta$, with probability at least $1 - \beta$, we have
\[
\left\| \mathcal{P}_Q(D) - Q(D) \right\|_{\ell_p} \leq \err_{\ell_p}(\mathcal{P}_Q, \varepsilon, \delta, \beta).
\]
Notably, the error of existing shuffle-DP protocols for union-preserving queries does not depend on the data size $n$. The exception happens in non-union-preserving queries, such as distinct count~\cite{chen21distinct}, which are beyond the scope of this work.
We quantify the communication cost in two parts. First, $\mathrm{Msg}(\mathcal{P}_Q\allowbreak,\allowbreak \varepsilon\allowbreak,\allowbreak \delta\allowbreak,\allowbreak n)$
denotes the expected number of messages each user sends in protocol $\mathcal{P}$ with a group of $n$ users, i.e., $\mathbb{E}[|\mathcal{R}(x_i)|]$.  
Second, $\mathrm{Bit}(\mathcal{P}_Q,\varepsilon, \delta, U,n)$ denotes the number of bits required to encode each message.  
For all existing shuffle-DP protocols, it is observed that $\err_{\ell_p}(\mathcal{P}_Q, \allowbreak \varepsilon \allowbreak, \delta, \beta)$, $\mathrm{Msg}(\mathcal{P}_Q, \varepsilon, \delta, n)$, and $\mathrm{Bit}(\mathcal{P}_Q\allowbreak,\allowbreak \varepsilon\allowbreak,\allowbreak \delta\allowbreak,\allowbreak U\allowbreak,\allowbreak n)$ are all polynomial functions of their respective parameters. This observation will simplify the analysis in this paper. For example, $\err_{\ell_p}(\mathcal{P}_Q, \varepsilon/c, \delta/c, \beta/c) = O\big(\err_{\ell_p}(\mathcal{P}_Q, \varepsilon, \delta, \beta)\big)$ for any constant $c$.

\begin{example}[Bit Counting Problem $Q_\text{count}$]
\label{exm:Bit counting problem}
We take the bit counting problem ($Q_\text{count}$) as an example.
$\mathrm{Range}(Q_\text{count}\allowbreak,\allowbreak n) = \{\allowbreak0,\allowbreak1,\allowbreak\dots,\allowbreak n\}$, and the diameter $\gamma_{\ell_1}(Q_\text{count},n) = n$. 
{The SOTA solution~\cite{ghazi2021differentially} lets each user $i$ send a single ``$+1$'' bit if $x_i = 1$. In addition, to preserve DP, each user samples a \textbf{random} but bounded number of noisy ``$+1$'' and ``$-1$'' messages. The protocol \cite{ghazi2021differentially} has the following properties:}

\begin{lemma}[Bit Counting Protocol $\mathcal{P}_{Q_{\text{count}}}$~\cite{ghazi2021differentially}]
\label{lem: Correlated Noise for bit counting}
Given $\varepsilon >0, \delta>0$ and $n$, \cite{ghazi2021differentially} solves the bit counting problem under $(\varepsilon,\delta)$-shuffle-DP with the following guarantees:
    \begin{itemize}
        \item Utility guarantee: $\err_{\ell_1}(\mathcal{P}_{Q_{\text{count}}}, \varepsilon, \delta, \beta)=O\left(\tfrac{1}{\varepsilon}\log\tfrac{1}{\beta}\right)$;
        \item Communication cost: $\mathrm{Msg}(\mathcal{P}_{Q_{\text{count}}}, \varepsilon, \delta, n)=1+{O}\left(\tfrac{\log(1/\delta)}{\varepsilon n}\right)$ and $\mathrm{Bit}(\mathcal{P}_{Q_{\text{count}}},\varepsilon, \delta,U,n) = 2$.
        % In expectation, each user sends $1+{O}\left(\tfrac{\log(1/\delta)}{\varepsilon n}\right)$ messages, with each message containing 2 bits.
    \end{itemize}
\end{lemma}
\end{example}

\subsection{Poisoning Attacks}
\label{sec:Threat Model}
Under the shuffle-DP model, we identify two types of adversaries: 

\begin{enumerate}
    \item \textit{Semi-honest analyzer}: it follows the protocol but attempts to infer private information from shuffled messages.
    \item \textit{Poisoning attackers}: they, a.k.a. corrupted users, actively manipulate their inputs to disrupt the protocol.
\end{enumerate}

Any protocol satisfying shuffle-DP effectively defends against semi-honest analyzers. However, poisoning attackers pose additional challenges that require further robustness measures.
Below, we detail their objectives and behaviors.

\subsubsection{Capacity}

The number of corrupted users, denoted by $k$, is bounded by a public parameter $\hat{k}$. While the exact value of $k$ is unknown, we assume that $\hat{k}$ is known to all parties. 
In MPC, $\hat{k}$ is typically set to a large value such as $n/2 - 1$, referred to as the ``honest-majority'' model, where the adversary can corrupt strictly fewer than half of the parties.
However, in the shuffle-DP setting, $\hat{k}$ must be significantly smaller.
{Consider the bit counting task: if $k$ corrupted users alter their inputs from $0$ to $1$, a shift between valid inputs, the final result may shift by as much as $k$. Since these users follow the protocol correctly, it is impossible for the analyzer to detect their manipulation. As a result, if $k$ is too large, any defense becomes ineffective.
Therefore, to ensure meaningful results, we set $\hat{k}$ to be polylogarithmic in $n$.}

\subsubsection{Objectives and behaviors}
\label{sec: Attacker behaviors}

\paragraph{Privacy attack: Exempt from generating randomness.}
Recall that shuffle-DP relies on collective randomness generation across all $n$ users. Therefore, a privacy-breaking attack occurs when poisoning attackers drop out in generating randomness. 
For example, in the case of bit counting, if $k$ users drop out, the resulting privacy protection is degraded compared to the intended level. 
Robust shuffle-DP protocols~\cite{balcer2021connecting} mitigate this by having honest users amplify their noise by $n/(n-\hat{k})$, effectively safeguarding privacy against malicious behavior.

\paragraph{Utility attack 1: Alter input data.}
As mentioned before, the most straightforward utility attack is to alter the input data.
{
In general, a single corrupted user can introduce at most $\gamma_{\ell_p}(Q,1)$ error to the final result through this attack.
Consider the bit counting query as an example. As described in Example~\ref{exm:Bit counting problem}, existing shuffle-DP protocols such as~\cite{ghazi2020private, ghazi2021differentially} only allow users to send bits with value ``$+1$'' or ``$-1$''. Hence, any message outside this domain can be easily detected and filtered by the analyzer. As a result, a fake input can affect the result by at most $O(1)$. Similarly, for sum estimation, the impact of a fake input is bounded by $O(U)$, the size of the domain.}
Such an attack cannot be fundamentally defended, as corrupted users still provide syntactically valid inputs that conform to the protocol. 
Fortunately, since $\hat{k}$ is at most a polylogarithmic function of $n$, the overall impact of such malicious actions on the final result remains limited. Above all, this subtle attack is not our focus.

\paragraph{Utility attack 2: Send excessive messages.}
A more overt and consequential attack is to send excessive messages, potentially leading to unbounded error. 
Still using the bit counting query as an example, in protocols~\cite{ghazi2020private, ghazi2021differentially}, the ``$+1$'' and ``$-1$'' messages sent by users consist of two components: (i) the actual messages from users holding bit~``$1$'', and (ii) the DP noise composed of both ``$+1$'' and ``$-1$'' messages added to ensure privacy. Due to the shuffler’s anonymity, the analyzer cannot distinguish between these two components. As a result, a corrupted user can introduce $O(n)$ error by injecting $O(n)$ noisy ``$+1$'' messages.
This attack is particularly undetectable when the number of true ``+1'' messages is small, as the plausible range of query results is already $O(n)$. 
Several studies~\cite{cheu2022differentially} have identified and tackled this ``flooding'' attack by fixing both the number of messages each user can send and the output domain of each message. However, these methods exhibit a critical limitation: they are often tailored to specialized tasks, lacking a unified framework applicable to more general queries.

\subsubsection{Limitation of naive baselines}
\label{sec:baseline_limit}

{In Section~\ref{sec:intro_poison}, we mentioned several straightforward defense strategies against poisoning attacks. Here we show why these solutions do not work in practice.}

{\paragraph{Trusted shuffler with authentication and message-integrity checks.} One might consider performing authentication and integrity checks at the shuffler to prevent malicious users from injecting unlimited messages. However, this is incompatible with the standard shuffle-DP model, where the shuffler only permutes messages; granting it additional functions changes the trust model and brings the model closer to central-DP. In practice, shufflers are often implemented via anonymous communication channels (e.g., mix networks or onion routing), making such authentication unrealistic. 
More importantly, even if additional integrity checks were available, they would only be effective in limited scenarios—specifically, when each user sends a fixed number of messages and each message originates from a bounded domain. As discussed in the introduction, these assumptions do not always hold.
For example in the BBGN(IKOS) protocol~\cite{balle2020privatesum} where each message is drawn from a range much larger than $nU$, a corrupted user can submit extremely large but valid values that pass verification yet introduce $O(nU)$ error. Hence, authentication and integrity checks may work for a narrow class of protocols with fixed message counts and well-bounded domains, but they cannot provide a general defense framework.}

{\paragraph{Local-DP + shuffler.} 
First, the local-DP model inherently resists poisoning attacks, as each user independently perturbs their data and integrity can be verified at the user side. 
However, it suffers from an $O(\sqrt{n})$ utility loss compared with central-DP and shuffle-DP~\cite{borja2019the}. 
Shuffle-DP alleviates this by introducing \textit{privacy amplification by shuffling}, which reduces the effective privacy budget $\varepsilon$ to $O(\varepsilon_0 / \sqrt{n})$, where $\varepsilon_0$ is the original local privacy budget and $\varepsilon$ is the resulting privacy guarantee after shuffling~\cite{borja2019the,erlingsson2019amplification}. 
Therefore, a common idea is to apply a local-DP mechanism followed by random shuffling. 
However, the privacy amplification from shuffling is significant only under a stringent privacy regime. 
To achieve $\varepsilon = O(1/\sqrt{n})$, the local mechanism must operate with a constant $\varepsilon_0 = \Theta(1)$. 
In contrast, for the typical setting of $\varepsilon = O(1)$ used in most DP works, the required local privacy budget scales as $\varepsilon_0 = \Theta(\log n)$, resulting in the same error as the original local-DP protocol. 
This fundamental limitation explains why shuffle-DP has emerged as an independent and active research area, rather than as a simple combination of local-DP and shuffling. 
Moreover, because the shuffler hides individual contributions, this approach can at best support detection but cannot recover meaningful results once corrupted users inject large noise, as any noise added before shuffling becomes indistinguishably mixed afterward.
}

{\paragraph{Two-round strategy.} The two-round strategy, in which we first detect and exclude corrupted users and then re-run the query, fails against adaptively corrupted users who behave honestly in the first round to avoid detection and launch their attack in the second round.
For example, in a bit counting query, the corrupted users can follow the protocol in the first round by sending a reasonable number of ``1'' messages, thereby avoiding identification as malicious. However, in the second round, they can inject an excessive number of ``1'' messages to distort the result by $O(n)$, resulting in the same error as if no defense were applied.}

\section{A Straw-man Solution}
\label{sec:strawman}

We first propose a simple \textit{Single-User Shuffle-DP} (SUSDP) protocol as a straw-man solution. Although it increases the error by a factor of $n$ compared to the original shuffle-DP protocol, thereby resulting in poor utility, it serves as a useful starting point and offers key insights that inform the design of our final protocol.

Given a shuffle-DP protocol $\mathcal{P}_Q = (\mathcal{R}, \mathcal{S}, \mathcal{A})$, SUSDP requests each user independently to perturb their own data, thus ensuring it remains unaffected by attackers and avoiding the difficulties of detection under anonymization: 
Each user $i$ privates their data $x_i$ using a local randomizer $\mathcal{R}$, and then sends the messages ${Y}_i$ to its corresponding shuffler $\mathcal{S}_i$. Finally, for each user $i$, the analyzer $\mathcal{A}$ generates an estimated result $\tilde{Q}_i$ and evaluates its reasonability by checking whether 
\[\mathrm{dis}_{\ell_p}(\tilde{Q}_i, \mathrm{Range}(Q, 1)) \leq \err_{\ell_p}(\mathcal{P}_Q, \varepsilon, \delta, \beta/n).\]
If not, the user is flagged as a poisoning attacker.\footnote{{
Note that our detection rule assumes all users are honest and treats outliers as honest as well. The confidence parameter \(\beta\) controls the success rate of this check, ensuring that with probability at least $1-\beta$, an honest outlier will not be misclassified.}}
The final result is obtained by aggregating all valid outputs.
The details of the randomizer and analyzer are presented in Algorithm~\ref{alg:Randomizer of SUSDP} and~\ref{alg:Analyzer of Single-User Shuffle-DP}, respectively.

\begin{small}
\begin{algorithm}[t]
\caption{Randomizer of SUSDP}
\label{alg:Randomizer of SUSDP}
\KwParam{{$\varepsilon, \delta, n$}}
\KwIn{$x_i$, $\mathcal{P}_Q = (\mathcal{R}, \mathcal{S}, \mathcal{A})$}
% ${Y}_i \leftarrow \mathcal{R}(x_i; \varepsilon,\delta,\beta/n,1)$ \;
{${Y}_i \leftarrow \mathcal{R}(x_i; \varepsilon,\delta,1)$ \;}
Send ${Y}_i \text{ to shuffler $\mathcal{S}_i$ }$;
\end{algorithm}
\end{small}

\begin{small}
\begin{algorithm}[t]
\caption{Analyzer of SUSDP}
\label{alg:Analyzer of Single-User Shuffle-DP}
\KwParam{$\varepsilon, \delta, \beta, n$ }
\KwIn{$\{Z_i = Y_i\}_i$, $\mathcal{P}_Q = (\mathcal{R}, \mathcal{S}, \mathcal{A})$}

\For{$i \leftarrow 1$ \KwTo $n$}{
    % $\tilde{Q}_i \leftarrow \sum_{t\in\mathcal{T}_i'}t$;
    $\tilde{Q}_i \leftarrow \mathcal{A}(Z_i; \varepsilon,\delta,\beta/n,1)$; 
    
    \If{$\mathrm{dis}_{\ell_p}(\tilde{Q}_i, \mathrm{Range}(Q, 1)) > \mathrm{Error}_{\ell_p}(\mathcal{P}_Q, \varepsilon, \delta, \beta/n)$}{
        $\tilde{Q}_i\gets 0$\;
    }
}
\Return $\sum_i \tilde{Q}_i$
\end{algorithm}
\end{small}
At first glance, this protocol appears to defend against both privacy and utility attacks, as each user employs a distinct shuffler and independently adds noise to their own data.
However, as we use a multi-shuffler architecture, attackers may impersonate others and send messages to others' shufflers, referred to as \textit{impersonation attacks}.
A straightforward countermeasure involves the analyzer assigning a unique random ID to each shuffler, such that only authorized users (for that particular shuffler) possess this ID. More precisely,
\begin{enumerate}
    \item \textit{Shuffler identification:} When $\mathcal{S}_j$ is instantiated, the analyzer generates an ID $r_j$ (e.g., a random string) and shares it only with those users who are authorized to send messages to $\mathcal{S}_j$.
    % to $\mathcal{S}_j$ according to $\phi$.
    \item \textit{Message authorization:} Each message sent to $\mathcal{S}_j$ must carry the corresponding ID $r_j$. If a user without knowledge of $r_j$ attempts to flood $\mathcal{S}_j$ with messages, the analyzer can immediately discard them as unauthorized.
\end{enumerate}
This idea effectively addresses the issue of impersonation attacks, even in the presence of multiple attackers, as each attacker is restricted to sending messages only to the shuffler they are authorized to access. In the following sections, we apply this idea as a general solution for defending against impersonation attacks.

\begin{theorem}
\label{theorem: Strawman}
    Given any $\varepsilon > 0, \delta >0, n \in \mathbb{Z}_+ $ and $U \in \mathbb{Z}_{+}$, when there is only one corrupted user, for any union-preserving query $Q$, the SUSDP protocol achieves that:

    \begin{itemize}
        \item The messages received by the analyzer preserve $(\varepsilon,\delta)$-DP;
        \item With probability at least $1-\beta$, the total error is bounded by: 
        \[
        O\big( n\cdot \err_{\ell_p}(\mathcal{P}_Q, \varepsilon, \delta, \beta/n)\big)\;+ \; \gamma_{\ell_p}(Q,1);
        \]
        \item In expectation, each user sends $\text{Msg}(\mathcal{P}_Q,\varepsilon,\delta,1)$ messages with each containing $\mathrm{Bit}(\mathcal{P}_Q, \varepsilon, \delta, U, 1)+O(\log n)$ bits.
    \end{itemize}

\end{theorem}

\begin{proof}
    \label{sec:Proof of Theorem 3.1}
For privacy, in the SUSDP protocol, each individual user ensures $(\varepsilon,\delta)$-DP, and the detection and recovery post processes the results from users. 
By parallel composition (Lemma~\ref{lem:Parallel Composition}) and post processing (Lemma~\ref{lem:Post Processing}), the SUSDP protocol satisfies $(\varepsilon,\delta)$-DP.

For utility, let $q$ denote the ID of the corrupted user, and $\mathcal{I}= \{1,2,\dots,n\} \setminus \{q\}$ denote the set of IDs of honest users.
% please check !!
{If the corrupted user sends excessive noisy messages and is detected, then $\tilde{Q}(D) = \sum_{i \in \mathcal{I}} \tilde{Q}_i$, and the error is bounded by
\begin{equation*}
    \begin{aligned}
        \|\tilde{Q}(D) - Q(D)\|_{\ell_p} & = \|\sum_{i \in \mathcal{I}} \tilde{Q}_i - Q(D)\|_{\ell_p} \\
        & \leq \sum_{i \in \mathcal{I}} \|\tilde{Q}_i - Q(\{x_i\})\|_{\ell_p} + \|0 - Q(\{x_q\})\|_{\ell_p} \\
        & \leq (n-1) \cdot \mathrm{Error}_{\ell_p}(\mathcal{P}_Q, \varepsilon, \delta, \beta/n) + \gamma_{\ell_p}(Q,1),
    \end{aligned}
\end{equation*}
with probability {at least} $1-\beta$ using {the} union bound.}

{If the corrupted user manipulates its data and is not detected, there must exist $c \in \mathrm{Range}(Q, 1)$ such that $\|\tilde{Q}_q - c\|_{\ell_p} \leq \mathrm{Error}_{\ell_p}(\mathcal{P}_Q,\allowbreak \varepsilon, \delta, \beta/n)$. By {the} triangle inequality,
\begin{equation*}
    \begin{aligned}
        \|\tilde{Q}_q - Q(\{x_q\})\|_{\ell_p}  \leq \|\tilde{Q}_q - c\|_{\ell_p} + \|c-Q(\{x_q\})\|_{\ell_p} 
         \leq \mathrm{Error}_{\ell_p}(\mathcal{P}_Q, \varepsilon, \delta, \beta/n) + \gamma_{\ell_p}(Q,1).
    \end{aligned}
\end{equation*}
{As} $\tilde{Q}(D) = \sum_{i} \tilde{Q}_i$ where $i \in \{1, 2, \dots, n\}$, the error is bounded by
\begin{equation*}
    \begin{aligned}
        \|\tilde{Q}(D) - Q(D)\|_{\ell_p} =  \|\sum_{i=1}^n \tilde{Q}_i - Q(D)\|_{\ell_p} \leq n \cdot \mathrm{Error}_{\ell_p}(\mathcal{P}_Q, \varepsilon, \delta, \beta/n) + \gamma_{\ell_p}(Q,1). \\ 
    \end{aligned}
\end{equation*}
Comparing these two bounds, it is clear that the worst-case error occurs when the attacker's malicious behavior is just not detected and their output is included in the final aggregation. In all, the total error is bounded by $n \cdot \mathrm{Error}_{\ell_p}(\mathcal{P}_Q, \varepsilon, \delta, \beta/n) + \gamma_{\ell_p}(Q,1)$ with probability {at least} $1-\beta$.}

For communication cost, by the complexity of protocol $\mathcal{P}_Q$, each honest user sends $\text{Msg}(\mathcal{P}_Q,\varepsilon,\delta,1)$ messages in expectation, where each message encodes two components:  
\begin{enumerate}
    \item The user data cost is $\mathrm{Bit}(\mathcal{P}_Q,\varepsilon, \delta,U,n)$ bits;
    \item A unique shuffler identifier requires $\log n$ bits.
\end{enumerate}

Note that, regardless of whether the corrupted user is present, honest users always send this amount of messages in expectation, and each message follows the same format. 
This declaration applies to all protocols proposed in this paper that build upon the shuffle-DP model and will not be reiterated in the following sections.
\end{proof}

\begin{example} [SUSDP for $Q_\text{count}$]
\label{exm:count for susdp}
    Following the bit counting example, we use $\mathcal{P}_{Q_\text{count}}$ \cite{ghazi2021differentially} to initialize our SUSDP protocol.
    Since each user independently executes the $\mathcal{P}_{Q_\text{count}}$ procedure, the difference threshold is set to $O\big(\frac{1}{\varepsilon}\log \frac{n}{\beta}\big)$ according to Lemma \ref{lem: Correlated Noise for bit counting}. Under this setup, SUSDP protocol with $\mathcal{P}_{Q_\text{count}}$ achieves the error $O\big(\frac{\sqrt{n}}{\varepsilon}\log \frac{n}{\beta}\big)$\footnote{Directly applying Theorem~\ref{theorem: Strawman} yields an error bound of $O\big(\frac{n}{\varepsilon}\log \frac{n}{\beta}\big)$. However, since the error of $\mathcal{P}_{Q_\text{count}}$ follows a Laplace distribution~\cite{ghazi2021differentially}, we can leverage its concentration bound to obtain a tighter estimate.}, and in expectation, each user sends $O\big(\frac{\log(1/\delta)}{\varepsilon}\big)$ messages, with each message containing $O(\log n)$ bits.

\end{example}

\section{Methodology}
\label{sec:Methodology}
While the SUSDP protocol effectively mitigates the poisoning attacks, the resulting increase in error is suboptimal.
This is because, in contrast to shuffle-DP which amortizes noise across all users for improved utility, each user in SUSDP independently adds noise.
In this section, we present our general framework which can not only detect and recover poisoning attacks but also enjoy the error benefits in shuffle-DP.

We first introduce a \textit{Block Shuffle-DP} (BSDP) protocol as a foundational defense mechanism, which has an error upgraded by a factor of $O(\sqrt{n})$ compared with the given shuffle-DP protocol, in Section \ref{sec:BSDP}. 
Then, we reduce the factor to $O(\log n)$ by \textit{Hierarchical Shuffle-DP} (HSDP) protocol in Section \ref{sec:HSDP}. 
Finally, in Section \ref{sec:Optimization}, we propose optimizations to reduce the communication cost and then extend our solution to the multiple corrupted users setting.

\subsection{Block Shuffle-DP Protocol}
\label{sec:BSDP}

The protocol operates on three levels. At the bottom level, we apply a shuffle-DP protocol for every single user (user-level). At the top level, we apply a shuffle-DP protocol across all users (output-level). Between this, we introduce an intermediate level, where users are partitioned into blocks of fixed size $\sqrt{n}$, and users in each block perform a shuffle-DP protocol (block-level).  

To defend against poisoning attacks, detection is performed at all three levels. 
At the user-level, the detection evaluates the reasonability of each user's output. 
At the block-level, attacks can be identified by comparing each group’s aggregated output with the sum of its members’ user-level outputs and checking whether the discrepancy falls within a reasonable bound. 
Similarly, for the output-level, attacks are detected by comparing the aggregated output over all users with the sum of block-level outputs. 
The illustration of the block shuffle-DP protocol is shown in Fig.~\ref{fig:BSDP}. For simplicity, we assume that $n$ is a perfect square. 
The detailed protocol works as follows:

\paragraph{Randomizer} Each user $i$ privates their data $x_i$ using three local randomizers, each operating with a privacy budget of \( \varepsilon/3 \) and \( \delta/3 \). 
To ensure only a constant factor increase in total error over the original shuffle-DP protocol in the absence of attackers, we allocate $\beta/2$ to output-level query, and distribute the remaining equally among the $\sqrt{n}+n$ block-level and user-level queries.
The first randomizer performs user-level randomization. The output of each user $i$ is sent to their pre-assigned shuffler $\mathcal{S}_{i}^{(1)}$.
The second randomizer performs block-level randomization. All users are partitioned into disjoint blocks of size $\sqrt{n}$, and each block jointly executes a shuffle-DP protocol.
Specifically, the block $b$ consists of users with ID in $\{(b-1)\sqrt{n}+1,\dots,b\sqrt{n}\}$.
Since the attacker may not contribute any noise, we need the remaining $\sqrt{n}-1$ honest users to collectively amortize the noise. 
Therefore, the privacy budget will be scaled by a factor of $(\sqrt{n}-1)/\sqrt{n}$.
Each user sends their output to the block-level shuffler $\mathcal{S}_{b}^{(2)}$, which collects and shuffles messages from all users belonging to block $g$.
The third randomizer performs output-level randomization, where all users jointly execute a shuffle-DP protocol. 
We also rescale the privacy budget by a factor of $(n-1)/n$.
The output is sent to a central shuffler $\mathcal{S}^{(3)}$.
The detailed procedure is shown in Algorithm~\ref{alg:Randomizer of BSDP}.
\begin{figure}[t]
  \centering

  \begin{subfigure}[b]{0.48\textwidth}
    \centering
    \includegraphics[page=1, width=\textwidth]{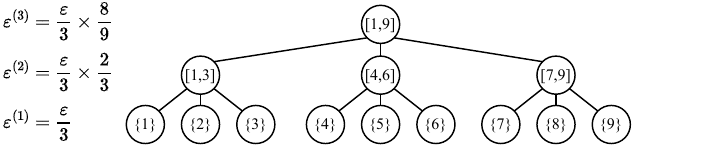}
    \caption{User grouping.}
    \label{fig:sub1}
  \end{subfigure}
  \hspace{0.01\textwidth}
  \begin{subfigure}[b]{0.48\textwidth}
    \centering
    \includegraphics[page=2, width=\textwidth]{figures/GSDP_v5.pdf}
    \caption{Estimated results from each group.}
    \label{fig:sub2}
  \end{subfigure}

  \begin{subfigure}[b]{0.48\textwidth}
    \centering
    \includegraphics[page=3, width=\textwidth]{figures/GSDP_v5.pdf}
    \caption{Bottom-to-top detection.}
    \label{fig:sub3}
  \end{subfigure}
  \hspace{0.01\textwidth}
  \begin{subfigure}[b]{0.48\textwidth}
    \centering
    \includegraphics[page=4, width=\textwidth]{figures/GSDP_v5.pdf}
    \caption{Bottom-to-top recovery.}
    \label{fig:sub4}
  \end{subfigure}

  \caption{Illustration of our block shuffle-DP protocol for bit counting with $n=9$ users. A corrupted user (ID = 1) sends excessive messages at different levels.}
  \label{fig:BSDP}
\end{figure}

\begin{small}
\begin{algorithm}[t]
\caption{Randomizer of BSDP}
\label{alg:Randomizer of BSDP}
\KwParam{{$\varepsilon, \delta, n$} }
\KwIn{$x_i$, $\mathcal{P}_Q = (\mathcal{R}, \mathcal{S}, \mathcal{A})$}

\tcp{User-Level Randomization}

{$Y_{i}^{(1)} \gets \mathcal{R}\bigl(x_i; \frac{\varepsilon}{3},\frac{\delta}{3},1\bigr)$ \;}

Send $Y_{i}^{(1)}$ \text{ to shuffler } $\mathcal{S}_{i}^{(1)}$\;

\tcp{Block-Level Randomization}

{$Y_{i}^{(2)} \gets \mathcal{R}\bigl(x_i; \frac{\varepsilon}{3}\cdot \frac{\sqrt{n}-1}{\sqrt{n}},\frac{\delta}{3}\cdot \frac{\sqrt{n}-1}{\sqrt{n}},\sqrt{n}\bigr)$ \;}

Send $Y_{i}^{(2)} $\text{ to shuffler } $\mathcal{S}^{(2)}_{b}$.

\tcp{Output-Level Randomization}

{$Y_{i}^{(3)} \gets \mathcal{R}\bigl(x_i; \frac{\varepsilon}{3}\cdot \frac{n-1}{n},\frac{\delta}{3}\cdot \frac{n-1}{n},n\bigr)$ \;}
Send $Y_{i}^{(3)} $\text{ to shuffler } $\mathcal{S}^{(3)}$\;

\end{algorithm}
\end{small}

\begin{small}
\begin{algorithm}
\caption{Analyzer of BSDP}
\label{alg:Analyzer of BSDP}
\KwParam{$\varepsilon, \delta,\beta, n$}
\KwIn{$\{Z^{(1)}_i = Y^{(1)}_i\}_i,$
$\{Z^{(2)}_{b} = \cup_{i=(b-1)\cdot\sqrt{n}+1}^{b\cdot\sqrt{n}}Y^{(2)}_i\}_{b},\;Z^{(3)} = \cup_{i} Y_i^{(3)},\;$
$\mathcal{P}_Q = (\mathcal{R}, \mathcal{S}, \mathcal{A})$}

$\Big(\varepsilon^{(1)},\delta^{(1)},\beta^{(1)}\Big)\gets \Big(\frac{\varepsilon}{3}, \frac{\delta}{3}, \frac{\beta}{2(\sqrt{n}+n)}\Big)$\;

$\Big(\varepsilon^{(2)},\delta^{(2)},\beta^{(2)}\Big)\gets \Big(\frac{\varepsilon}{3}\cdot \frac{\sqrt{n}-1}{\sqrt{n}},\frac{\delta}{3}\cdot \frac{\sqrt{n}-1}{\sqrt{n}}, \frac{\beta}{2(\sqrt{n}+n)} \Big)$\;

$\Big(\varepsilon^{(3)},\delta^{(3)},\beta^{(3)}\Big)\gets \Big( \frac{\varepsilon}{3}\cdot \frac{n-1}{n}, \frac{\delta}{3}\cdot \frac{n-1}{n}, \frac{\beta}{2} \Big)$\;

$\theta^{(r)} \leftarrow \err_{\ell_p}\Big(\mathcal{P}_Q, \varepsilon^{(r)}, \delta^{(r)},\beta^{(r)}\Big), \forall 1 \leq r \leq 3$\;

\tcp{User-Level Detection}
\For{$i \leftarrow 1$ \KwTo $n$}{
    $\tilde{Q}^{(1)}_{i} \gets \mathcal{A}(Z^{(1)}_i; \varepsilon^{(1)},\delta^{(1)},\beta^{(1)},1)$ \;
    \If{$\mathrm{dis}_{\ell_p}(\tilde{Q}^{(1)}_{i}, \mathrm{Range}(Q, 1)) > \theta^{(1)}$}{
        $\tilde{Q}^{(1)}_{i} \gets \perp$\; 
    }
}

\tcp{Block-Level Detection}

\For{$b \leftarrow 1$ \KwTo $\sqrt{n}$}{
    $\tilde{Q}^{(2)}_b \gets \mathcal{A}(Z^{(2)}_{b}; \varepsilon^{(2)},\delta^{(2)},\beta^{(2)},\sqrt{n})$\;
    \If{ $\exists \;\tilde{Q}^{(1)}_i = \perp , \;i \in [(b-1)\cdot\sqrt{n} + 1,\, b\cdot\sqrt{n}]$ \KwOr $\left|\tilde{Q}^{(2)}_b \! - \sum_{i=(b-1) \cdot \sqrt{n} +1}^{b \cdot \sqrt{n}} \tilde{Q}^{(1)}_{i}\right|_{\ell_p} > \sqrt{n} \cdot \theta^{(1)}+\theta^{(2)}$}{
      $\tilde{Q}^{(2)}_b \gets \perp$\;
    }
}
\tcp{Output-Level Detection}
$\tilde{Q}^{(3)} \gets \mathcal{A}(Z^{(3)}; \varepsilon^{(3)},\delta^{(3)},\beta^{(3)},n)$\;

\If{$ \exists \; \tilde{Q}^{(2)}_b = \perp, \; b\in[1,\sqrt{n}]$ \KwOr $\left|\tilde{Q}^{(3)} - \sum_b \tilde{Q}^{(2)}_b\right|_{\ell_p} > \sqrt{n}\cdot\theta^{(2)}+\theta^{(3)}$}{
$\tilde{Q}^{(3)} \gets \perp$;
}

\tcp{ Recovery}
\For{$ i \gets 1 $ \KwTo $n$}{
    \If{$\tilde{Q}_i^{(1)}=\perp$}{
        $\tilde{Q}_i^{(1)}\gets 0$\;
    }
}

\For{$b \leftarrow 1$ \KwTo $\sqrt{n}$}{
    \If{$\tilde{Q}^{(2)}_b = \perp$}{
        $\tilde{Q}^{(2)}_b \gets \sum_{i=(b-1)\cdot\sqrt{n}+1}^{b \cdot \sqrt{n}} \tilde{Q}^{(1)}_i$\;
    }
}

\If{$\tilde{Q}^{(3)}= \perp$}{
$\tilde{Q}^{(3)}\gets \sum_b \tilde{Q}^{(2)}_b$\;
}

\Return $\tilde{Q}^{(3)}$

\end{algorithm}
\end{small}

\paragraph{Analyzer} After receiving the shuffled messages from the shufflers, denoted by $Z^{(1)}_i$, $Z^{(2)}_b$, and $Z^{(3)}$ for the user-, block-, and output-level outputs respectively, the analyzer invokes the given $\mathcal{A}$ on each message set to compute the estimated results, written as $\tilde{Q}^{(1)}_i$, $\tilde{Q}^{(2)}_b$, and $\tilde{Q}^{(3)}$, where $i$ and $b$ index users and blocks, respectively. After collecting all results, the analyzer performs detection and recovery:

\begin{itemize}
    \item \textit{Detection.} The detection is performed in a bottom-to-top manner. At the user-level, analyzer checks whether each $\tilde{Q}^{(1)}_{i}$ falls within a reasonable range and flags abnormal outputs by setting $\tilde{Q}^{(1)}_{i} \gets \perp$. At the block-level, analyzer invalidates $\tilde{Q}^{(2)}_b$ if any member in block $b$ has been flagged at user-level; otherwise, analyzer compares $\tilde{Q}^{(2)}_b$ with the sum of $\tilde{Q}^{(1)}_{i}$ in the block. Similarly, at the output-level, analyzer checks whether $\tilde{Q}^{(3)}$ is consistent with the sum of all valid block-level outputs.

    \item \textit{Recovery.} The recovery proceeds in a bottom-to-top manner as well. For each invalid block-level output $\tilde{Q}^{(2)}_b$, the analyzer recovers it by summing all unflagged user-level outputs $\tilde{Q}^{(1)}_i$ of users within the block. After reconstructing all block-level results, the analyzer attempts to recover the output-level result. If the aggregated output-level result $\tilde{Q}^{(3)}$ is valid, it is returned as the final result. Otherwise, the analyzer aggregates all block-level outputs $\tilde{Q}^{(2)}_b$ (including those recovered from user-level results) to produce the final result.
\end{itemize}
The detailed procedure is shown in Algorithm~\ref{alg:Analyzer of BSDP}.

\begin{theorem}
\label{the: BSDP theorem}
    Given any $\varepsilon > 0, \delta >0, n \in \mathbb{Z}_+ $ and $U \in \mathbb{Z}_{+}$, when there is only one corrupted user, for any union-preserving query $Q$, the BSDP protocol achieves that:
    \begin{itemize}
        \item The messages received by the analyzer preserve $(\varepsilon,\delta)$-DP;
        \item With probability at least $1-\beta$, the total error is bounded by:
        \[
        O\big(\sqrt{n}\cdot \err_{\ell_p}(\mathcal{P}_Q, \varepsilon,\delta, \beta/n)\big)+ \gamma_{\ell_p}(Q,1)
        \]
        \item In expectation, each user sends $O\big( \mathrm{Msg}(\mathcal{P}_Q,\varepsilon,\delta,1)\big)$ messages, each containing $O\big(\mathrm{Bit}(\mathcal{P}_Q,\varepsilon,\delta,U,1)+\log n\big)$ bits.
    \end{itemize}
\end{theorem}

\begin{proof}
\label{sec:Proof of Theorem 4.1}
For privacy, each of the three levels in the BSDP protocol ensures $(\varepsilon/3, \delta/3)$-DP, where parallel composition (Lemma~\ref{lem:Parallel Composition}) applies within each level. 
By basic composition (Lemma~\ref{lem:Sequential Composition}) and post processing (Lemma~\ref{lem:Post Processing}), the whole BSDP protocol with three levels satisfies $(\varepsilon, \delta)$-DP.

    For utility, let $q$ denote the ID of the corrupted user, and 
    $b$ denotes the ID of the block with the corrupted user.
    {The overall error of the BSDP protocol consists of two parts: i) input bias, the result difference between the true input and the corrupted input, and ii) message noise, resulting from noisy messages including both legitimate randomness from the protocol and any excessive noise injected by corrupted users.
    By the union bound, with probability at least $1 - \beta$, the message noise within each level does not exceed its respective $\theta^{(r)}$ threshold; any group whose deviation surpasses $\theta^{(r)}$ is flagged as corrupted.
    }

    {If the corrupted user is not detected at any level, similar to the proof of Theorem~\ref{theorem: Strawman}, the total error of $\Sigma_i Q_i^{(1)}$ for block $b$ is bounded by $\sqrt{n}\cdot \theta^{(1)} + \gamma_{\ell_p}(Q,1)$. 
    In this case, since the corrupted user is also not detected at the block-level, we have 
    \begin{equation*}
        \left\|\tilde{Q}^{(2)}_b \! - \sum_{i=(b-1) \cdot \sqrt{n} +1}^{b \cdot \sqrt{n}} \tilde{Q}^{(1)}_{i}\right\|_{\ell_p} \leq \sqrt{n} \cdot \theta^{(1)}+\theta^{(2)}.
    \end{equation*}
    Thus, the total error for $\tilde{Q}_b^{(2)}$ is bounded by $\sqrt{n}\cdot \theta^{(1)} + \gamma_{\ell_p}(Q,1) + \sqrt{n} \cdot \theta^{(1)}+\theta^{(2)}  = 2\sqrt{n}\cdot \theta^{(1)} + \theta^{(2)} + \gamma_{\ell_p}(Q,1)$. 
    Similarly, we can conclude that the total error of $\tilde{Q}^{(3)}$ is bounded by
     \begin{equation*}
         2\sqrt{n}\cdot \theta^{(1)} +  2\sqrt{n}\cdot\theta^{(2)} + 
         \theta^{(3)} +
        \gamma_{\ell_p}(Q,1).
     \end{equation*}}

    {If the corrupted user sends excessive messages and is detected, we consider three cases based on the detection level. If the corrupted user is detected at the user-level, then both $\tilde{Q}_b^{(2)}$ and $\tilde{Q}^{(3)}$ will be marked as $\perp$. After recovery, the error for $\tilde{Q}_b^{(2)}$ is bounded by $(\sqrt{n} -1)\cdot \theta^{(1)} + \gamma_{\ell_p}(Q,1)$, 
    and the error for $\tilde{Q}^{(3)}$ is bounded by 
    $(\sqrt{n} -1)\cdot \theta^{(1)} + (\sqrt{n}-1) \cdot \theta^{(2)} + \gamma_{\ell_p}(Q,1)$.
    If the corrupted user is detected at the block-level, from the non-detection result, the total error of recovered $\tilde{Q}_b^{(2)}$ is bounded by $\sqrt{n}\cdot \theta^{(1)} + \gamma_{\ell_p}(Q,1)$. 
    The error for $\tilde{Q}^{(3)}$ is then bounded by $ \sqrt{n}\cdot \theta^{(1)} + (\sqrt{n} -1) \cdot \theta^{(2)} + \gamma_{\ell_p}(Q,1)$.
    If the corrupted user is detected at the output-level, then the error for $\tilde{Q}^{(3)}$ is bounded by $2\sqrt{n}\cdot \theta^{(1)} +  \sqrt{n}\cdot\theta^{(2)} + 
        \gamma_{\ell_p}(Q,1)$.}

    {In summary, by comparing the final error bounds for $\tilde{Q}^{(3)}$ across all cases, the worst-case error occurs when the corrupted user's malicious contribution is just large enough to evade detection and is not excluded at any level. This scenario yields the largest error bound of $2\sqrt{n}\cdot \theta^{(1)} + 2\sqrt{n}\cdot\theta^{(2)} + \theta^{(3)} + \gamma_{\ell_p}(Q,1)$, as all malicious contributions are fully included in the final output.}

    {To bound the final error in asymptotic form, recall that in  Algorithm~\ref{alg:Analyzer of BSDP}, each level $r \in \{1,2,3\}$ uses the following privacy parameters:
    \begin{align*}
    \big(\varepsilon^{(1)}, \delta^{(1)}, \beta^{(1)}\big) &= \left(\tfrac{\varepsilon}{3}, \tfrac{\delta}{3}, \tfrac{\beta}{2(\sqrt{n}+n)}\right), \\
    \big(\varepsilon^{(2)}, \delta^{(2)}, \beta^{(2)}\big) &= \left(\tfrac{\varepsilon}{3} \cdot \tfrac{\sqrt{n}-1}{\sqrt{n}}, \tfrac{\delta}{3} \cdot \tfrac{\sqrt{n}-1}{\sqrt{n}}, \tfrac{\beta}{2(\sqrt{n}+n)}\right), \\
    \big(\varepsilon^{(3)}, \delta^{(3)}, \beta^{(3)}\big) &= \left(\tfrac{\varepsilon}{3} \cdot \tfrac{n-1}{n}, \tfrac{\delta}{3} \cdot \tfrac{n-1}{n}, \tfrac{\beta}{2}\right).
    \end{align*}}   
    {{Additionally,} the error threshold at level $r$ is defined as: $$\theta^{(r)} \leftarrow \mathrm{Error}_{\ell_p}(\mathcal{P}_Q,\allowbreak \varepsilon^{(r)},\allowbreak \delta^{(r)}, \beta^{(r)}).$$ Since all privacy parameters $(\varepsilon^{(r)}, \delta^{(r)}, \beta^{(r)})$ are within constant factors of $(\varepsilon, \delta, \beta/n)$, we have:
    \[
    \theta^{(r)} = O\left( \mathrm{Error}_{\ell_p}(\mathcal{P}_Q, \varepsilon, \delta, \tfrac{\beta}{n} )\right ), \quad \text{for each } r = 1,2,3.
    \]}
    {Substituting this into the total error expression yields the final worst-case bound:
    \begin{equation*}
    \begin{aligned}
    2\sqrt{n}\cdot \theta^{(1)} + 2\sqrt{n}\cdot\theta^{(2)} + \theta^{(3)} + \gamma_{\ell_p}(Q,1) = O\left(\sqrt{n} \cdot \mathrm{Error}_{\ell_p}\left(\mathcal{P}_Q, \varepsilon, \delta, \tfrac{\beta}{n} \right)\right) + \gamma_{\ell_p}(Q,1),
    \end{aligned}
    \end{equation*}
    with probability at least $1 - \beta$.}

For the communication cost, the BSDP protocol consists of three levels, each with specific privacy budgets $(\varepsilon^{(r)}, \delta^{(r)})$ according to Algorithm~\ref{alg:Analyzer of BSDP}. 
Let $\mathrm{M}^{(r)}$ denote the number of messages sent by each user at level $r$, {we have} $\mathrm{Msg}(\mathcal{P}_Q, \varepsilon^{(r)}, \delta^{(r)}, m^{(r)})$
where $m^{(r)} = n^{(r-1)/2}$. 
By aggregating all three levels, the total number of messages per user is bounded by
\begin{equation*}
\begin{aligned}
    \sum_{r =1}^{3}\mathrm{M}^{(r)}= & \mathrm{Msg}\big(\mathcal{P}_Q, \tfrac{\varepsilon}{3}, \tfrac{\delta}{3}, 1\big) 
    +\mathrm{Msg}\big(\mathcal{P}_Q, \tfrac{\varepsilon}{3} \cdot \tfrac{\sqrt{n}-1}{\sqrt{n}}, \tfrac{\delta}{3} \cdot \tfrac{\sqrt{n}-1}{\sqrt{n}}, \sqrt{n}\big) +\mathrm{Msg}\big(\mathcal{P}_Q, \tfrac{\varepsilon}{3} \cdot \tfrac{n-1}{n}, \tfrac{\delta}{3} \cdot \tfrac{n-1}{n}, n\big) \\
    =&\; O\big(\mathrm{Msg}(\mathcal{P}_Q, \varepsilon, \delta, 1)\big).
\end{aligned}
\end{equation*}

    Similar to Theorem~\ref{theorem: Strawman}, each message contains both user data and a shuffler identifier. 
    The user data requires $\mathrm{Bit}(\mathcal{P}_Q, \varepsilon^{(r)}, \delta^{(r)}, U, m^{(r)})$ bits at level $r$, which is bounded by $O(\mathrm{Bit}(\mathcal{P}_Q, \varepsilon, \delta, U, 1))$. 
    Since the BSDP protocol employs $n + \sqrt{n} + 1$ shufflers in total, each message should include a $\log(n + \sqrt{n} + 1)$-bit identifier to indicate the intended shuffler.

\end{proof}

\begin{example}[BSDP for $Q_\text{count}$]
    We initialize the BSDP protocol using $\mathcal{P}_{Q_\text{count}}$ from~\cite{ghazi2021differentially}.
    When there are no attackers, BSDP achieves an error of $O\big(\frac{1}{\varepsilon}\log\frac{1}{\beta}\big)$, which is asymptotically the same as the original $\mathcal{P}_{Q_\text{count}}$.
    In the presence of a single corrupted user, the error changes to $O\big(\frac{n^{1/4}}{\varepsilon}\log\frac{n}{\beta}\big)$\footnote{
    We achieve this error bound by applying Theorem~\ref{the: BSDP theorem} and the concentration bound.
    Furthermore, we replace the terms $\sqrt{n}$ with $n^{1/4}$ in lines 11 and 14 of Algorithm~\ref{alg:Analyzer of BSDP}.
    }.
    Each user sends $O\big(\frac{\log(1/\delta)}{\varepsilon}\big)$ messages in expectation, each message containing $O(\log n)$ bits.
\end{example}

\subsection{Hierarchical Shuffle-DP Protocol}
\label{sec:HSDP}

Recall that in the block shuffle-DP protocol, recovering the block-level or output-level result requires aggregating $O(\sqrt{n})$ results from the lower level, which contributes an error term \( O(\sqrt{n}) \).
Our idea is to introduce additional levels, allowing us to iteratively combine and verify smaller blocks, thereby limiting the attacker's ability to disrupt the protocol.
Building on this insight, we propose an advanced multi-level block protocol, referred to as the \emph{Hierarchical Shuffle-DP} (HSDP) protocol, which is structured as a binary tree. This hierarchical design reduces the number of lower-level results required to reconstruct an upper-level output, thereby enabling more precise control over error accumulation across levels.

Specifically, at the bottom level, each user individually applies the single-user shuffle-DP protocol. Then, in each upper level, two groups from the previous level are merged to form a new group, continuing recursively until all users are consolidated into a single group at the top level. For simplicity, we assume that \( n \) is a power of 2. 
There are \( \log n + 1 \) levels, where at level \( r \in \{1, 2, \dots, \log n + 1\} \), each group has size \( 2^{r-1} \). Each group performs a shuffle-DP protocol. 
On the analyzer's side, detection and recovery are carried out in a manner similar to BSDP, where we check the reasonability of each group's result and perform recovery using the results of its subgroups from the lower level if needed.
This hierarchical process is illustrated in Fig.~\ref{fig:HSDP_analyzer}.
Precisely, the protocol works as follows: 

\begin{figure*}[t]
  \centering

  \begin{subfigure}[b]{0.48\textwidth}
    \centering
\includegraphics[page=1,width=\textwidth]{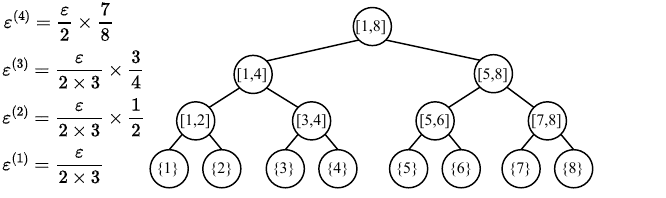}
    \caption{Binary tree structure of user grouping.}
    \label{fig:sub1}
  \end{subfigure}
  \hspace{0.01\textwidth}
  \begin{subfigure}[b]{0.48\textwidth}
    \centering
    \includegraphics[page=2,width=\textwidth]{figures/HSDP/HSDP.drawio.pdf}
    \caption{Estimated results from each group.}
    \label{fig:sub2}
  \end{subfigure}

  \begin{subfigure}[b]{0.48\textwidth}
    \centering
    \includegraphics[page=3,width=\textwidth]{figures/HSDP/HSDP.drawio.pdf}
    \caption{Bottom-to-top detection.}
    \label{fig:sub3}
  \end{subfigure}
  \hspace{0.01\textwidth}
  \begin{subfigure}[b]{0.48\textwidth}
    \centering
    \includegraphics[page=4,width=\textwidth]{figures/HSDP/HSDP.drawio.pdf}
    \caption{Bottom-to-top recovery.}
    \label{fig:sub4}
  \end{subfigure}

  \caption{Illustration of our hierarchical shuffle-DP protocol for bit counting, with $n=8$ users. A corrupted user (ID = 1) sends excessive messages at different levels.}
  \label{fig:HSDP_analyzer}
\end{figure*}

\paragraph{Randomizer} Each user \( i \) privatizes their data \( x_i \) using \( \log n + 1 \) local randomizers, each corresponding to one level of the hierarchical protocol.
We allocate half of the privacy budget, i.e., \( \varepsilon/2 \) and \( \delta/2 \), to the randomizer at the top level, and distribute the remaining half equally among the lower \( \log n \) levels. 
For $\beta$, we allocate $\beta/2$ to the query at the top level, and split the remaining half equally across the $2n-2$ queries at lower levels.
This allocation ensures that, in the absence of any attacker, the overall error increases by at most a constant factor compared to the given shuffle-DP protocol.
The first randomizer is used to perform user-level randomization with its output sent to the shuffler $\mathcal{S}^{(1)}_i$ assigned to user $i$.
The $r$-th randomizer with $r>1$ performs a shuffle-DP over a group $g$, which consists of users with ID in $\{(g-1)\cdot 2^{r-1}+1, \dots, g \cdot 2^{r-1}\}$. The output is sent to shuffler $\mathcal{S}^{(r)}_g$. 
To defend against the case where the attacker contributes no noise, 
similar to the block approach, we rescale the privacy budget by a factor of $(2^{r-1}-1)/2^{r-1}$ at each level $r>1$.
The detailed procedure is shown in Algorithm~\ref{alg:Randomizer of HSDP}.

\begin{small}
\begin{algorithm}[t]
\caption{Randomizer of HSDP}
\label{alg:Randomizer of HSDP}
\KwParam{{$\varepsilon, \delta ,n$}}
\KwIn{$x_i$, $\mathcal{P}_Q = (\mathcal{R}, \mathcal{S}, \mathcal{A})$}

{$Y^{(1)}_i \gets \mathcal{R}\bigl(x_i;\frac{\varepsilon}{2\log n} ,\frac{\delta}{2\log n}, 1 \bigr)$\;}
Send $Y^{(1)}_i$ to shuffler $\mathcal{S}^{(1)}_i$\;

\For{$r \gets 2$ \KwTo $\log n$}{

    {$Y^{(r)}_i \gets \mathcal{R}\bigl(x_i;\frac{\varepsilon}{2\log n} \! \cdot \! \frac{2^{r-1}-1}{2^{r-1}},\frac{\delta}{2\log n} \! \cdot \! \frac{2^{r-1}-1}{2^{r-1}}, 2^{r-1} \bigr);$}

    Send $Y^{(r)}_i$ to shuffler $\mathcal{S}^{(r)}_g$, where $g = \lceil i / 2^{r-1} \rceil$\;
}

{$Y^{(\log n+1)}_i \gets \mathcal{R}\bigl(x_i;\;\frac{\varepsilon}{2} \cdot\frac{n-1}{n},\frac{\delta}{2} \cdot \frac{n-1}{n},n \bigr) $\;}

Send $Y^{(\log n+1)}_i$ to shuffler $\mathcal{S}^{(\log n+1)}_1$\;

\end{algorithm}
\end{small}

\paragraph{Analyzer}  After receiving shuffled messages $Z^{(r)}_g$ from the shuffler $\mathcal{S}^{(r)}_g$, the analyzer invokes the given $\mathcal{A}$ to compute the estimated results $\tilde{Q}_g^{(r)}$. After collecting all results for every $r\in \{1, \dots , \log n+1\} $ and $g\in \{1, \dots, n/2^{r-1}\}$, the analyzer performs detection and recovery:
\begin{itemize}
    \item \textit{Detection.} Similar to the block shuffle-DP protocol, the detection is split into two primary stages. At the user-level ($r=1$), analyzer checks if each $\tilde{Q}^{(1)}_{i}$ falls within a reasonable range and flags abnormal outputs by setting $\tilde{Q}^{(1)}_{i} \gets \perp$. At higher levels, the analyzer proceeds from bottom to top: for each group at level \(r\), check its result with the sum of results from its two subgroups at level \(r-1\). If any subgroup was previously flagged as invalid \((\tilde{Q}^{(r-1)}_{2g-1} \text{ or } \tilde{Q}^{(r-1)}_{2g} = \perp)\), the current group \(\tilde{Q}^{(r)}_{g}\) is also marked invalid. Otherwise, the analyzer compares \(\tilde{Q}^{(r)}_{g}\) against the sum of its two subgroups $\tilde{Q}^{(r-1)}_{2g-1} + \tilde{Q}^{(r-1)}_{2g}$.
    The group \( \tilde{Q}^{(r)}_{g} \) is marked invalid if the difference exceeds the cumulative error bound, that is, the sum of the error bound from the current group and its subgroups.
    \item \textit{Recovery.} 
    The recovery process proceeds in a bottom-up manner. At the user-level, any invalid result is set to $0$. From the second level onward, if a result is marked as invalid (i.e., equal to $\perp$), it is recovered by aggregating the outputs of its two subgroups from the lower level: $\tilde{Q}^{(r)}_g \gets \tilde{Q}^{(r-1)}_{2g-1} + \tilde{Q}^{(r-1)}_{2g}$.
    After completing the recovery, the final output is $\tilde{Q}^{(\log n + 1)}_1$.
    
\end{itemize}
The detailed procedure is shown in Algorithm~\ref{alg:Analyzer of HSDP}.

\begin{small}
\begin{algorithm}[t]
\caption{Analyzer of HSDP}
\label{alg:Analyzer of HSDP}
\KwParam{$\varepsilon, \delta, \beta, n$}
\KwIn{
{$ \{Z_g^{(r)} = \cup_{i=(g-1)\cdot2^{r-1}+1}^{g\cdot{2^{r-1}}} Y^{(r)}_i \}_{r,g}$}, $ \mathcal{P}_Q = (\mathcal{R}, \mathcal{S}, \mathcal{A})$
}

$\left(\varepsilon^{(r)}, \delta^{(r)}, \beta^{(r)}\right) \gets \begin{cases}
    \left(\frac{\varepsilon}{2\log n}, \frac{\delta}{2\log n}, \frac{\beta}{2(2n-2)}\right) & \text{if } r = 1 \\
    \left(\frac{\varepsilon}{2\log n}  \! \cdot \! \frac{2^{r-1}-1}{2^{r-1}}, \frac{\delta}{2\log n}  \! \cdot \! \frac{2^{r-1}-1}{2^{r-1}}, \frac{\beta}{2(2n-2)} \right) & \text{if } 2 \leq r \leq \log n \\
    \left( \frac{\varepsilon}{2}, \frac{\delta}{2}, \frac{\beta}{2} \right) & \text{if } r = \log n + 1
\end{cases}$\\
$\theta^{(r)} \leftarrow \err_{\ell_p}\Big(\mathcal{P}_Q, \varepsilon^{(r)}, \delta^{(r)},\beta^{(r)}\Big), \forall 1 \leq r \leq \log n + 1$\;

\BlankLine
\tcp{Detection in the Bottom Level}
 \For{$i \leftarrow 1$ \KwTo $n$}{
    $\tilde{Q}^{(1)}_{i} \gets \mathcal{A}\bigl(Z^{(1)}_i;\;\varepsilon^{(1)}  ,\, \delta^{(1)},\, \beta^{(1)},\,1\bigr);$
    
    \If{$\mathrm{dis}_{\ell_p}\bigl(\tilde{Q}^{(1)}_{i}, \mathrm{Range}(Q,1)\bigr) 
        > \theta^{(1)} $}{
        $\tilde{Q}^{(1)}_{i} \gets \perp$\;
    }
}

\BlankLine
\tcp{Detection in Upper Levels}
\nl \For{$r \leftarrow 2$ \KwTo $\log n+1$}{
    \For{$g \leftarrow 1$ \KwTo $n/2^{r-1}$}{
        $\tilde{Q}^{(r)}_{g} \gets \mathcal{A}\bigl(Z^{(r)}_g;\;\varepsilon^{(r)}, \delta^{(r)},\beta^{(r)},2^{r-1}\bigr)$\;
        \If{$\bigl(\tilde{Q}^{(r-1)}_{2g-1} = \perp\bigr) \;\lor\;\bigl(\tilde{Q}^{(r-1)}_{2g} = \perp\bigr)$ \KwOr $\left| \tilde{Q}^{(r)}_g - \left( \tilde{Q}^{(r-1)}_{2g-1} + \tilde{Q}^{(r-1)}_{2g} \right) \right|_{\ell_p} > 2 \cdot \theta^{(r-1)}+\theta^{(r)}$}{
            $\tilde{Q}^{(r)}_{g} \gets \perp$\;
        }

    }
}

\BlankLine
\tcp{Recovery}
\For{$i \leftarrow 1$ \KwTo $n$}{
        \If{$\tilde{Q}^{(1)}_{i} = \perp$}{
            $\tilde{Q}^{(1)}_{i} \gets 0$\;
        }
    }

\For{$r \leftarrow 2$ \KwTo $\log n+1$}{
    \For{$g \leftarrow 1$ \KwTo $n/2^{r-1}$}{
        \If{$\tilde{Q}^{(r)}_{g} = \perp$}{
            $\tilde{Q}^{(r)}_{g} \gets \tilde{Q}^{(r-1)}_{2g}+\tilde{Q}^{(r-1)}_{2g-1}$\;
        }
    }
}

\KwRet{$\tilde{Q}^{(\log n+1)}_{1}$}

\end{algorithm}
\end{small}

\begin{theorem}
\label{the:HSDP theorem}
Given any $\varepsilon > 0, \delta >0, n \in \mathbb{Z}_+, $ $U \in \mathbb{Z}_{+}$, when there is only one corrupted user, for any union-preserving query $Q$, the HSDP protocol achieves that:

\begin{itemize}
    \item The messages received by the analyzer preserve $(\varepsilon,\delta)$-DP;
    \item 
    With probability at least $1-\beta$, the total error is bounded by:
    \[
        O\Big(\log n\cdot \err_{\ell_p}\big(\mathcal{P}_Q,\frac{\varepsilon}{\log n},\frac{\delta}{\log n}, \frac{\beta}{n}\big)\Big)+\gamma_{\ell_p}(Q,1);
    \]

    \item In expectation, each user sends $O\left(\log n \cdot \mathrm{Msg}(\mathcal{P}_Q, \frac{\varepsilon}{\log n }, \frac{\delta}{\log n}, 1)\right)$ messages, each containing $O\big(\mathrm{Bit}(\mathcal{P}_Q,\frac{\varepsilon}{\log n}, \frac{\delta}{\log n},U,1)+\log n\big)$ bits.

\end{itemize}

\end{theorem}

\begin{proof}
For privacy, among the $\log n + 1$ levels in the HSDP protocol, the first $\log n$ levels each ensure $(\varepsilon / 2\log n, \delta / 2\log n)$-DP, while the final level ensures $(\varepsilon / 2, \delta / 2)$-DP. 
By basic composition (Lemma~\ref{lem:Sequential Composition}) and post processing (Lemma~\ref{lem:Post Processing}), the HSDP protocol as a whole satisfies $(\varepsilon, \delta)$-DP.

For utility analysis, let $q$ denote the ID of the corrupted user. By the union bound, with probability at least $1 - \beta$, the {message noise} at level $r$ prior to recovery is bounded by $\theta^{(r)}$.

{We show that the worst-case error of the HSDP protocol occurs when the corrupted user evades all detection. 
Similar to the proof of Theorem~\ref{the: BSDP theorem},
if the corrupted user is not detected at any level, the error upper bound $B_r$ for level $r=1$ is given by $\theta^{(1)} + \gamma_{\ell_p}(Q,1)$.
For the group at level $r \in \{2, 3, \dots, \log n + 1\}$ containing the corrupted user, the error bound is
\begin{equation*}
    B_r = 4 \cdot \theta^{(1)} + 4 \cdot \theta^{(2)} + \dots +4 \cdot \theta^{(r-1)} + \theta^{(r)} + \gamma_{\ell_p}(Q,1).
\end{equation*}
The error for other groups composed entirely of honest users is bounded by $\theta^{(r)}$.}

{
If the HSDP protocol detects the attack, let $p$ denote the lowest level at which the corrupted user is identified. We define $R^p_r$ as the error upper bound for the group containing the corrupted user at level $r$ after recovery.
For level $r =p$, if $p =1$, $R^p_p = \gamma_{\ell_p}(Q,1)$. 
For cases where $r=p = \{2, 3,\dots, \log n+1\}$, 
\begin{equation*}
   R^p_p= B_{p -1} + \theta^{(p-1)}.
\end{equation*}
According to Algorithm~\ref{alg:Analyzer of HSDP}, all higher-level groups that contain the corrupted user also require recovery. As a result, for $i \geq 0$,
\begin{equation*}
 R^p_{p + i +1} = R^p_{p + i}+ \theta^{(p+i)},
\end{equation*}
Thus, the error bound for the final result satisfies:
\begin{equation*}
R^p_{\log n + 1} = 
\begin{cases}
\gamma_{\ell_p} + \sum_{i=i}^{\log n} \theta^{(i)}  & \text{if } p = 1, \\
B_{p - 1} + \sum_{i=p-1}^{\log n} \theta^{(i)} & \text{if } p > 1. 
\end{cases}
\end{equation*}
}

{Since $B_{\log n + 1} \geq R^p_{\log n + 1}$ for any $p \in \{1, 2, \dots, \log n + 1\}$, we can conclude the worst error happens when the corrupted user evades all detection.}
{
Therefore, with probability at least $1- \beta$, the error of the HSDP protocol is bounded by
\begin{equation*}
\begin{aligned}
    B_{\log n + 1} &= 4\cdot \theta^{(1)} +  4\cdot\theta^{(2)} + \dots +4 \cdot \theta^{(\log n)}+ 
         \theta^{(\log n + 1)} +
        \gamma_{\ell_p}(Q,1).
\end{aligned}
\end{equation*}}

{We recall that $\theta^{(r)}$ is defined as $\mathrm{Error}_{\ell_p}\left(\mathcal{P}_Q, \varepsilon^{(r)}, \delta^{(r)}, \beta^{(r)}\right)$, where the parameters at each level are allocated as:
\begin{align*}
&(\varepsilon^{(r)}, \delta^{(r)}, \beta^{(r)}) \gets 
\begin{cases}
\left( \frac{\varepsilon}{2 \log n}, \frac{\delta}{2 \log n}, \frac{\beta}{2(2n - 2)} \right), & \text{if } r = 1; \\
\left( \frac{\varepsilon}{2 \log n} \cdot \frac{2^{r - 1} - 1}{2^{r - 1}}, 
        \frac{\delta}{2 \log n} \cdot \frac{2^{r - 1} - 1}{2^{r - 1}}, 
        \frac{\beta}{2(2n - 2)} \right), & \text{if } 2 \le r \le \log n; \\
\left( \frac{\varepsilon}{2}, \frac{\delta}{2}, \frac{\beta}{2} \right), & \text{if } r = \log n + 1.
\end{cases}
\end{align*}}

{{As illustrated,} all $(\varepsilon^{(r)}, \delta^{(r)}, \beta^{(r)})$ are either within constant factors of, or strictly larger than $(\frac{\varepsilon}{\log n}, \frac{\delta}{\log n}, \frac{\beta}{n})$.
% Since $\theta^{(r)}$ decreases with larger privacy budgets, 
Therefore, we can uniformly upper bound all $\theta^{(r)}$ by:
\[
\theta^{(r)} = O\left( \mathrm{Error}_{\ell_p}(\mathcal{P}_Q, \frac{\varepsilon}{\log n},\frac{\delta}{\log n}, \frac{\beta}{n} ) \right), \quad \forall r {\in \{1,2, \dots, \log n + 1\}}.
\]
Therefore, the worst-case total error is bounded by:
\begin{equation*}
\begin{aligned}
    B_{\log n + 1} &= 4\cdot \theta^{(1)} +  4\cdot\theta^{(2)} + \dots +4 \cdot \theta^{(\log n)}+ 
         \theta^{(\log n + 1)} +
        \gamma_{\ell_p}(Q,1) \\
    &= O\Big(\log n\cdot \err_{\ell_p}\big(\mathcal{P}_Q,\frac{\varepsilon}{\log n},\frac{\delta}{\log n}, \frac{\beta}{n}\big)\Big)+\gamma_{\ell_p}(Q,1).
\end{aligned}
\end{equation*}}

For communication cost, similar to 
% the analysis 
{the proof of}
% for 
Theorem~\ref{the: BSDP theorem}, 
% the total messages consist of messages sent by user at all levels.
% Each level is allocated specific privacy budgets according to Algorithm~\ref{alg:Analyzer of HSDP}. 
% At each level, 
each user sends $\mathrm{M}^{(r)} = \mathrm{Msg}(\mathcal{P}_Q,\varepsilon^{(r)},\allowbreak \delta^{(r)},m^{(r)})$ messages {at level~$r$}, where $\varepsilon^{(r)}$ and $\delta^{(r)}$ follows Algorithm~\ref{alg:Analyzer of HSDP}, and $m^{(r)}=2^{r-1}$.
Summing up all {messages}
% messaged amount 
in each level, we get
\begin{equation*}
\begin{aligned}
\sum_{r=1}^{\log n +1} \mathrm{M}^{(r)} & = \sum_{r=1}^{\log n} O(\mathrm{Msg}(\mathcal{P}_Q,\frac{\varepsilon}{\log n},\frac{\delta}{\log n},2^{r-1})) +O(\mathrm{Msg}(\mathcal{P}_Q,{\varepsilon},{\delta},n)) \\
& = O(\log n \cdot \mathrm{Msg}(\mathcal{P}_Q,\frac{\varepsilon}{\log n},\frac{\delta}{\log n},1)).
\end{aligned}
\end{equation*}
As previously discussed for both SUSDP and BSDP protocols, each message consists of user data and a shuffler identifier. 
For the HSDP protocol, the data payload at level $r$ requires $\mathrm{Bit}(\mathcal{P}_Q, \varepsilon^{(r)}, \delta^{(r)}, U, m^{(r)})$ bits, 
which is bounded by $O(\mathrm{Bit}(\mathcal{P}_Q, \tfrac{\varepsilon}{\log n}, \tfrac{\delta}{\log n}, U, 1))$. 
The number of shufflers used in HSDP is $2n - 1$. Thus, each message includes an additional $\log(2n - 1)$ bits to specify the intended shuffler.

\end{proof}

{\paragraph{Intuition of our enhancement over SUSDP and BSDP}
The key is to strike a balance between utility and robustness. 
Shuffle-DP achieves high utility by amortizing noise across all users, but its strong anonymity makes recovery infeasible under poisoning attacks.
SUSDP goes to the opposite extreme: each user uses a dedicated shuffler, allowing the analyzer to observe and filter individual outputs. However, this degenerates to local-DP, where aggregating $n$ independently noised reports incurs $O(n)$ error. BSDP takes a step toward balance with a three-level structure,  reducing noise accumulation to $O(\sqrt{n})$ by grouping users into blocks of size $\sqrt{n}$.
}
{HSDP further improves this trade-off through a hierarchical structure, offering two key benefits. (i) Utility: with only O($\log n$) levels, the final output aggregates a logarithmic number of independently noised outputs, leading to O($\log n$) noise accumulation. (ii) Robustness: detection begins at lower-level groups where adversarial impact is limited, and recursively validates higher-level outputs against their subgroups, thereby bounding the influence of corrupted users throughout the hierarchy.}

\begin{example}[HSDP for $Q_\text{count}$]
\label{exm: HSDP for Bit Counting}
We instantiate the HSDP protocol using $\mathcal{P}_{Q_\text{count}}$~\cite{ghazi2021differentially}.
When there are no attackers, HSDP achieves an error of $O\big( \frac{1}{\varepsilon} \log \frac{1}{\beta} \big)$, which is asymptotically the same as the original $\mathcal{P}_{Q_\text{count}}$.
In the presence of a single corrupted user, the error increases to $O\big(\frac{\log^{2} n}{\varepsilon}\log \frac{n}{\beta}\big)$.
In expectation, each user sends $O\big( \frac{\log n \cdot \log (1/\delta)}{\varepsilon}\big)$ messages, each containing $O(\log n)$ bits.
\end{example}

\subsection{Optimization and Extension}
\label{sec:Optimization}

So far, the hierarchical shuffle-DP protocol has successfully reduced the error overhead of defending against poisoning attacks to a polylogarithmic level. In this section, we further propose a solution to reduce communication costs and extend the protocol to handle multiple corrupted users.

\subsubsection{Reducing communication cost}
\label{sec:lambda}
Using bit counting as an example: each user sends $O\big( \frac{\log n \cdot \log\left(1/\delta\right)}{\varepsilon }\big)$
messages. The $\log n$ factor arises because each user participates in $\log n$ levels of the shuffle-DP protocols. At lower levels, where group sizes are small, each user must contribute a larger number of noise messages to ensure sufficient privacy.
Therefore, a natural idea is to increase the group size at the lower levels and reduce the number of levels.

Specifically, if we change the first-level group size from 1 to \(\lambda\), then each user’s overall communication cost changes to $$\sum_{i=\log \lambda}^{\log n} 1 + O\left(\frac{\log n \cdot \log(1/\delta)}{\varepsilon \cdot 2^i}\right)= \log \frac{n}{\lambda} + O\left(\frac{\log n \cdot \log (1/\delta)}{\varepsilon\lambda}\right),$$
indicating a reduction in the total number of messages as \(\lambda\) grows.

However, increasing $\lambda$ leads to a significant increase in error, because the corrupted user can destroy the result of an entire group of size $\lambda$. In the extreme case where $\lambda = n$, it will degrade into traditional shuffle-DP protocol.
Therefore, the choice of $\lambda$ involves a trade-off between communication cost and accuracy.

In our optimization, we set $\lambda = \log n \cdot \log (1/\delta)$. Similar to the hierarchical shuffle-DP protocol, we allocate half of the privacy budget to the top level, and distribute the remaining half equally among the lower $\log \frac{n}{\lambda}$ levels. 
The detection and recovery processes are the same, except that we treat every $\lambda$ users as a group and need to check whether $\text{dis}_{\ell_p}(\tilde{Q}_i^{(1)}, \text{Range}(Q, \lambda)) > \theta^{(1)}$ at the first level.

\begin{theorem}
\label{the:OHSDP}
    Given any $\varepsilon>0,\; \delta >0,\; n \in \mathbb{Z}_+, \;U \in \mathbb{Z}_{+}$, $\lambda = \log n \cdot \log (1/\delta)$, when there is only a corrupted user, for any union-preserving query $Q$, the \textit{Optimized HSDP} (OHSDP) protocol achieves the following:
    \begin{itemize}
        \item The messages received by the analyzer preserve $(\varepsilon,\delta)$-DP;
        \item With probability at least $1-\beta$, the total error is bounded by:\[O\Big(\log n \cdot \err_{\ell_p}\big(\mathcal{P}_Q, \frac{\varepsilon}{\log n}, \frac{\delta}{\log n},  \frac{\beta}{n}\big)\Big) +  \gamma_{\ell_p}(Q, \log n \cdot \log (1/\delta));\]
        \item In expectation, each user sends 
        \[
        O\Big(\log n \cdot \mathrm{Msg}(\mathcal{P}_Q, \frac{\varepsilon}{\log n}, \frac{\delta}{\log n},\log n \cdot \log (1/\delta))\Big)
        \]
        messages, each containing $O\big(\mathrm{Bit}(\mathcal{P}_Q,\frac{\varepsilon}{\log n}, \frac{\delta}{\log n},U,\log n \cdot \log (1/\delta))+\log n\big)$ bits. 
    \end{itemize}
\end{theorem}

\begin{proof}
Similar to the proof of Theorem~\ref{the:HSDP theorem}, the OHSDP protocol also satisfies $(\varepsilon, \delta)$-DP.

{For utility, with probability at least $1 - \beta$, the {message noise} at level $r$ is bounded by $\theta^{(r)}$, for all $r \in \{1, 2, \dots, \log(n/\lambda) + 1\}$.
The main difference between OHSDP and HSDP protocols is that OHSDP aggregates level $1$ users into groups of size $\lambda$, which are then handled identically to higher-level groups.
Unlike the proof of Theorem~\ref{the:HSDP theorem}, in the OHSDP protocol,  {if the corrupted user is not detected at the bottom level, the error upper bound for the bottom-level group containing the corrupted user is
\begin{equation*}
    B_1 = \theta^{(1)} + \gamma_{\ell_p}(Q, \lambda).
\end{equation*}
If the corrupted user is detected, then the error for the same group after recovery becomes
\begin{equation*}
    R^1_1 = \gamma_{\ell_p}(Q, \lambda).
\end{equation*}
Nevertheless, this difference does not affect the conclusion that, with probability at least $1 - \beta$, the worst-case error in the OHSDP protocol occurs when the corrupted user evades all levels of detection.}
} 

{
Therefore, substituting $\lambda = \log n \cdot \log (1/\delta)$, the error of the OHSDP protocol is bounded by 
\begin{equation*}
\begin{aligned}
    B_{\log (n/\lambda) + 1} = 4\cdot \theta^{(1)} + \dots 
       + 4\cdot\theta^{(\log (n/\lambda))} +  \theta^{(\log (n/\lambda) + 1)} +
        \gamma_{\ell_p}(Q,\lambda). 
\end{aligned}
\end{equation*}
with probability at least $1- \beta$.
Similar to the proof of Theorem~\ref{the:HSDP theorem}, $\theta ^{(r)}$ can be safely bounded by:
\begin{equation*}
\begin{aligned}
    \theta ^{(r)} &= O\Big( \err_{\ell_p}\big(\mathcal{P}_Q,\frac{\varepsilon}{\log \frac{n}{\lambda}},\frac{\delta}{\log \frac{n}{\lambda}}, \frac{\lambda \beta}{n}\big)\Big)= \\
    &O\Big(\err_{\ell_p}\big(\mathcal{P}_Q,\frac{\varepsilon}{\log n},\frac{\delta}{\log n}, \frac{\beta}{n}\big)\Big), \quad \forall r {\in \{1,2, \dots, \log \frac{n}{\lambda} + 1\}}.
\end{aligned}
\end{equation*}
Therefore, the worst-case total error of the OHSDP protocol is bounded by: 
\begin{equation*}
\begin{aligned}
    B_{\log (n/\lambda) + 1} &= 4\cdot \theta^{(1)} + \dots 
       + 4\cdot\theta^{(\log (n/\lambda))} +  \theta^{(\log (n/\lambda) + 1)} +
        \gamma_{\ell_p}(Q,\lambda) \\
    &= O\Big(\log {n} \cdot \err_{\ell_p}\big(\mathcal{P}_Q,\frac{\varepsilon}{\log n},\frac{\delta}{\log n}, \frac{\beta}{n}\big)\Big)  +\gamma_{\ell_p}(Q,\log n \cdot \log (1/\delta)), 
\end{aligned}
\end{equation*}
with probability at least $1- \beta$.
}

For communication cost, each user sends messages across $L = \log \frac{n}{\lambda} + 1$ levels in the OHSDP protocol.
At level $r \in \{1, \ldots, L\}$, the number of messages sent by each user is given by
\(
\mathrm{M}^{(r)} = \mathrm{Msg}(\mathcal{P}_Q, \varepsilon^{(r)}, \delta^{(r)}, m^{(r)}),
\)
where the parameters are defined as:
\[
(\varepsilon^{(r)}, \delta^{(r)}, m^{(r)}) =
\begin{cases}
\left(\tfrac{\varepsilon}{2L}, \tfrac{\delta}{2L}, \lambda\right) & \text{if } r = 1, \\[6pt]
\left(\tfrac{\varepsilon(2^{r-1}-1)}{2^{r}L}, \tfrac{\delta(2^{r-1}-1)}{2^{r}L}, 2^{r-1}\lambda\right) & \text{if } 1 < r < L, \\[6pt]
\left(\tfrac{\varepsilon(n-1)}{2n}, \tfrac{\delta(n-1)}{2n}, n\right) & \text{if } r = L.
\end{cases}
\]
By summing over all levels, the total number of messages sent by each user is bounded by:
\begin{equation*}
\begin{aligned}
    \sum_{r=1}^{L} \mathrm{M}_r &= O(\log\frac{n}{\lambda}\cdot \mathrm{Msg}(\mathcal{P}_Q, \frac{\varepsilon}{\log \frac{n}{\lambda}}, \frac{\delta}{ \log \frac{n}{\lambda}},\lambda)+\mathrm{Msg}(\mathcal{P}_Q, \frac{\varepsilon}{2}, \frac{\delta}{ 2},n)) \\
    &= O(\log n \cdot \mathrm{Msg}(\mathcal{P}_Q, \frac{\varepsilon}{\log n}, \frac{\delta}{ \log n},\lambda))
\end{aligned}
\end{equation*}

At level $r$, the user data requires $\mathrm{Bit}(\mathcal{P}_Q, \varepsilon^{(r)}, \delta^{(r)},\allowbreak U, m^{(r)})$ bits, 
which is asymptotically bounded by $O(\mathrm{Bit}(\mathcal{P}_Q, \tfrac{\varepsilon}{\log n}, \allowbreak\tfrac{\delta}{\log n}, U, \lambda))$ across all levels. 
The total number of shufflers used in OHSDP is $2n/\lambda - 1$. 
Therefore, each message additionally includes $O(\log n)$ bits to specify the target shuffler.

\end{proof}

\begin{example}[OHSDP for $Q_\text{count}$]
We instantiate the OHSDP protocol using $\mathcal{P}_{Q_\text{count}}$~\cite{ghazi2021differentially}.
When there are no attackers, the protocol achieves an error of $O\big(\frac{1}{\varepsilon}\log\frac{1}{\beta}\big)$, which is asymptotically the same as $\mathcal{P}_{Q_\text{count}}$.
In the presence of a single corrupted user, the error increases to $O\big(\frac{\log^{2} n}{\varepsilon} \log \frac{n}{\beta} \big)$.
In expectation, each user sends $O\big(\log n\big)$ messages, with each message containing $O(\log n)$ bits.
\end{example}

\subsubsection{Handling multiple corrupted users}
\label{sec: Multiple Attackers}

Consider the case where the poisoning attacker can corrupt multiple users. 
To ensure privacy, since in one group there are at most $\hat{k}$ corrupted users, the privacy budget must be scaled by a factor of $(c - \hat{k})/{c}$, where $c$ denotes the group size. However, if $c$ is too small, the privacy budget will be split by at most $k$, thus significantly degrading utility.

To address this issue, we require each group to have strictly more honest members than attackers. Recall that $\hat{k}$ is polylogarithmic in $n$, we choose the first-level size \(\lambda > 2\hat{k}\) to ensure that the lack of noise from attackers cannot compromise the privacy of honest users. For simplicity, we still set $\lambda=\log n \cdot \log (1/\delta)$ with the assumption that $\hat{k}<\log n \cdot \log (1/\delta) / 2$.
It is trivial to see that our detection and recovery mechanisms naturally support this scenario to address detection and recovery in the multiple corrupted users setting. The primary difference lies in the utility guarantee, since multiple corrupted users can collectively affect a larger number of results.

\begin{theorem}
\label{the:k_OHSDP}
    Given any $\varepsilon>0,\; \delta >0,\; n \in \mathbb{Z}_+, \;U \in \mathbb{Z}_{+}$, $\lambda = \log n \cdot \log (1/\delta)$, when there are $k$ corrupted users ($k \leq \hat{k} < \log n \cdot \log (1/\delta) / 2$), for any union-preserving query $Q$, the OHSDP protocol achieves the following:
        \begin{itemize}
             \item The messages received by the analyzer satisfy $(\varepsilon,\delta)$-DP;
             \item With probability at least $1-\beta$, the total error is bounded by 
            \[O\Big(k\log n \cdot \err_{\ell_p}\big(\mathcal{P}_Q, \frac{\varepsilon}{\log n}, \frac{\delta}{\log n},  \frac{\beta}{n}\big)\Big) +  \gamma_{\ell_p}\big(Q, k\cdot\log n \cdot \log (1/\delta)\big);\]
            \item In expectation, each user sends 
            \[
            O\left(\log n \cdot \mathrm{Msg}\big(\mathcal{P}_Q, \frac{\varepsilon}{\log n}, \frac{\delta}{\log n},\log n \cdot \log (1/\delta)\big)\right)
            \]
            messages, each containing $O\Big(\mathrm{Bit}\big(\mathcal{P}_Q,\frac{\varepsilon}{\log n}, \frac{\delta}{\log n},U,\log n \cdot \log (1/\delta)\big)+\log n\Big)$ bits. 
         \end{itemize}  
\end{theorem}

\begin{proof}
\label{sec:Proof of Theorem 4.4}

For privacy, an increase in the number of corrupted users does not affect the privacy guarantees of the OHSDP protocol with privacy budgets scaled by $(c - \hat{k})/{c}$. Therefore, the protocol continues to satisfy $(\varepsilon, \delta)$-DP.

For utility, with probability at least $1- \beta$, the {message noise} for level $r$ falls within $\theta^{(r)}$ for $r \in \{1, 2, \dots, \log (n/\lambda) + 1\}$.
The upper bounds $B_1$ and $R^1_1$ for the level $1$ groups in the proof of Theorem~\ref{the:OHSDP} still hold in the presence of multiple corrupted users.

{If all corrupted users evade detection, the worst-case error occurs when the $k$ corrupted users are distributed across distinct groups at the same level. In this scenario, all $k$ corrupted users remain in separate groups until the final $\lceil \log k \rceil$ levels.
Moreover, since each group contains only one corrupted user, 
\begin{equation*}
     B_r = 4\cdot \theta^{(1)} +  4\cdot\theta^{(2)} + \dots 
         +4 \cdot \theta^{(r-1)} +  \theta^{(r)} +
        \gamma_{\ell_p}(Q,\lambda)
\end{equation*}
holds for level $r \in \{2, 3, \cdots, \log(n/\lambda) + 1 - \lceil \log k\rceil\}$. 
As a result, at level $l = \log(n/\lambda) + 1 - \lceil \log k\rceil$, there are $k$ groups each with an error upper bound of $B_l$, while the remaining groups have errors bounded by $\theta^{(l)}$.
Since each merge of two groups from level $r$ to level $r +1$ adds up the error of at most $2\theta^{(r)} + \theta^{(r+1)}$, with probability at least $1 -\beta$, the final error bound satisfies 
\begin{equation*}
\begin{aligned}
    B_{\log(n/\lambda) + 1} &\leq k \cdot B_l + 3k\cdot \theta^{(l)} + 
    \sum_{i = l+1}^{\log(n/\lambda)} 2^{l-i}k\cdot\theta^{(i)} + 
    \theta^{(\log(n/\lambda) + 1)} \\
    &\leq 4k \cdot \sum_{i = 1}^{\log (n/\lambda)} \theta^{(i)} +  \theta^{(\log(n/\lambda) + 1)} + k\cdot \gamma_{\ell_p}(Q,\lambda)
\end{aligned}
\end{equation*}}
{where each \(\theta^{(r)}\) denotes \(\mathrm{Error}_{\ell_p}(\mathcal{P}_Q, \allowbreak \varepsilon^{(r)}, \delta^{(r)}, \beta^{(r)})\), with privacy parameters allocated identically to those in the single-attacker setting (see the proof of Theorem~\ref{the:OHSDP}).}

{Additionally, we can show that the worst-case error bound, with probability at least $1 - \beta$, arises when all detection fails.
For any group at level $r > 1$, if detection and recovery occur, the post-recovery error bound is equal to the sum of the error bounds of its subgroups at level $r - 1$. In contrast, if detection fails, the error may increase by at most $2\theta^{(r)} + \theta^{(r+1)}$ beyond this base.
At {the bottom level}, we have $B_1 \geq R^1_1$, ensuring that the worst-case scenario remains valid at the lowest level.}

In summary, the error bound for the OHSDP protocol with $k$ corrupted users is
\begin{equation*}
\begin{aligned}
    B_{\log(n/\lambda) + 1} 
    &\leq 4k \cdot \sum_{i = 1}^{\log (n/\lambda)} \theta^{(i)} +  \theta^{(\log(n/\lambda) + 1)} + k\cdot \gamma_{\ell_p}(Q,\lambda) \\
    &= O\left(k\log n \cdot \err_{\ell_p}(\mathcal{P}_Q, \frac{\varepsilon}{\log n}, \frac{\delta}{\log n},  \frac{\beta}{n})\right) +  k \cdot \gamma_{\ell_p}\big(Q, \lambda\big) 
\end{aligned}
\end{equation*}
with probability at least $1 - \beta$. 

As adversarial behavior has no impact on the communication process of honest users, their cost remains the same as in the single-attacker scenario. 
Accordingly, the analysis proceeds as in the proof of Theorem~\ref{the:OHSDP}.

\end{proof}

\begin{remark}
    Our protocol can be easily extended to the case where $\hat{k}$ exceeds $\log n \cdot \log (1/\delta) / 2$ by setting $\lambda = 2\hat{k}$ feasibly. 
    The error bound becomes $O\big(k\log n \cdot \err_{\ell_p}(\mathcal{P}_Q, \frac{\varepsilon}{\log n}, \frac{\delta}{\log n}, \frac{\beta}{n}) \big)+ \gamma_{\ell_p}(Q, 2k\hat{k})$ with probability at least $1-\beta$, and in expectation, each user sends $O\big(\log n\cdot\text{Msg}(\mathcal{P}_Q, \frac{\varepsilon}{ \log n}, \frac{\delta}{\log n}, \hat{k})\big)$ messages, each containing $O\big(\mathrm{Bit}(\mathcal{P}_Q,\frac{\varepsilon}{\log n}, \frac{\delta}{\log n},U,\hat{k})+\log n\big)$ bits.
\end{remark}

\begin{example}[OHSDP for $Q_\text{count}$ in Multiple Corrupted Users Setting]
    We instantiate the OHSDP protocol using $\mathcal{P}_{Q_\text{count}}$~\cite{ghazi2021differentially}, which achieves an error of $ O\big( \frac{k\cdot\log^{2} n } {\varepsilon}\cdot \log \frac{n}{\beta} + k\cdot\log n \cdot \log (1/\delta) \big) $. 
    In expectation, each user sends $O(\log n) $ messages, with each message containing $O(\log n)$ bits.
\end{example}

\section{Application}
\label{sec:Application}

In this section, we apply our OHSDP framework in three union-preserving queries: summation, frequency estimation, and range counting, as defined in Section~\ref{sec:notation}, and further explore its extension to complex queries such as $k$-selection and OLAP.

\paragraph{Summation}
For the summation problem, we introduce and integrate two SOTA solutions respectively.
The first protocol, IKOS-based BBGN \cite{balle2020privatesum}, improved from the IKOS exact summation protocol \cite{ishai06cryptography}, achieves central-DP error of $O\big(\tfrac{U}{\varepsilon}\log\tfrac{1}{\beta}\big)$. Each of $n$ users is expected to send $O\big(1+\tfrac{\log U}{\log n}\big)$ messages of $O(\log U + \log n)$ bits.
By integrating it into the OHSDP framework, we have:
\begin{itemize}
\item Without attackers, the error remains $O\big(\tfrac{U}{\varepsilon}\log\tfrac{1}{\beta}\big)$;
\item With one corrupted user, the error becomes $O\big(\tfrac{U \log^{2}n}{\varepsilon}\log\tfrac{n}{\beta}\big)$;
\item The expected number of messages is $O(\log U \log \log n)$, each of $O(\log U + \log n)$ bits.
\end{itemize}

The second protocol, GKMPS \cite{ghazi2021differentially}, also achieves central-DP error $O\big(\tfrac{U}{\varepsilon}\log\tfrac{1}{\beta}\big)$, with each user sending $1 + O\big(\tfrac{U \log^2 U \log (U/\delta)}{\varepsilon n}\big)$ messages of $O(\log U)$ bits.
By integrating it into OHSDP, we have:
\begin{itemize}
\item The error remains $O\big(\tfrac{U}{\varepsilon}\log\tfrac{1}{\beta}\big)$ without attackers;
\item With one corrupted user, the error becomes $O\big(\tfrac{U \log^{2}n}{\varepsilon}\log\tfrac{n}{\beta}\big)$;
\item The expected number of messages is $O\big(\log n + \tfrac{U \log^2 U \log (U/\delta)}{\varepsilon \log (1/\delta)}\big)$, each of $O(\log U + \log n)$ bits.
\end{itemize}

\paragraph{Frequency estimation}

For the frequency estimation (a.k.a. histogram) problem, the SOTA solution LWY \cite{luo2022frequency} sends $O(1)$ messages of $O(\log U)$ bits and achieves an $\ell_{\infty}$-error of
$O\big(\tfrac{1}{\varepsilon}\sqrt{\log \tfrac{U}{\beta} \log \tfrac{1}{\delta}}\big)$.
By integrating it into the OHSDP framework, we have:
\begin{itemize}
\item Without attackers, the error remains $O\big(\tfrac{1}{\varepsilon}\sqrt{\log \tfrac{U}{\beta} \log \tfrac{1}{\delta}}\big)$;
\item With one corrupted user, the error becomes $O\big(\tfrac{\log^{2}n}{\varepsilon}\sqrt{\log \tfrac{U}{\beta} \log \tfrac{1}{\delta}}\big);$
\item The expected number of messages is $O(\log n)$, each of $O(\log U + \log n)$ bits.
\end{itemize}

\paragraph{Range counting}
Range counting can be reduced to frequency estimation by building a binary tree over the domain $[0,U]$, where each node represents a range and each level forms a frequency estimation problem.
The SOTA shuffle-DP protocol LYD+ \cite{rm2} achieves an $\ell_{\infty}$-error of
$O\big(\tfrac{\log^{1.5} U}{\varepsilon} \sqrt{\log \tfrac{U}{\beta} \log \tfrac{1}{\delta}}\big)$, with $O(\log U)$ messages per user, each of $O(\log U)$ bits.
By integrating it into the OHSDP framework, we have:
\begin{itemize}
\item Without attackers, the error remains $O\big(\tfrac{\log^{1.5} U}{\varepsilon} \sqrt{\log \tfrac{U}{\beta} \log \tfrac{1}{\delta}}\big)$;
\item With one corrupted user, it becomes $O\big(\tfrac{\log^{2} n \cdot \log^{1.5} U}{\varepsilon} \sqrt{\log \tfrac{U}{\beta} \log \tfrac{1}{\delta}}\big)$;
\item The expected number of messages is $O(\log U \log n)$, each of $O(\log U + \log n)$ bits.
\end{itemize}

Furthermore, $k$-selection queries (which return the $k$-th largest item in the dataset) can be answered via range counting~\cite{dong2023continual}. Our framework supports these queries under poisoning attacks with $\tilde{O}(1)$ rank error, but only of theoretical interest due to its large polylogarithmic terms.

\paragraph{OLAP queries} 
{OLAP queries typically involve five core operations: join, select, project, aggregate, and group-by. Among these, join and projection (distinct count) are not union-preserving and are therefore outside the scope of our framework. Queries that combine selection with sum or count aggregation can be reduced to standard sum or count queries, which our framework directly supports. Group-by queries can be interpreted as answering multiple queries, which can be handled using DP basic composition theory (see Lemma~\ref{lem:Sequential Composition}).}

\section{Experiments}

We conduct experiments on three union-preserving queries: bit counting ($Q_\text{count}$), summation ($Q_\text{sum}$), and frequency estimation ($Q_\text{hist}$), comparing our optimized hierarchical shuffle-DP protocols with the SOTA methods on both synthetic and real-world datasets.
Specifically, for $Q_\text{count}$, we compare against CSUZZ~\cite{cheu2019distributed}, GKMPS~\cite{ghazi2021differentially}, and BBGN~\cite{balle2020privatesum}\footnote{We use the IKOS-based solution only in our experiments.} with our defense framework built on top of GKMPS and BBGN, referred to as Ours+GKMPS and Ours+BBGN.
For $Q_\text{sum}$, we compare against GKMPS and BBGN with Ours+BBGN.
For $Q_\text{hist}$, we compare against LWY~\cite{luo2022frequency} and CZ~\cite{cheu2022differentially} with Ours+LWY.

\subsection{Setup}
\paragraph{Datasets} 
We incorporate a total of seven datasets in our experiments.  
The real-world datasets include two salary datasets from Kaggle, namely San Francisco Salary (SF-Sal) \cite{kaggle_sf_salaries_2014}, and Brazil Salary (BR-Sal) \cite{kaggle_brazil_salary_2020}. 
The remaining two are the Adult dataset (Adult)~\cite{kohavi1996scaling} from the U.S. Census, and the AOL dataset (AOL)~\cite{pass2006picture} containing real-world web search accesses.
The synthetic datasets are generated using Uniform distribution (Unif), Zipfian distribution\footnote{We sample $n$ values from the Zipfian distribution (probability mass function $f(x) \propto x^{-a}$) and take modulo $U$.}(Zipf) with $a=1.5$, and Gaussian distribution\footnote{We sample $n$ values from the Gaussian distribution (probability mass function $f(x) \propto \text{exp}(-\frac{(x-\mu)^2}{2\sigma^2})$) and round to integers from $0$ to $U$.}(Gauss) with $\mu=\sigma=U/5$.

\begin{itemize}
    \item $Q_\text{count}$: we use there real-world datasets: Adult with $n=2^{15}$, SF-Sal with $n=2^{17}$, and BR-Sal with $n=2^{20}$. We count the number of females in the Adult dataset, and the number of individuals earning above the average salary in the salary datasets.
    Synthetic datasets are generated from three distributions with $n=2^{24}$.
    \item $Q_\text{sum}$: we use the same there real-world datasets. On Adult, we sum the ages of individuals, with age values capped at $130$. On the salary datasets, we compute the total salary, with salary values capped at $U = 2.5 \times 10^5$. 
    Synthetic datasets are generated from three distributions with $n = 2^{24}$ and $U = 10^5$.
    \item $Q_\text{hist}$: three real-world datasets are used: AOL for website visit frequencies, SF-Sal and BR-Sal for salary frequencies. 
    Synthetic datasets are generated from three distributions, with both real and synthetic datasets set to $n = U = 2^{17}$.
\end{itemize}

\paragraph{Poisoning attacks}
For all queries, we randomly select $k$ users from the user population as corrupted users. By default, we set $k = 1$.  
For $Q_\text{count}$, each corrupted user sends $n$ messages with value $1$ to the shufflers.  
For $Q_\text{sum}$, each corrupted user sends $n$ messages with value $U$ to the shufflers.  
For $Q_\text{hist}$, each corrupted user sends $n$ messages with value $1$ for every bin in the range $[0, U]$.

\paragraph{{Experimental environment and parameters}}
\label{sec:parameters and environment}
{All experiments are conducted on a single Linux server with an Intel Xeon CPU @ 2.50GHz and $178$ GB memory, where we simulate the outputs of all users (with fresh randomness) and aggregate their results, following the convention of previous distributed DP settings~\cite{rm2,cormode2018marginal,li2024local,he2024common,kulkarni2019answering,li2020estimating,ren2022ldp,he2025robust,yu2025privrm,10.14778/3424573.3424576}.}
{This single-node simulation exactly reproduces the algorithmic process of a real distributed setting in terms of both accuracy and communication, as both depend solely on the generated messages rather than the physical distribution of computation. Moreover, as achieving high utility requires a large $n$, implementing true $n$-node distributed environment would be prohibitively expensive.}
We set the privacy budgets as $\varepsilon = 1$, $\delta = n^{-2}$, and confidence level $\beta = 0.1$ by default. While we set $\varepsilon = 4$ for $Q_\text{hist}$. We choose the optimal $\lambda$ for each setting in our protocols.
We use $\ell_1$-error for $Q_\text{count}$ and $Q_\text{sum}$, and $\ell_\infty$-error for $Q_\text{hist}$.
Each experiment is repeated $100$ times, and we report the average after discarding the top $10\%$ and bottom $10\%$ of the results.

\subsection{Experimental Results}

\subsubsection{Utility and communication}
\paragraph{Performance comparison with SOTA protocols}

We evaluate the performances, including the relative error\footnote{For $Q_\text{hist}$, we define the relative error with respect to the data size $n$, i.e., \(\ell_\infty\) error divided by $n$.}, the average number of messages per user, and the number of bits per message, under with and without attack settings to the three queries on both synthetic and real-world datasets, shown in Table~\ref{tab:synthetic_all} and Table~\ref{tab:realdata_all}, respectively.

The results show a clear superiority of our defense frameworks in terms of utility under poisoning attacks (see Relative Error(w/ atk) column).
For $Q_\text{count}$ and $Q_\text{sum}$, all SOTA methods fail to return a reasonable result, i.e., the relative error is larger than 100\%. It shows the damage of one corrupted user conducting poisoning attacks.
But luckily, our frameworks successfully detect and mitigate the attacks, recovering results with relative errors well below 1\%, thereby preserving high utility.
The only exception is CZ for $Q_\text{hist}$, which is the only related protocol considering poisoning attacks. 
However, it sends more than 16KB\footnote{We do not include the extra cost of blind signature~\cite{chaum1983blind}. The 16KB message cannot defend against poisoning attacks.} per message, making it only practical for small domain settings.
When there is no attack, our frameworks consistently achieve only $2 \times$ larger error than the base protocols.
Comparing the error with and without attack for our protocols, the error with attack is no more than $225 \times$ larger than that without attack, which is roughly a $(\log n)^2 \times$ increase.
The gap is much smaller (only $10 \times$) for $Q_{\text{hist}}$ because the users in the drop out group split their contribution to different counters and reduce the error.
The distribution of dataset does not affect the error without attack. But when there is an attacker, our defense frameworks will drop out the entire group of size $\lambda$ containing corrupted user, and this leads to different errors in different distributions.

\begin{table}[t]
\centering
\renewcommand{\arraystretch}{1.2}
\resizebox{\textwidth}{!}{
\begin{tabular}{c|c||c|c|c||c|c|c||c|c|c||c}
\hline
\multirow{3}{*}{Query} & \multirow{3}{*}{Protocol}
  & \multicolumn{3}{c||}{Unif}
  & \multicolumn{3}{c||}{Zipf}
  & \multicolumn{3}{c||}{Gauss}
  & \multirow{3}{*}{\#Bits} \\
\cline{3-11}
& &
  \multicolumn{2}{c|}{Relative Error (\%)} & \multirow{2}{*}{\#Msgs}
  & \multicolumn{2}{c|}{Relative Error (\%)} & \multirow{2}{*}{\#Msgs}
  & \multicolumn{2}{c|}{Relative Error (\%)} & \multirow{2}{*}{\#Msgs}
  & \\
  \cline{3-4} \cline{6-7} \cline{9-10}
& &
  (w/o atk) & (w/atk) 
  & 
  & (w/o atk) & (w/atk)
  & 
  & (w/o atk) & (w/atk)
  & 
  & \\
\hline
\multirow{5}{*}{$Q_{\text{count}}$}
& CSUZZ      & $3.80 \times 10^{-5}$ & 200.05 & 1   
             & $3.01 \times 10^{-5}$ & 154.66 & 1   
             & $3.00 \times 10^{-4}$ & 1495.45 & 1   & 1 \\
& GKMPS      & $1.12 \times 10^{-5}$ & 200.05 & 0.50  
             & $1.09 \times 10^{-5}$ & 154.67 & 0.65  
             & $8.80 \times 10^{-5}$ & 1495.45 & 0.07 & 2 \\
& BBGN       & $1.04 \times 10^{-5}$ & 200.05 & 6    
             & $6.91 \times 10^{-6}$ & 154.67 & 6    
             & $6.69 \times 10^{-5}$ & 1495.45 & 6    & 28 \\
& Ours+GKMPS & $2.82 \times 10^{-5}$ & $3.25 \times 10^{-3}$ & 2704.97 
             & $2.27 \times 10^{-5}$ & $2.88 \times 10^{-3}$ & 2707.23 
             & $2.17 \times 10^{-4}$ & $1.53 \times 10^{-2}$ & 2697.92  & 18 \\
& Ours+BBGN  & $1.95 \times 10^{-5}$ & $3.38 \times 10^{-3}$ & 140  
             & $1.41 \times 10^{-5}$ & $3.03 \times 10^{-3}$ & 140  
             & $1.28 \times 10^{-4}$ & $1.13 \times 10^{-2}$ & 140   & 45 \\
\hline
\multirow{3}{*}{$Q_{\text{sum}}$}
& GKMPS     & $1.49 \times 10^{-5}$ & 199.98 & 2583.63 
            & $2.31 \times 10^{-3}$ & $2.84 \times 10^{4}$ & 2587.28 
            & $2.97 \times 10^{-5}$ & 461.51 & 2586.98 & 15 \\
& BBGN      & $9.70 \times 10^{-6}$ & 199.98 & 7 
            & $1.56 \times 10^{-3}$ & $2.84 \times 10^{4}$ & 7 
            & $2.08 \times 10^{-5}$ & 461.51 & 7 & 44 \\
& Ours+BBGN & $2.26 \times 10^{-5}$ & $3.16 \times 10^{-3}$ & 159 
            & $2.89 \times 10^{-3}$ & 0.23 & 159 
            & $4.00 \times 10^{-5}$ & $4.25 \times 10^{-3}$ & 159 & 61 \\
\hline
\multirow{3}{*}{$Q_\text{hist}$}
& CZ         & $5.84 \times 10^{-2}$ & $5.75 \times 10^{-2}$ & 3    
             & 0.12                  & 0.11                  & 3    
             & $5.50 \times 10^{-2}$ & $5.83 \times 10^{-2}$ & 3    & $2^{17^*}$ \\
& LWY        & $2.53 \times 10^{-2}$ & 100.05 & 66.88   
             & $2.58 \times 10^{-2}$ & 100.23 & 66.88   
             & $2.59 \times 10^{-2}$ & 99.87 & 66.88   & 54 \\
& Ours+LWY   & $5.02 \times 10^{-2}$ & 0.57 & 514.83  
             & $5.03 \times 10^{-2}$ & 0.64 & 514.83  
             & $5.02 \times 10^{-2}$ & 0.58 & 514.83  & 69 \\
\hline
\end{tabular}
}
\caption{Comparison of protocols on synthetic datasets. }
\label{tab:synthetic_all}
\end{table}

\begin{table}[t]
\centering
\renewcommand{\arraystretch}{1.2}
\resizebox{\textwidth}{!}{
\begin{tabular}{c|c||c|c|c|c||c|c|c|c||c|c|c|c}
\hline
\multirow{2}{*}{} & \multirow{2}{*}{} 
  & \multicolumn{2}{c|}{\makecell{Relative Error (\%)}} & \multirow{2}{*}{\#Msgs} & \multirow{2}{*}{\#Bits}
  & \multicolumn{2}{c|}{\makecell{Relative Error (\%)}} & \multirow{2}{*}{\#Msgs} & \multirow{2}{*}{\#Bits}
  & \multicolumn{2}{c|}{\makecell{Relative Error (\%)}} & \multirow{2}{*}{\#Msgs} & \multirow{2}{*}{\#Bits} \\
\cline{3-4} \cline{7-8} \cline{11-12}
& & (w/o atk) & (w/ atk) & & 
  & (w/o atk) & (w/ atk) & & 
  & (w/o atk) & (w/ atk) & & \\
\hline
\multirow{6}{*}{$Q_\text{count}$}
& Dataset
  & \multicolumn{4}{c||}{Adult}
  & \multicolumn{4}{c||}{SF\_Sal}
  & \multicolumn{4}{c}{BR\_Sal} \\
\cline{2-14}
& CSUZZ
  & $2.31\times10^{-2}$ & $302.55$   & $1$   & $1$
  & $3.71\times10^{-3}$ & $206.99$   & $1$   & $1$
  & $7.70\times10^{-4}$ & $279.83$   & $1$   & $1$ \\
& GKMPS
  & $1.13\times10^{-2}$ & $302.29$   & $0.74$   & $2$
  & $1.38\times10^{-3}$ & $206.96$   & $0.60$   & $2$
  & $2.96\times10^{-4}$ & $279.82$   & $0.36$   & $2$ \\
& BBGN
  & $7.61\times10^{-3}$ & $302.34$   & $9$   & $19$
  & $1.07\times10^{-3}$ & $206.96$   & $8$   & $21$
  & $2.05\times10^{-4}$ & $291.99$   & $7$   & $24$ \\
& Ours+GKMPS
  & $2.61\times10^{-2}$ & $0.84$     & $1638.72$ & $10$
  & $4.09\times10^{-3}$ & $0.11$     & $1051.23$ & $11$
  & $7.21\times10^{-4}$ & $0.04$     & $1679.21$ & $14$ \\
& Ours+BBGN
  & $1.36\times10^{-2}$ & $0.55$     & $109$ & $29$
  & $1.91\times10^{-3}$ & $0.15$     & $146$ & $34$
  & $4.91\times10^{-4}$ & $0.03$     & $147$ & $39$ \\
\hline
\multirow{4}{*}{$Q_\text{sum}$}
& Dataset
  & \multicolumn{4}{c||}{Adult}
  & \multicolumn{4}{c||}{SF\_Sal}
  & \multicolumn{4}{c}{BR\_Sal} \\
\cline{2-14}
& GKMPS
  & $1.27\times10^{-2}$ & $336.93$    & $3909.41$ & $9$
  & $3.36\times10^{-3}$ & $352.21$    & $11367.11$ & $11$
  & $8.22\times10^{-3}$  & $7369.82$   & $6053.09$  & $13$ \\
& BBGN
  & $9.36\times10^{-3}$ & $336.94$    & $9$   & $26$
  & $2.08\times10^{-3}$ & $352.21$    & $9$   & $39$
  & $6.55\times10^{-3}$  & $7369.82$   & $8$        & $42$ \\
& Ours+BBGN
  & $1.80\times10^{-2}$ & $0.51$      & $115$     & $36$
  & $4.48\times10^{-3}$ & $0.13$      & $125$     & $50$
  & $1.32\times10^{-2}$ & $0.41$      & $101$     & $53$ \\
\hline
\multirow{4}{*}{$Q_\text{hist}$}
& Dataset
  & \multicolumn{4}{c||}{AOL}
  & \multicolumn{4}{c||}{SF\_Sal}
  & \multicolumn{4}{c}{BR\_Sal} \\
\cline{2-14}
& CZ
  & $7.12\times10^{-2}$  & $7.10\times10^{-2}$   & $3$ & ${2^{17}}^{*}$
  & $5.81\times10^{-2}$  & $5.90\times10^{-2}$   & $3$ & ${2^{17}}^{*}$
  & $5.85\times10^{-2}$  & $5.80\times10^{-2}$   & $3$ & ${2^{17}}^{*}$ \\
& LWY
  & $2.57\times10^{-2}$  & $100.03$ & $66.88$ & $54$
  & $2.57\times10^{-2}$  & $100.02$ & $66.88$ & $54$
  & $2.57\times10^{-2}$  & $100.00$ & $66.88$ & $54$ \\
& Ours+LWY
  & $5.09\times10^{-2}$  & $0.58$  & $514.83$ & $69$
  & $5.08\times10^{-2}$  & $0.57$  & $514.83$ & $69$
  & $4.95\times10^{-2}$  & $0.58$  & $514.83$ & $69$ \\
\hline

\end{tabular}
}
\caption{Comparison of protocols on real-world datasets.}
\label{tab:realdata_all}
\end{table}

In terms of communication cost, our defense frameworks achieve approximately $(\log n) \times$ larger number of messages than the base protocols BBGN and LWY, which send $O(1)$ messages. 
Since GKMPS has a poor performance when $n$ is not larger than $U$ too much, the communication costs are very large for Ours+GKMPS in $Q_\text{count}$ and GKMPS in $Q_\text{sum}$, which matches our analysis in Table~\ref{tab:contribution}.
Our framework sends $O(\log n)$ more bits in each message.
The distribution of dataset does not affect the performance.

\begin{table}[t]
    \centering
    \renewcommand{\arraystretch}{1.2} 
    \resizebox{\textwidth}{!}{ 
    {\begin{tabular}{c||cccc|cccc||cccc} 
        \toprule
        & \multicolumn{4}{c|}{Error (w/o attacker)}
        & \multicolumn{4}{c||}{Error (w/ attacker)}
        & \multicolumn{4}{c}{\#Msgs per user} \\
        \midrule
        $n$
        & $2^{12}$ & $2^{16}$ & $2^{20}$ & $2^{24}$
        & $2^{12}$ & $2^{16}$ & $2^{20}$ & $2^{24}$
        & $2^{12}$ & $2^{16}$ & $2^{20}$ & $2^{24}$ \\
        \midrule \midrule
        SUSDP+GKMPS
        & 89.09 & 391.73 & 1605.68 & 6415.31
        & 94.34 & 341.90 & 1461.25 & 6826.24
        & 10895.32 & 14184.59 & 17477.93 & 21256.80 \\
        BSDP+GKMPS
        & 3.41 & 3.63 & 3.79 & 2.85
        & 43.71 & 87.43 & 172.41 & 402.44
        & 46646.16 & 59943.92 & 73251.49 & 86559.85 \\
        OHSDP+GKMPS
        & 2.38 & 2.34 & 2.56 & 2.40
        & 52.00 & 116.7 & 160.67 & 197.03
        & 6134.05 & 13274.97 & 10484.93 & 7581.61 \\
        SUSDP+BBGN
        & 65.81 & 233.60 & 1122.93 & 4823.43
        & 68.18 & 225.80 & 963.40 & 4915.28
        & 3 & 3 & 3 & 3 \\
        BSDP+BBGN
        & 2.24 & 2.33 & 2.34 & 2.46
        & 30.36 & 72.56 & 143.95 & 261.18
        & 35 & 26 & 22 & 20 \\
        OHSDP+BBGN
        & 1.61 & 1.50 & 1.75 & 1.64
        & 41.40 & 77.56 & 104.43 & 125.60
        & 102 & 117 & 147 & 173\\
        \bottomrule
    \end{tabular}}
    }
    \caption{{Comparison of SUSDP, BSDP, and OHSDP protocols on uniform synthetic datasets varying data size $n$.}}
    \label{tab:baseline_experiments}
\end{table}

\paragraph{Performance comparison with SUSDP and BSDP protocols.}
{We evaluate the performance of the strawman single-user shuffle-DP (SUSDP) and the block-based shuffle-DP (BSDP) protocols, each instantiated with the BBGN and GKMPS protocols, on the bit-counting query $Q_{\text{count}}$. 
The comparison is conducted on Unif under varying data sizes $n \in \{2^{12}, 2^{16}, 2^{20}, 2^{24}\}$, alongside our proposed hierarchical protocol (OHSDP). 
The results are consistent with our theoretical analysis and validate the predicted asymptotic behaviors: 
the error of SUSDP grows at a rate of $O(\sqrt{n})$ (Theorem~\ref{theorem: Strawman}), 
that of BSDP scales approximately as $O(n^\frac{1}{4})$ (Theorem~\ref{the: BSDP theorem}), 
while the error of OHSDP increases only logarithmically, following $O(\log^2 n)$ (Theorem~\ref{the:HSDP theorem}). 
These results demonstrate that, OHSDP achieves better utility compared with both baselines at large data size $n$. }

\begin{figure}[t]
    \centering
        \centering
    \includegraphics[width=\linewidth]{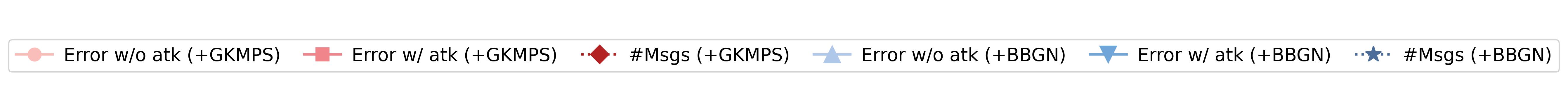} 
    \begin{subfigure}[t]{0.247\textwidth}
        \centering
        \includegraphics[width=\linewidth]{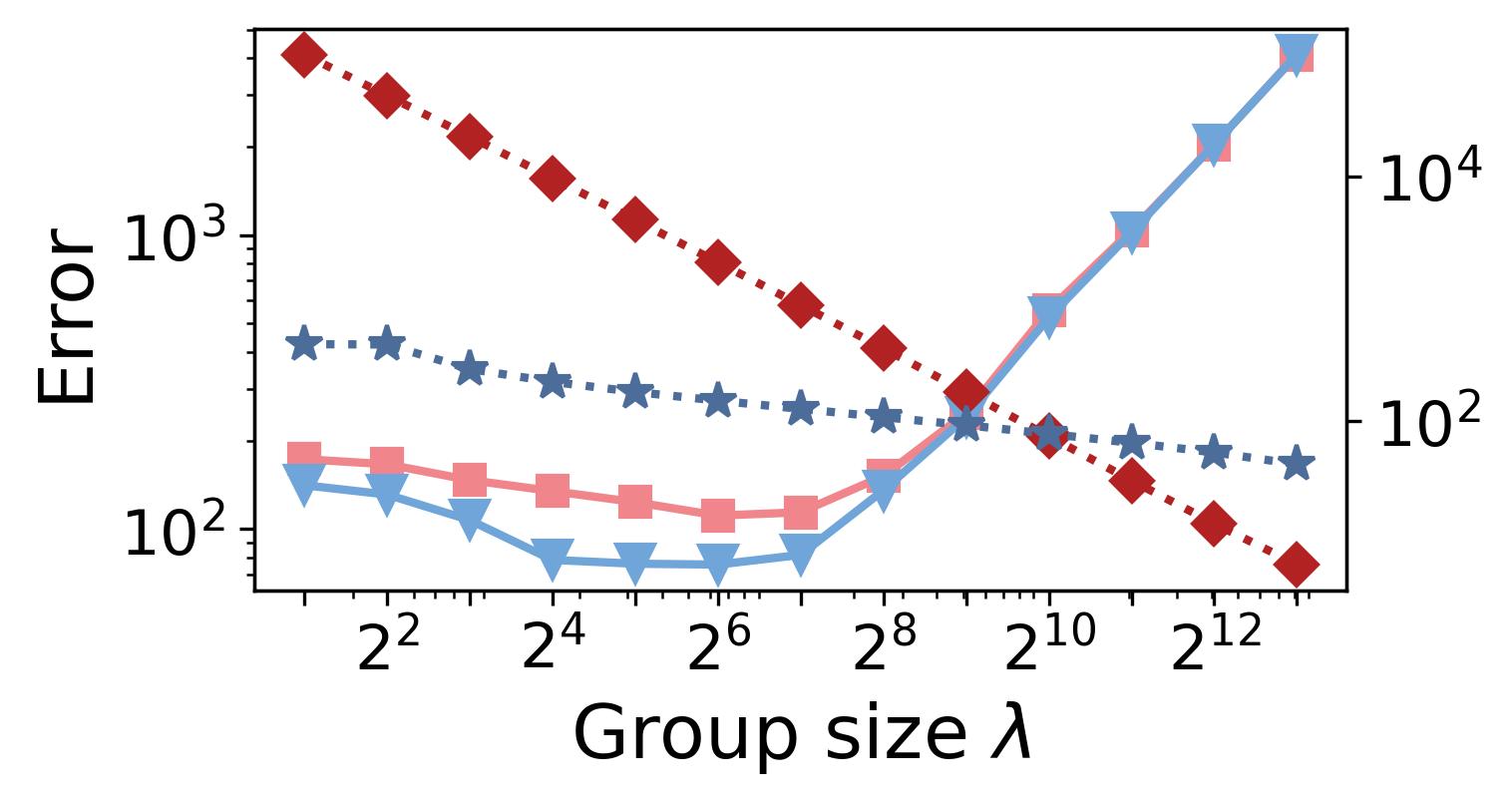}
        \caption{Varying group size $\lambda$}\label{fig:paras-a}
    \end{subfigure}
    \hfill
    \begin{subfigure}[t]{0.239\textwidth}
        \centering
        \includegraphics[width=\linewidth]{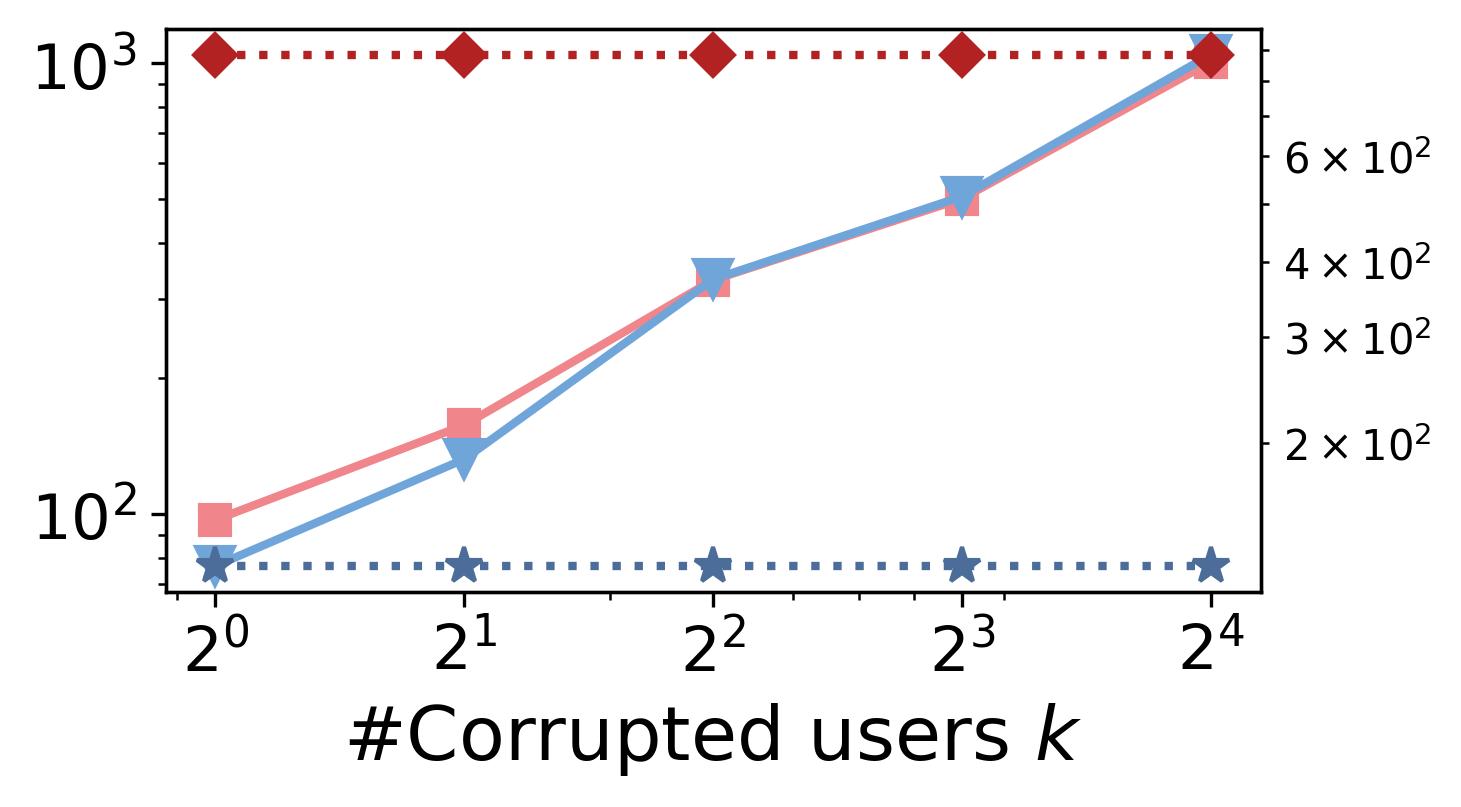}
        \caption{Varying \#corrupted users $k$}\label{fig:paras-b}
    \end{subfigure}
    \hfill
    \begin{subfigure}[t]{0.235\textwidth}
        \centering
        \includegraphics[width=\linewidth]{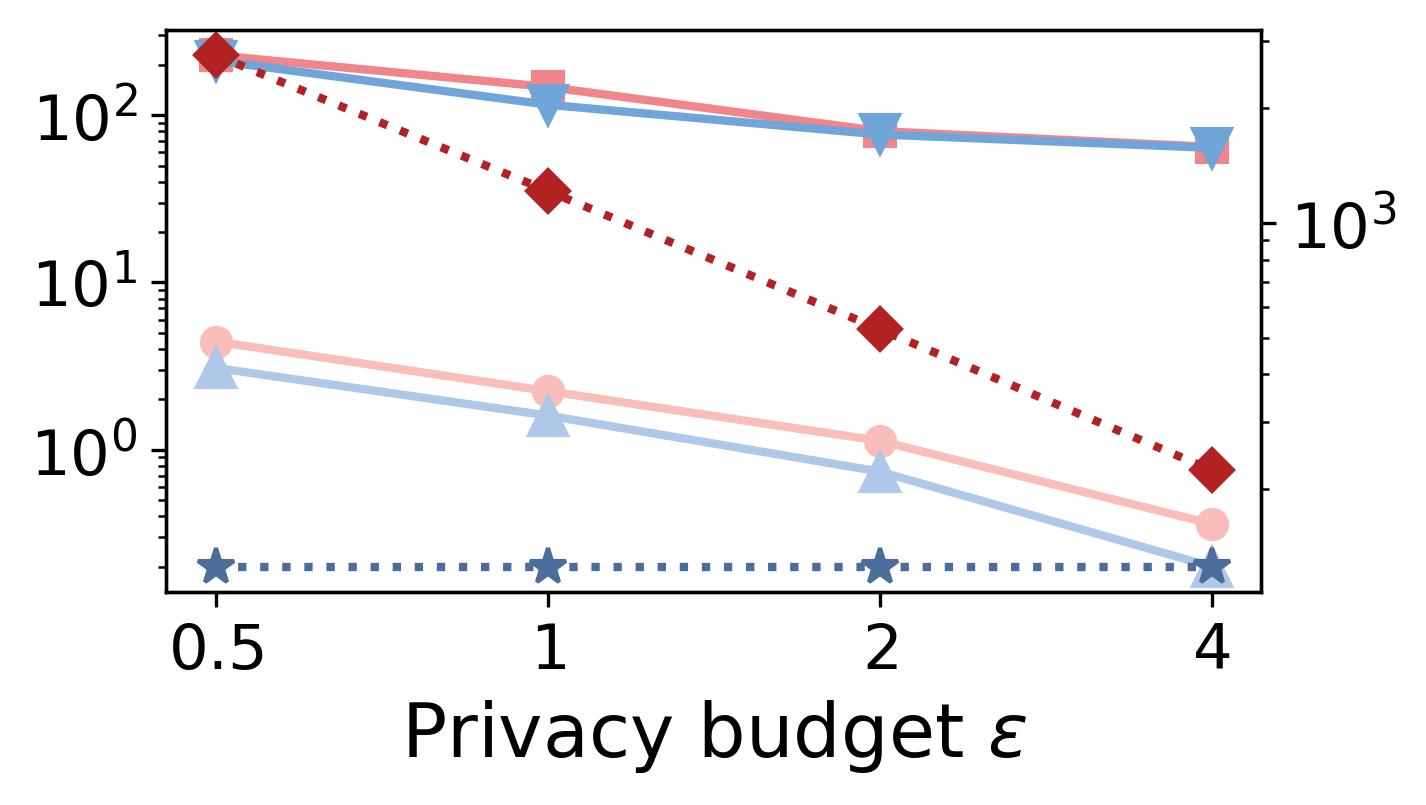}
        \caption{Varying privacy budget $\varepsilon$}\label{fig:paras-c}
    \end{subfigure}
    \hfill
    \begin{subfigure}[t]{0.25\textwidth}
        \centering
        \includegraphics[width=\linewidth]{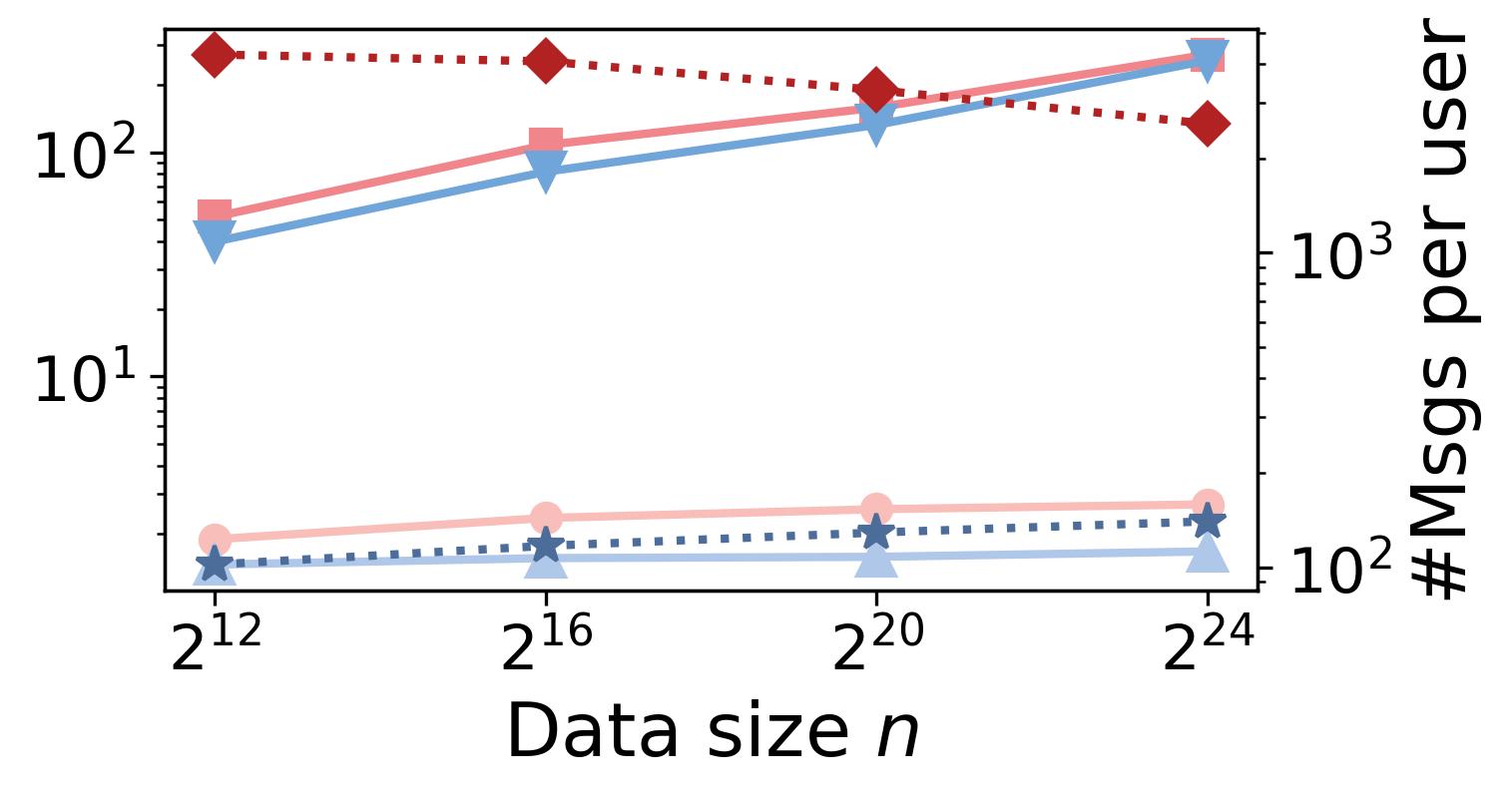}
        \caption{Varying data size $n$}\label{fig:paras-d}
    \end{subfigure}
    \caption{Comparison of Ours+GKMPS and Ours+BBGN for $Q_\mathrm{count}$ on uniform distribution varying different.}
    \label{fig:paras}
\end{figure}

\begin{figure*}[t]
    \raggedright
    \includegraphics[width=0.6\linewidth]{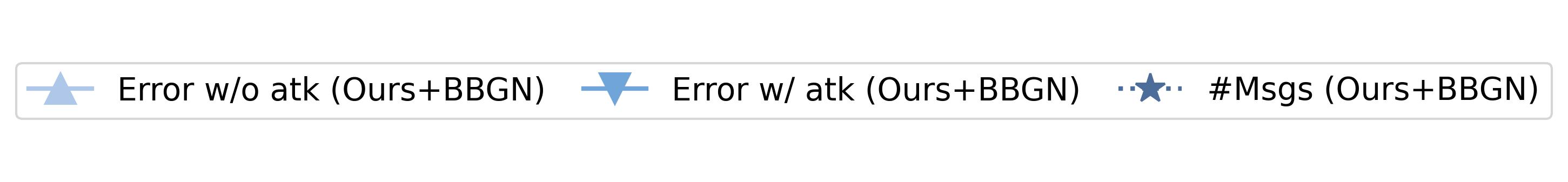} 
    
    \begin{subfigure}[t]{0.247\textwidth}
        \centering
        \includegraphics[width=\linewidth]{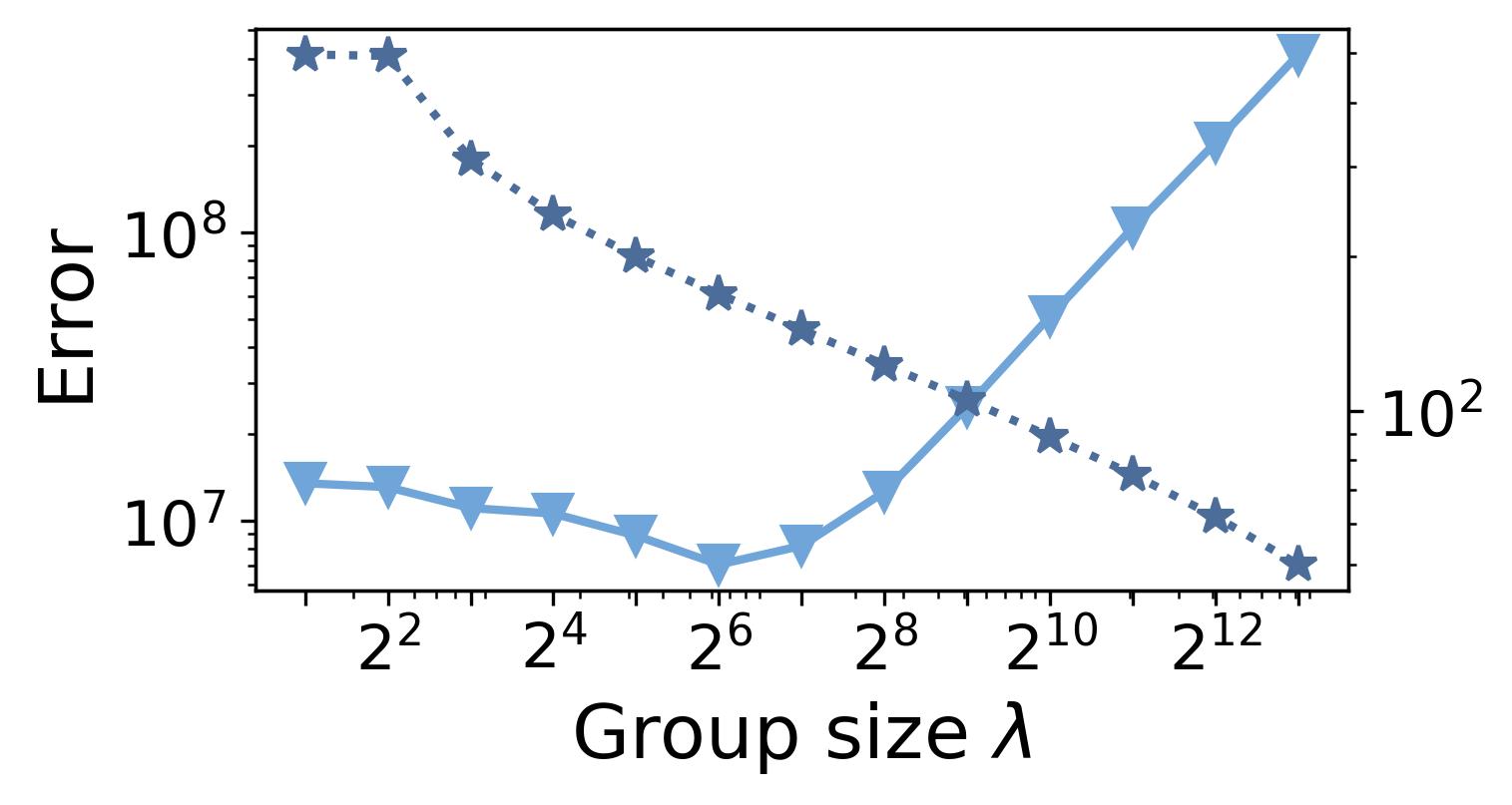}
        \caption{Varying group size $\lambda$}
    \end{subfigure}
    \hfill
    \begin{subfigure}[t]{0.233\textwidth}
        \centering
        \includegraphics[width=\linewidth]{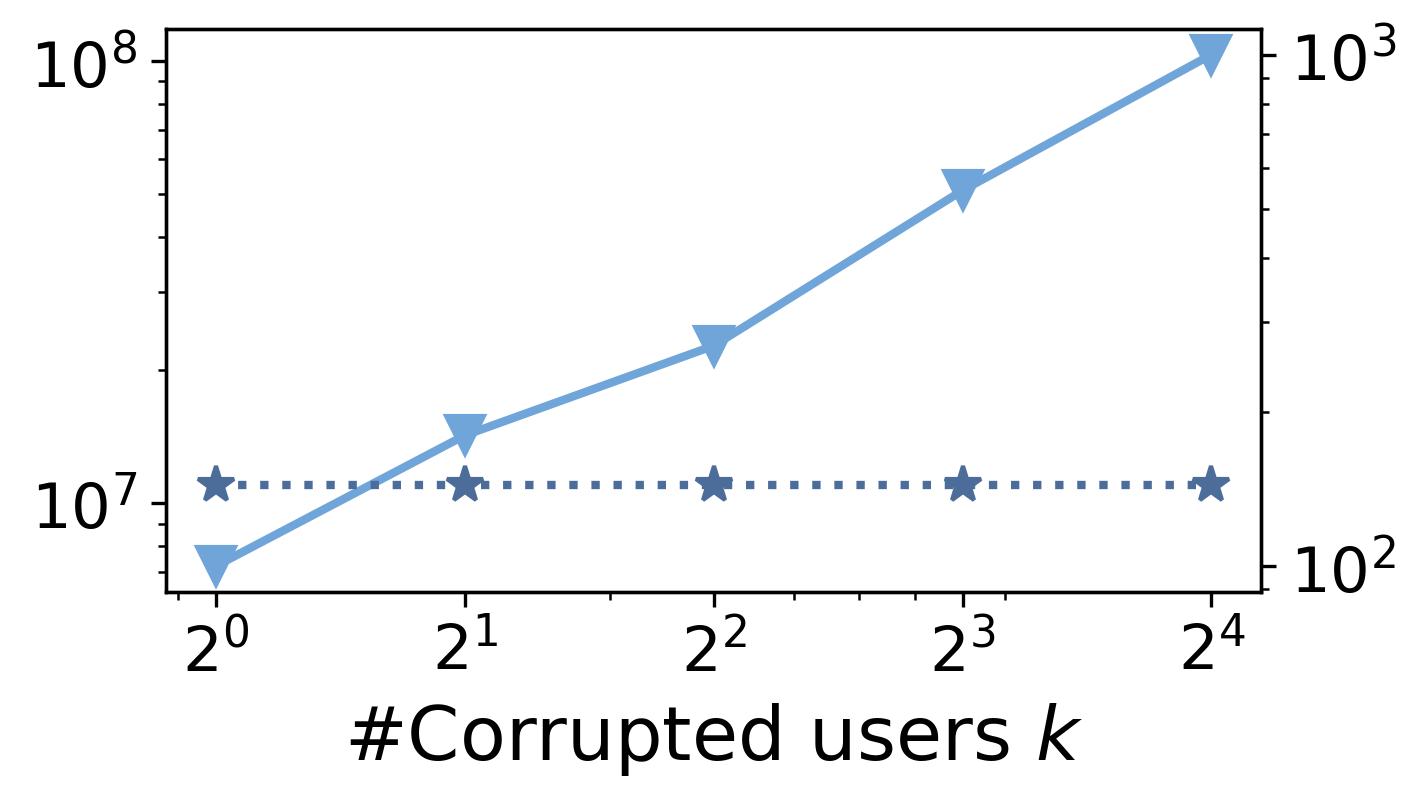}
        \caption{Varying \#corrupted users $k$}
    \end{subfigure}
    \hfill
    \begin{subfigure}[t]{0.233\textwidth}
        \centering
        \includegraphics[width=\linewidth]{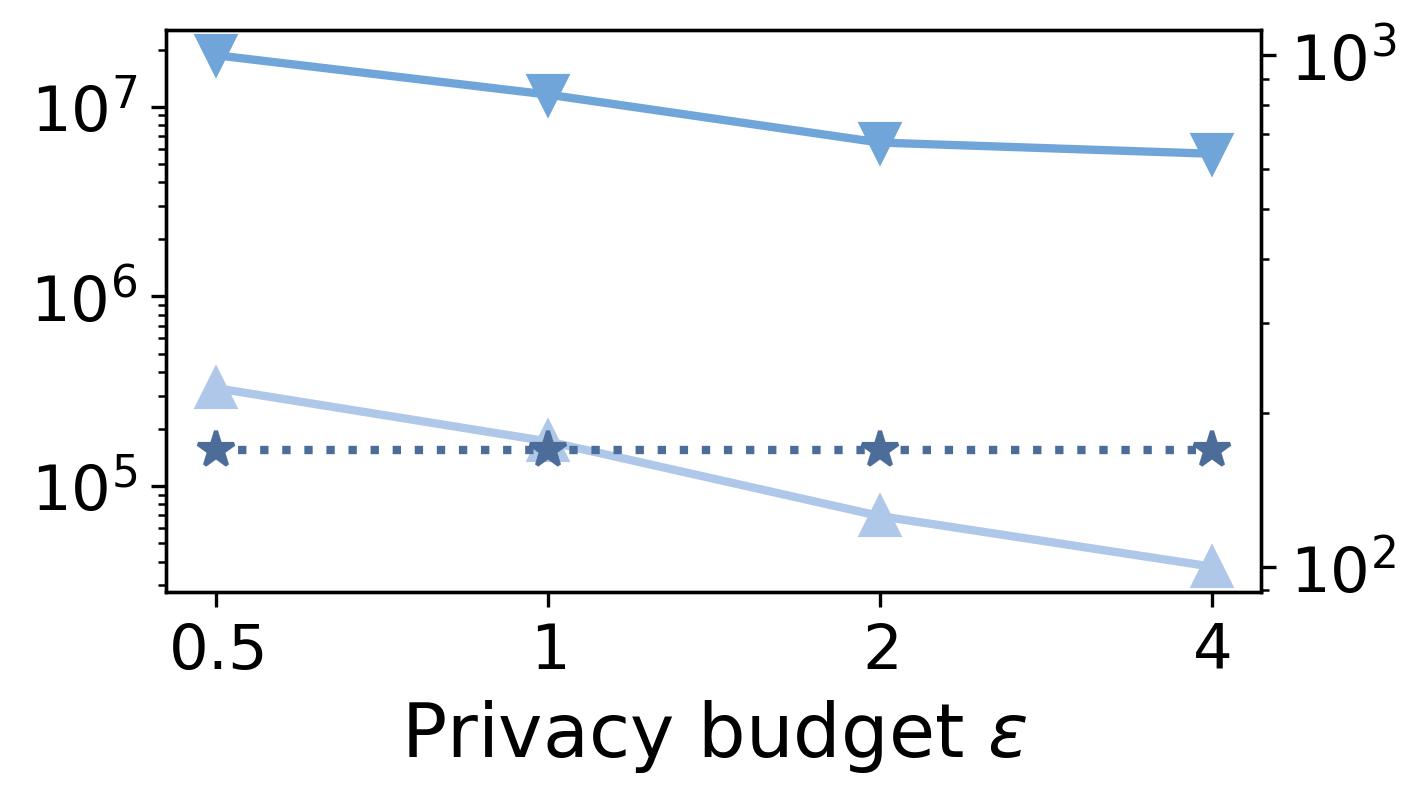}
        \caption{Varying privacy budget $\varepsilon$}
    \end{subfigure}
    \hfill
    \begin{subfigure}[t]{0.261\textwidth}
        \centering
        \includegraphics[width=\linewidth]{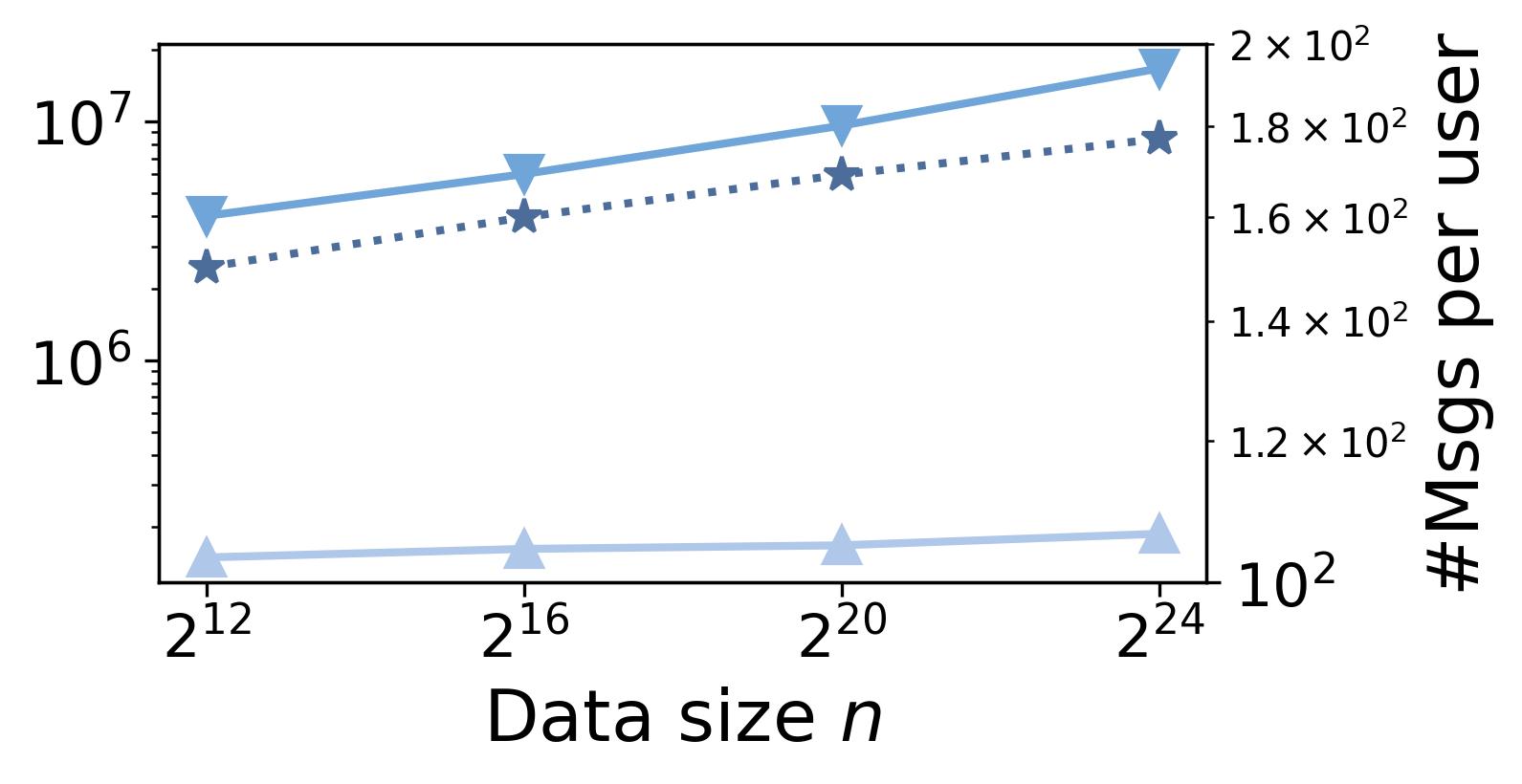}
        \caption{Varying data size $n$}
    \end{subfigure}
    \caption{Comparison of Ours+BBGN for $Q_\mathrm{sum}$ on uniform distribution varying different parameters.}
    \label{fig:paras_sum}
\end{figure*}

\begin{figure*}[t]
    
        \raggedright
    \includegraphics[width=0.6\linewidth]{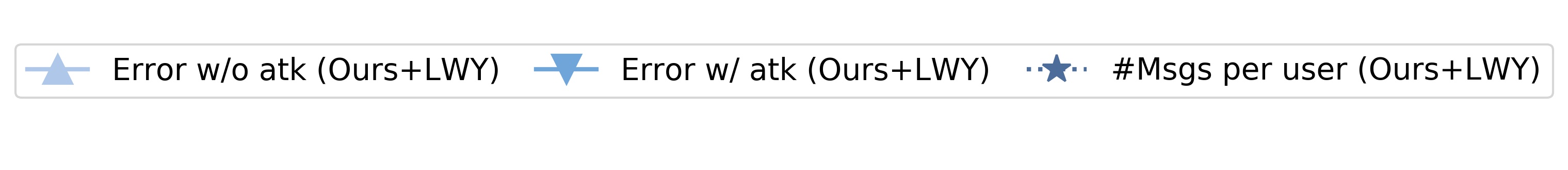} 
    
    \begin{subfigure}[t]{0.247\textwidth}
        \centering
        \includegraphics[width=\linewidth]{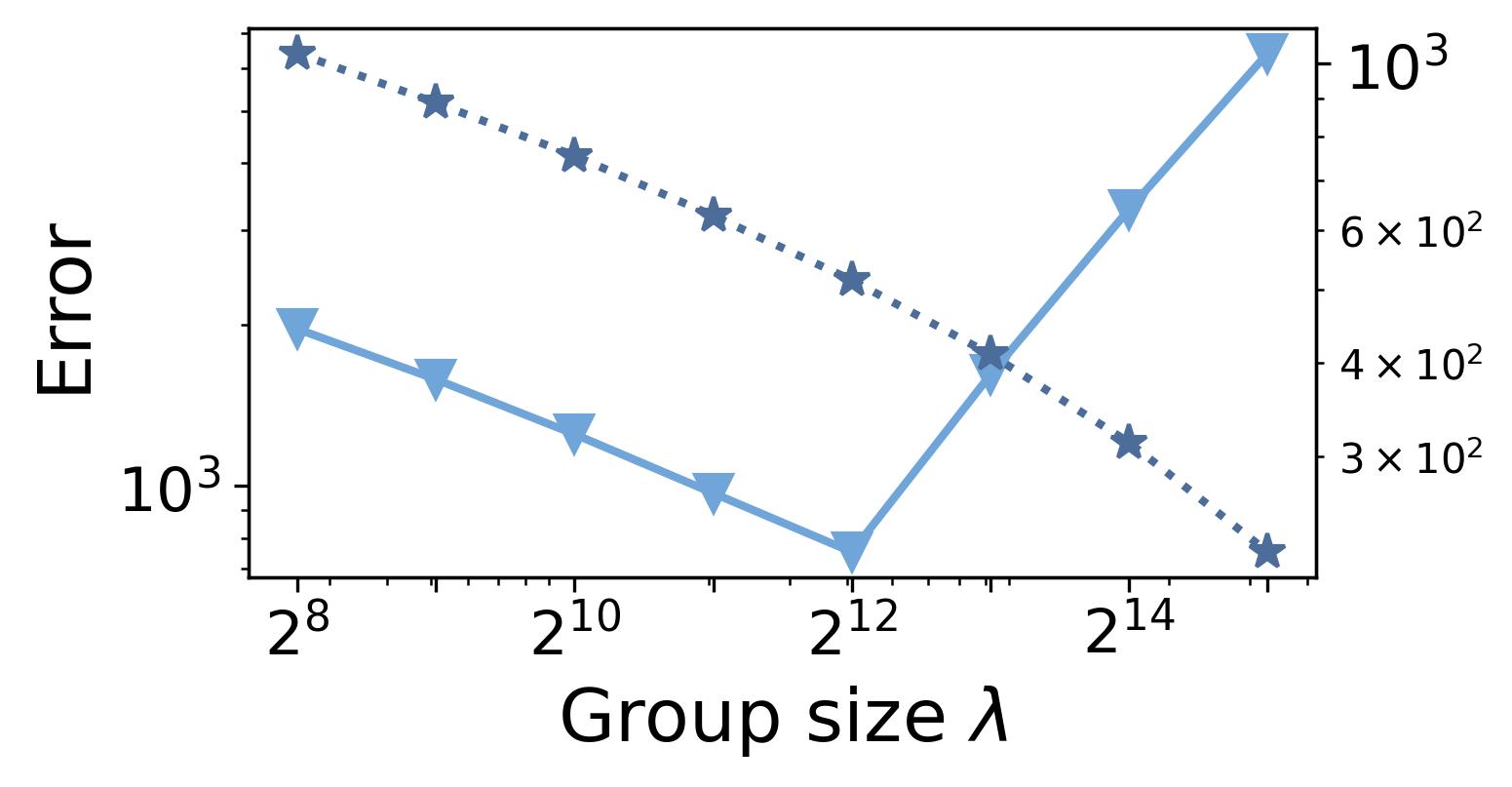}
        \caption{Varying group size $\lambda$}
    \end{subfigure}
    \hfill
    \begin{subfigure}[t]{0.233\textwidth}
        \centering
        \includegraphics[width=\linewidth]{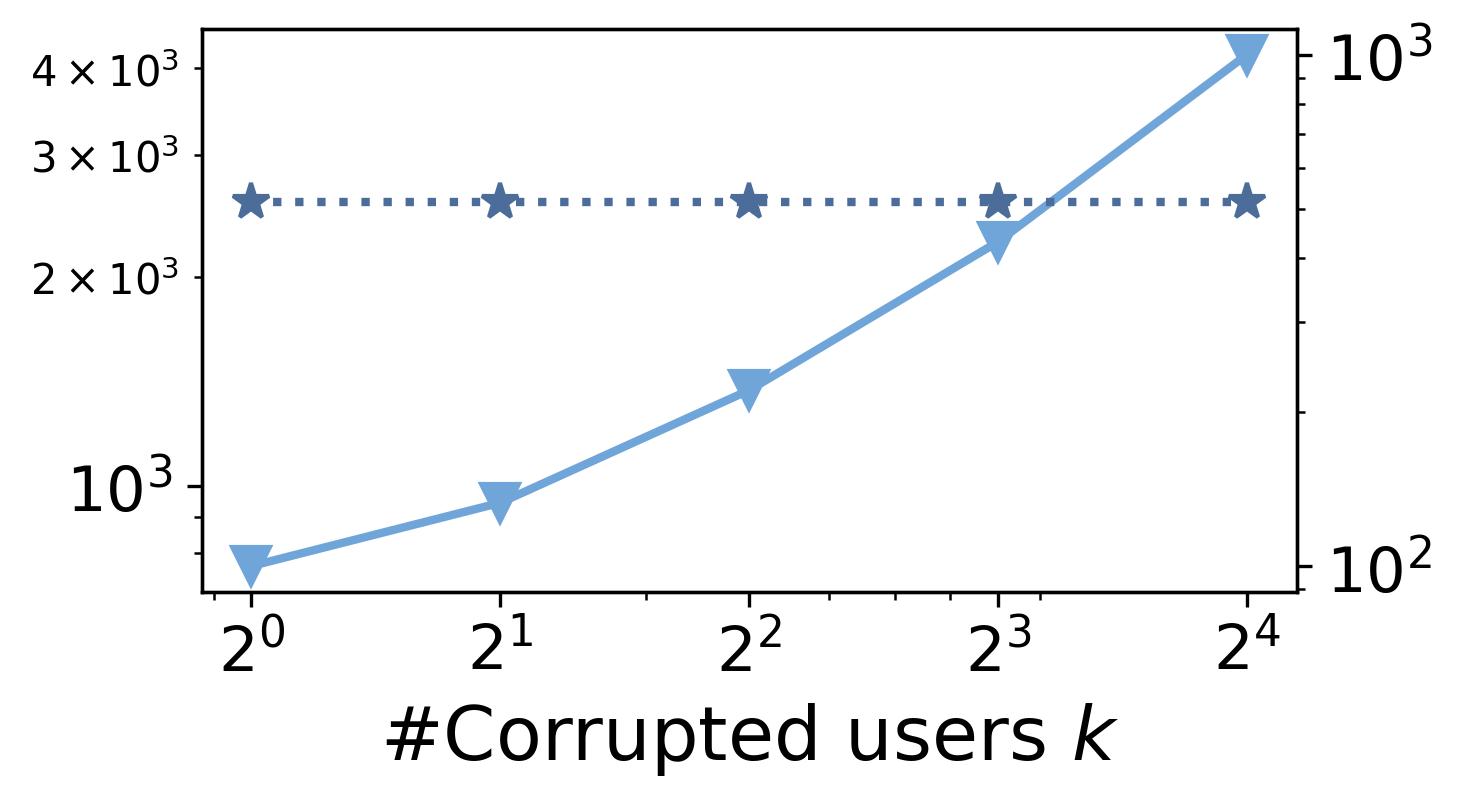}
        \caption{Varying \#corrupted users $k$}
    \end{subfigure}
    \hfill
    \begin{subfigure}[t]{0.244\textwidth}
        \centering
        \includegraphics[width=\linewidth]{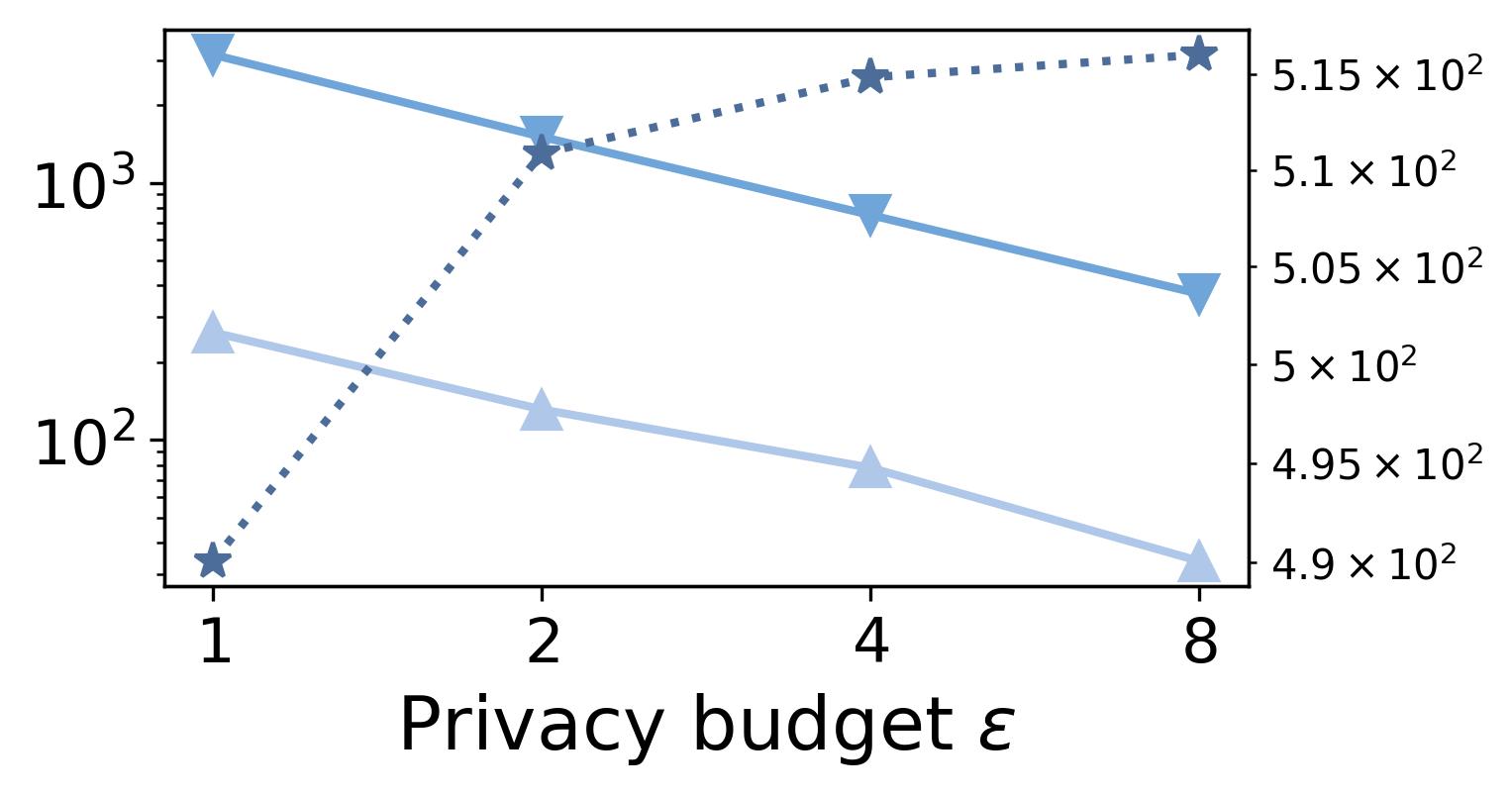}
        \caption{Varying privacy budget $\varepsilon$}
    \end{subfigure}
    \hfill
    \begin{subfigure}[t]{0.249\textwidth}
        \centering
        \includegraphics[width=\linewidth]{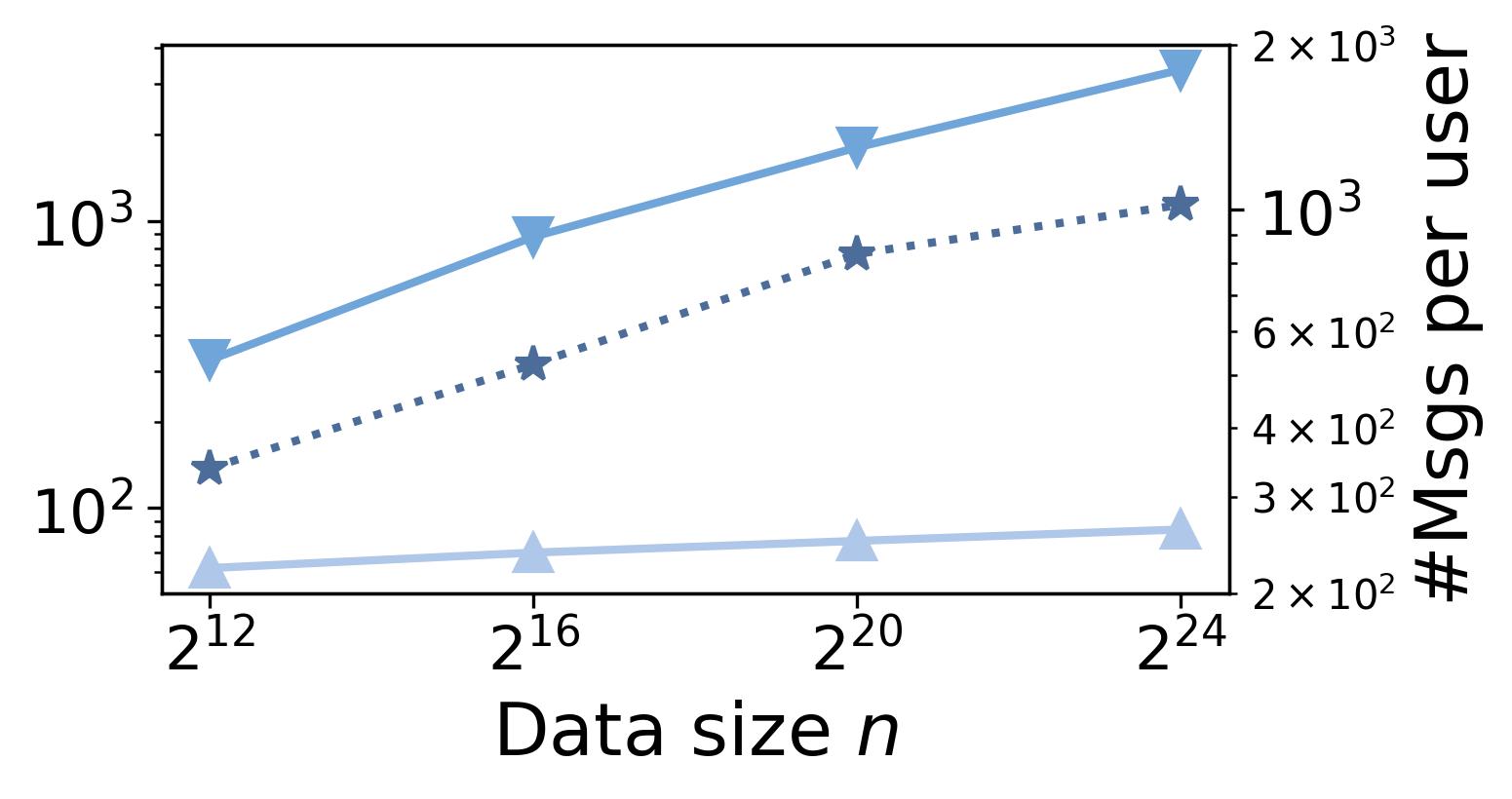}
        \caption{Varying data size $n$}
    \end{subfigure}
    \caption{Comparison of Ours+LWY for $Q_\mathrm{hist}$ on zipfian distribution varying different parameters.}
    \label{fig:paras_hist}
\end{figure*}

\subsubsection{Parameter experiments}
We further evaluate the performance of our defense frameworks Ours+GKMPS and Ours+BBGN for the three queries varying parameters, including the group size $\lambda$, the number of corrupted users $k$, the privacy budget $\varepsilon$, and the data size $n$. The results for $Q_\text{count}$, $Q_\text{sum}$, and $Q_\text{hist}$ are presented in Fig.~\ref{fig:paras}, Fig.~\ref{fig:paras_sum}, and Fig.~\ref{fig:paras_hist}, respectively. Given that the performance trends are consistent across all queries, we focus our detailed analysis on $Q_\text{count}$ as a representative example in the following discussions.

\paragraph{Group size $\lambda$.}
Fig.~\ref{fig:paras-a} shows the error with attack and messages per user. We do not include the error without attack because it uses separate privacy budgets and keeps the same.
The results match our expectations:
The error with attack first decreases and then increases with $\lambda$ growth.
This is because the error consists of two parts: the accumulated error from the $O(\log n/\lambda)$ counters and the ignored group of size $\lambda$.
The first part decreases with $\lambda$ growth because the number of counters decreases and the privacy budgets per counter increase.
The second part increases with $\lambda$ growth.
The communication cost decreases for both protocols as $\lambda$ increases. However, Ours+GKMPS drops faster, since its cost is proportional to $1/\lambda$, whereas Ours+BBGN scales with $1/\log \lambda$.

The results demonstrate that an optimal $\lambda$ will help to improve both the utility and efficiency, i.e., the lowest point of the line of error. 
And further increasing the $\lambda$ leads to a trade-off between utility and efficiency.

\paragraph{Number of corrupted users $k$.} 
Fig.~\ref{fig:paras-b} shows that the error with attack inevitably grows linearly with the number of corrupted users $k$. However, the number of messages per user remains constant. This observation aligns with our theoretical analysis of multiple corrupted users (Theorem~\ref{the:k_OHSDP}).

\paragraph{privacy budget $\varepsilon$.}
As shown in Fig.~\ref{fig:paras-c}, both protocols have lower error, whether with or without attack, when $\varepsilon$ grows. The communication cost of Ours+GKMPS also decreases, but Ours+BBGN keeps the same. Because the latter one has a fixed number of messages, whatever $\varepsilon$ is.

\paragraph{Data size $n$.}
The results are in Fig.~\ref{fig:paras-d}. Both protocols have larger errors with $n$ growth, as it is proportional to polylogarithmic of $n$. The communication cost changes differently and shows an interesting phenomenon of different protocols: Ours+BBGN increases as larger $n$ leads to more levels in the structure, while the number of messages in each level is stable. Its communication cost is dominated by the number of levels. Ours+GKMPS chooses a larger $\lambda$ for larger $n$, which reduces the number of messages in the bottom level, which dominates the total number of messages.

\begin{figure}[t]
\raggedright
    \includegraphics[width=0.35\linewidth]{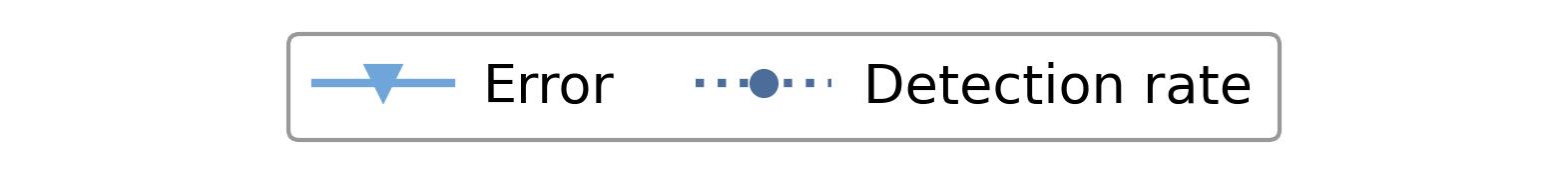} 

    \begin{subfigure}{.4\textwidth} 
        \centering
        \includegraphics[width=\textwidth]{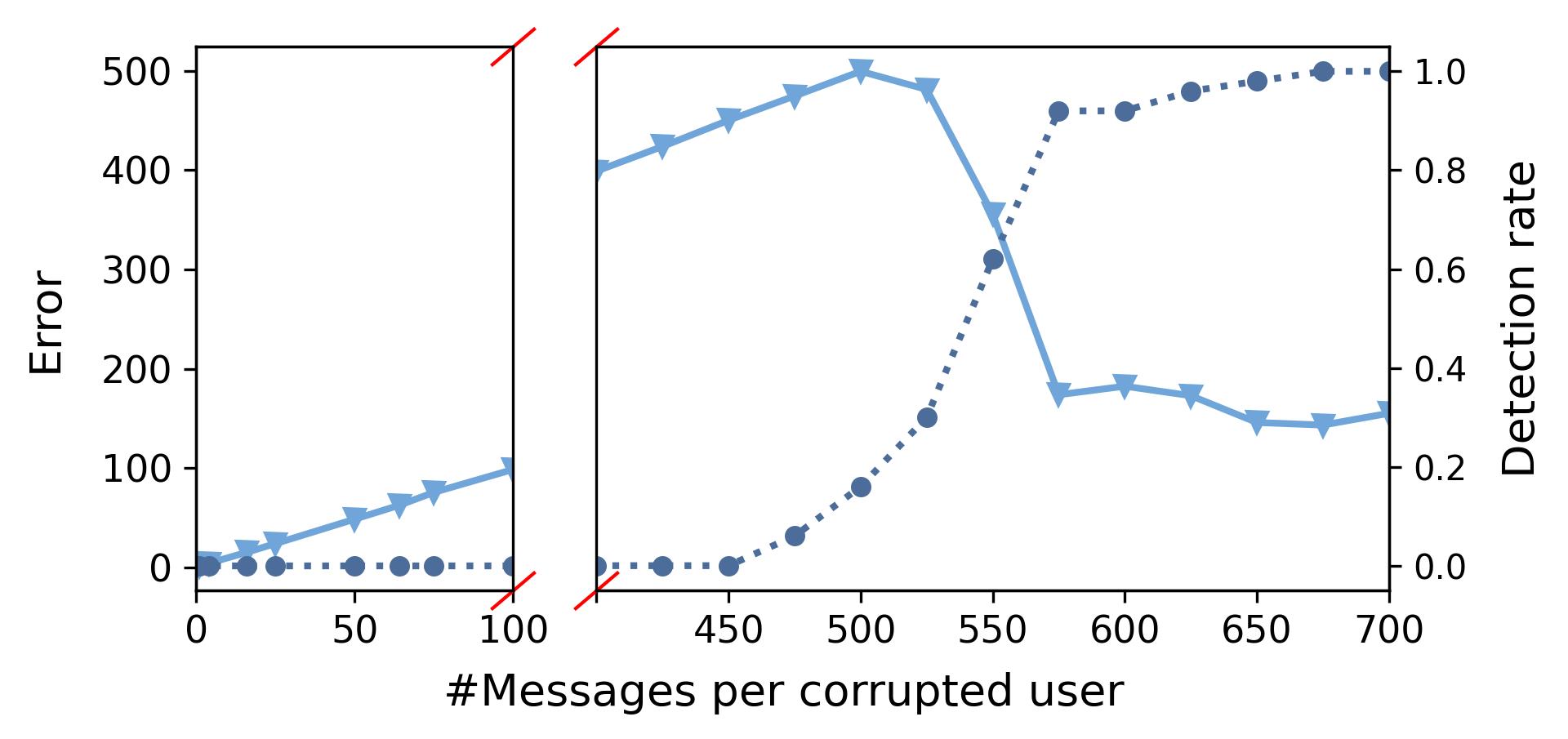} 
        \caption{Ours + GKMPS} 
    \end{subfigure}
    % \hfill
    \begin{subfigure}{.4\textwidth} 
        \centering
        \includegraphics[width=\textwidth]{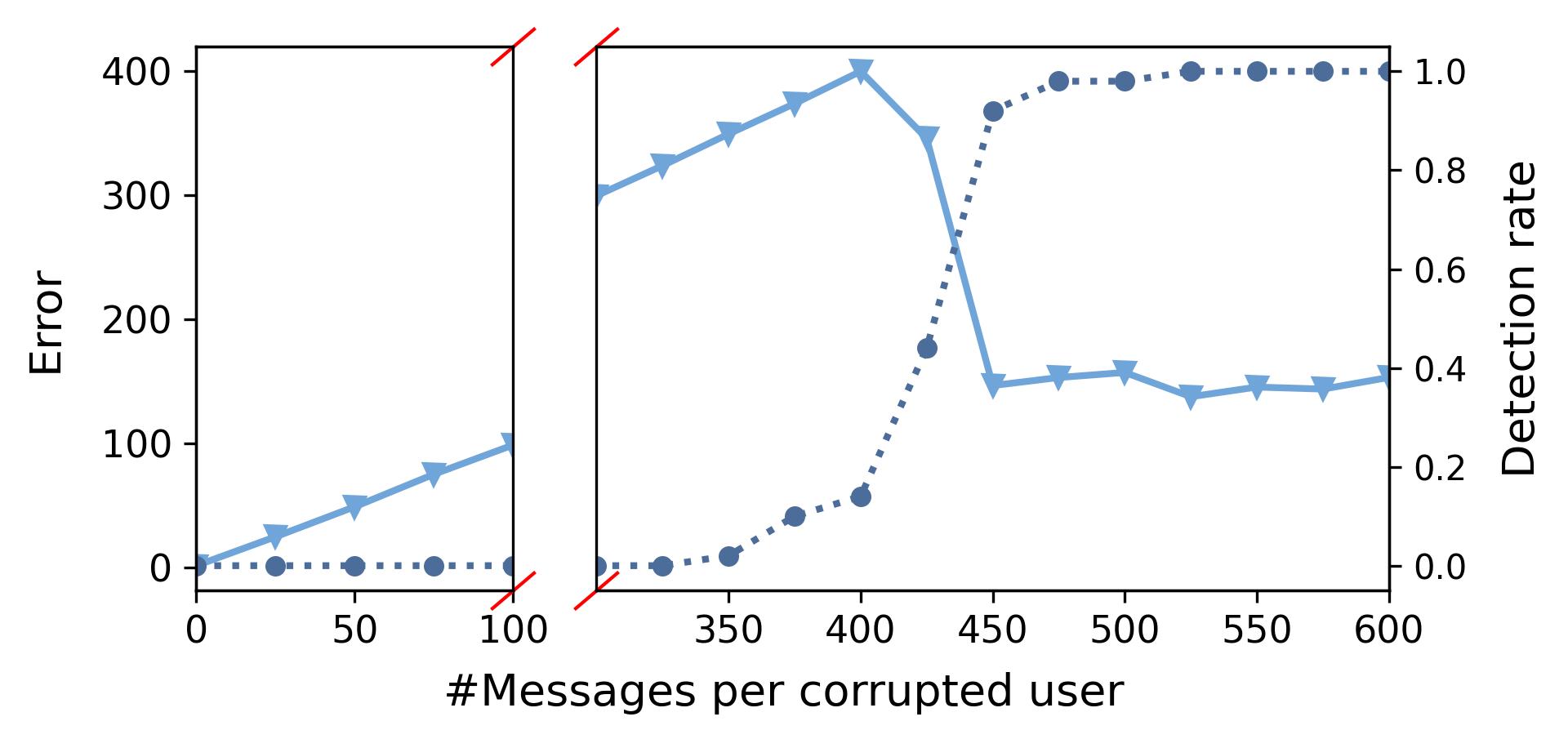}
        \caption{Ours + BBGN} 
    \end{subfigure}
    
    \caption{Comparison of Ours+GKMPS and Ours+BBGN for $Q_\text{count}$ on uniform distribution varying \#messages the corrupted user send. }
    \label{fig:atk_msg}
\end{figure}

\paragraph{Effect of the number of malicious messages sent by corrupted users.}
To empirically validate our theoretical analysis that the worst-case error occurs when corrupted users evade detection, we conduct an experiment by varying the number of malicious messages each corrupted user sends.
The experiment is performed on the bit counting query $Q_{\text{count}}$ over Unif dataset, 
using {Ours+GKMPS} and {Ours+BBGN} as representative protocols. 
We fix the total number of users to $n = 2^{20}$ and set the number of corrupted users to $k = 1$. 
The corrupted user sends a fixed number of messages per level, and we vary this number across experiments. 
The bottom-group size is set to $\lambda = 256$, the fixed parameters are $\varepsilon = 1, \delta = 1/n^2, \beta = 0.1$, 
the corresponding bottom-level error thresholds are $\theta^{(0)} = 288.65$ for {Ours+BBGN} 
and $\theta^{(0)} = 412.14$ for {Ours+GKMPS}. 
For each configuration, we run $50$ independent trials and report 
(i) the absolute error and 
(ii) the detection rate (i.e., the fraction of runs in which the corrupted user is successfully detected). 
As shown in Fig.~\ref{fig:atk_msg}, 
the detection rate remains nearly zero when the number of malicious messages is small, 
indicating that the corrupted user successfully evades detection. 
In this regime, the overall error grows rapidly, approximately linearly, with the number of malicious messages 
and reaches its maximum near the detection boundary 
(around $400$ messages for Ours+BBGN and $500$ messages for Ours+GKMPS),
which is roughly $(\lambda + \theta_0) - \lambda/2$ as our theoretical expectation.
Once the number of malicious messages exceeds the detection threshold, 
the corrupted user is effectively identified and filtered out, 
leading to a sharp decrease in the overall error.

\section{Conclusion}
\label{sec:Conclusion}
In this paper, we studied the poisoning attacks under the shuffle-DP model, including breaking the privacy and destroying the utility. We have proposed a general defense framework for all union-preserving queries that can convert any shuffle-DP protocol to a version defensible against such attacks with a limited communication cost, and have demonstrated its versatility through a number of common query workloads.
There are some interesting open questions for future research:
    (i) Can our defense framework support more general queries beyond the union-preserving class, such as join and projection queries?
    (ii) The communication overhead is still large compared to its base protocol. \cite{rm2} reduces the communication costs of hierarchical structure queries. Can we adopt the same idea to improve our communication efficiency?

\begin{acks}
    This work has been supported by the NTU–NAP Startup Grant (024584-00001), the Singapore Ministry of Education Tier 1 Grant (RG19/25), the Ant Group Research Intern Program, NSFC (62441238), the National Research Foundation, Singapore, and the Cyber Security Agency of Singapore under the National Cybersecurity R\&D Programme and the CyberSG R\&D Programme Office (Award CRPO-GC3-NTU-001). Any opinions, findings, conclusions, or recommendations expressed in these materials are those of the author(s) and do not reflect the views of the National Research Foundation, Singapore, the Cyber Security Agency of Singapore, or the CyberSG R\&D Programme Office.
    We also thank the anonymous reviewers for valuable suggestions on improving the paper.
\end{acks}

\clearpage

\bibliographystyle{ACM-Reference-Format}
\bibliography{ref}

@inproceedings{balcer2021connecting,
    author = {Victor Balcer and Albert Cheu and Matthew Joseph and Jieming Mao},
    title = {Connecting Robust Shuffle Privacy and Pan-Privacy},
    booktitle = {Proceedings of the ACM-SIAM Symposium on Discrete Algorithms},
    year = {2021},
    pages = {2384-2403},
}

@inproceedings{imola2022differentially,
  title={Differentially private triangle and 4-cycle counting in the shuffle model},
  author={Imola, Jacob and Murakami, Takao and Chaudhuri, Kamalika},
  booktitle={Proceedings of the 2022 ACM SIGSAC Conference on Computer and Communications Security},
  pages={1505--1519},
  year={2022}
}

@inproceedings{balle2020privatesum,
    author = {Balle, Borja and Bell, James and Gasc\'{o}n, Adri\`{a} and Nissim, Kobbi},
    title = {Private Summation in the Multi-Message Shuffle Model},
    year = {2020},
    booktitle = {Proceedings of the ACM SIGSAC Conference on Computer and Communications Security},
    pages = {657–676},
    numpages = {20},
}

@INPROCEEDINGS{cheu2021manipulation,
  author={Cheu, Albert and Smith, Adam and Ullman, Jonathan},
  booktitle={2021 IEEE Symposium on Security and Privacy (SP)}, 
  title={Manipulation Attacks in Local Differential Privacy}, 
  year={2021},
  volume={},
  number={},
  pages={883-900},
  doi={10.1109/SP40001.2021.00001}
}

@inproceedings{dwork2006calibrating,
  title={Calibrating noise to sensitivity in private data analysis},
  author={Dwork, Cynthia and McSherry, Frank and Nissim, Kobbi and Smith, Adam},
  booktitle={Theory of Cryptography: Third Theory of Cryptography Conference, TCC 2006, New York, NY, USA, March 4-7, 2006. Proceedings 3},
  pages={265--284},
  year={2006},
  organization={Springer}
}

@article{kasiviswanathan2011can,
  title={What can we learn privately?},
  author={Kasiviswanathan, Shiva Prasad and Lee, Homin K and Nissim, Kobbi and Raskhodnikova, Sofya and Smith, Adam},
  journal={SIAM Journal on Computing},
  volume={40},
  number={3},
  pages={793--826},
  year={2011},
  publisher={SIAM}
}

@inproceedings{beimel2008distributed,
  title={Distributed private data analysis: Simultaneously solving how and what},
  author={Beimel, Amos and Nissim, Kobbi and Omri, Eran},
  booktitle={Advances in Cryptology--CRYPTO 2008: 28th Annual International Cryptology Conference, Santa Barbara, CA, USA, August 17-21, 2008. Proceedings 28},
  pages={451--468},
  year={2008},
  organization={Springer}
}

@inproceedings{chan2012optimal,
  title={Optimal lower bound for differentially private multi-party aggregation},
  author={Chan, TH Hubert and Shi, Elaine and Song, Dawn},
  booktitle={European Symposium on Algorithms},
  pages={277--288},
  year={2012},
  organization={Springer}
}

@inproceedings{bittau2017prochlo,
  title={Prochlo: Strong privacy for analytics in the crowd},
  author={Bittau, Andrea and Erlingsson, {\'U}lfar and Maniatis, Petros and Mironov, Ilya and Raghunathan, Ananth and Lie, David and Rudominer, Mitch and Kode, Ushasree and Tinnes, Julien and Seefeld, Bernhard},
  booktitle={Proceedings of the 26th symposium on operating systems principles},
  pages={441--459},
  year={2017}
}

@inproceedings{cheu2019distributed,
  title={Distributed differential privacy via shuffling},
  author={Cheu, Albert and Smith, Adam and Ullman, Jonathan and Zeber, David and Zhilyaev, Maxim},
  booktitle={Advances in Cryptology--EUROCRYPT 2019: 38th Annual International Conference on the Theory and Applications of Cryptographic Techniques, Darmstadt, Germany, May 19--23, 2019, Proceedings, Part I 38},
  pages={375--403},
  year={2019},
  organization={Springer}
}

@inproceedings{erlingsson2019amplification,
  title={Amplification by shuffling: From local to central differential privacy via anonymity},
  author={Erlingsson, {\'U}lfar and Feldman, Vitaly and Mironov, Ilya and Raghunathan, Ananth and Talwar, Kunal and Thakurta, Abhradeep},
  booktitle={Proceedings of the Thirtieth Annual ACM-SIAM Symposium on Discrete Algorithms},
  pages={2468--2479},
  year={2019},
  organization={SIAM}
}

@inproceedings{ghazi2020private,
  title={Private counting from anonymous messages: Near-optimal accuracy with vanishing communication overhead},
  author={Ghazi, Badih and Kumar, Ravi and Manurangsi, Pasin and Pagh, Rasmus},
  booktitle={International Conference on Machine Learning},
  pages={3505--3514},
  year={2020},
  organization={PMLR}
}

@inproceedings{ghazi2021differentially,
  title={Differentially private aggregation in the shuffle model: Almost central accuracy in almost a single message},
  author={Ghazi, Badih and Kumar, Ravi and Manurangsi, Pasin and Pagh, Rasmus and Sinha, Amer},
  booktitle={International Conference on Machine Learning},
  pages={3692--3701},
  year={2021},
  organization={PMLR}
}

@inproceedings{luo2022frequency,
  title={Frequency Estimation in the Shuffle Model with Almost a Single Message},
  author={Luo, Qiyao and Wang, Yilei and Yi, Ke},
  booktitle={Proceedings of the 2022 ACM SIGSAC Conference on Computer and Communications Security},
  pages={2219--2232},
  year={2022}
}

@inproceedings{dong2023continual,
  title={Continual observation under user-level differential privacy},
  author={Dong, Wei and Luo, Qiyao and Yi, Ke},
  booktitle={2023 IEEE Symposium on Security and Privacy (SP)},
  pages={2190--2207},
  year={2023},
  organization={IEEE}
}

@article{chaum1982untraceable,
author = {Chaum, David L.},
title = {Untraceable Electronic Mail, Return Addresses, and Digital Pseudonyms},
year = {1981},
volume = {24},
journal = {Commun. ACM},
pages = {84–90},
}

@INPROCEEDINGS{danezis2003Mixminion,  
author={Danezis, G. and Dingledine, R. and Mathewson, N.},  
booktitle={Symposium on Security and Privacy.},   
title={Mixminion: design of a type III anonymous remailer protocol},   
year={2003},  
pages={2-15}
}

@ARTICLE{reed1998onion,  
author={Reed, M.G. and Syverson, P.F. and Goldschlag, D.M.}, 
journal={IEEE Journal on Selected Areas in Communications},   
title={Anonymous connections and onion routing},   
year={1998},  
pages={482-494}}

@inproceedings {roger2004tor,
author = {Roger Dingledine and Nick Mathewson and Paul Syverson},
title = {Tor: The {Second-Generation} Onion Router},
booktitle = {13th USENIX Security Symposium (USENIX Security 04)},
year = {2004}
}

@article{reiter1998crowds,
author = {Reiter, Michael K. and Rubin, Aviel D.},
title = {Crowds: Anonymity for Web Transactions},
year = {1998},
journal = {ACM Transactions on Information and System Security},
pages = {66–92}
}

@InProceedings{badih2021on,
author="Ghazi, Badih and Golowich, Noah and Kumar, Ravi and Pagh, Rasmus and Velingker, Ameya",
title="On the Power of Multiple Anonymous Messages: Frequency Estimation and Selection in the Shuffle Model of Differential Privacy",
booktitle="Advances in Cryptology -- EUROCRYPT 2021",
year="2021",
pages="463--488",
}

@inproceedings{chen21distinct,
title	= {On Distributed Differential Privacy and Counting Distinct Elements},
author	= {Lijie Chen and Badih Ghazi and Ravi Kumar and Pasin Manurangsi},
booktitle = {ITCS},
year	= {2021},
pages	= {56:1-56:18}
}

@inproceedings{cheu2022differentially,
  title={Differentially private histograms in the shuffle model from fake users},
  author={Cheu, Albert and Zhilyaev, Maxim},
  booktitle={2022 IEEE Symposium on Security and Privacy (SP)},
  pages={440--457},
  year={2022},
  organization={IEEE}
}

@INPROCEEDINGS{ishai06cryptography,
  author={Ishai, Yuval and Kushilevitz, Eyal and Ostrovsky, Rafail and Sahai, Amit},
  booktitle={2006 47th Annual IEEE Symposium on Foundations of Computer Science (FOCS'06)}, 
  title={Cryptography from Anonymity}, 
  year={2006},
  volume={},
  number={},
  pages={239-248},
  doi={10.1109/FOCS.2006.25}}

@inproceedings{kohavi1996scaling,
  title={Scaling up the accuracy of naive-bayes classifiers: A decision-tree hybrid.},
  author={Kohavi, Ron and others},
  booktitle={Kdd},
  volume={96},
  pages={202--207},
  year={1996}
}

@inproceedings{pass2006picture,
  title={A picture of search},
  author={Pass, Greg and Chowdhury, Abdur and Torgeson, Cayley},
  booktitle={Proceedings of the 1st international conference on Scalable information systems},
  pages={1--es},
  year={2006}
}

@InProceedings{borja2019the,
	author="Balle, Borja
	and Bell, James
	and Gasc{\'o}n, Adri{\`a}
	and Nissim, Kobbi",
	title="The Privacy Blanket of the Shuffle Model",
	booktitle="CRYPTO",
	year="2019",
}

@inproceedings{ghazi2024pure,
  author       = {Badih Ghazi and
                  Ravi Kumar and
                  Pasin Manurangsi},
  title        = {Pure-DP Aggregation in the Shuffle Model: Error-Optimal and Communication-Efficient},
  booktitle    = {5th Conference on Information-Theoretic Cryptography, {ITC} 2024},
  volume       = {304},
  pages        = {4:1--4:13},
  year         = {2024},
}

@inproceedings {cao2021data,
    author = {Xiaoyu Cao and Jinyuan Jia and Neil Zhenqiang Gong},
    title = {Data Poisoning Attacks to Local Differential Privacy Protocols},
    booktitle = {30th USENIX Security Symposium (USENIX Security 21)},
    year = {2021},
    pages = {947--964},
}

@inproceedings {li2023fine,
    author = {Xiaoguang Li and Ninghui Li and Wenhai Sun and Neil Zhenqiang Gong and Hui Li},
    title = {Fine-grained Poisoning Attack to Local Differential Privacy Protocols for Mean and Variance Estimation},
    booktitle = {32nd USENIX Security Symposium (USENIX Security 23)},
    year = {2023},
    pages = {1739--1756},
}

@inproceedings {wu2022poisoning,
    author = {Yongji Wu and Xiaoyu Cao and Jinyuan Jia and Neil Zhenqiang Gong},
    title = {Poisoning Attacks to Local Differential Privacy Protocols for {Key-Value} Data},
    booktitle = {31st USENIX Security Symposium (USENIX Security 22)},
    year = {2022},
    pages = {519--536},
}

@inproceedings{tong2024data,
    author = {Tong, Wei and Chen, Haoyu and Niu, Jiacheng and Zhong, Sheng},
    title = {Data Poisoning Attacks to Locally Differentially Private Frequent Itemset Mining Protocols},
    year = {2024},
    booktitle = {Proceedings of the 2024 on ACM SIGSAC Conference on Computer and Communications Security},
    pages = {3555–3569},
    numpages = {15},
}

@inproceedings{wang2017malicious,
    author = {Wang, Xiao and Ranellucci, Samuel and Katz, Jonathan},
    title = {Authenticated Garbling and Efficient Maliciously Secure Two-Party Computation},
    year = {2017},
    booktitle = {Proceedings of the 2017 ACM SIGSAC Conference on Computer and Communications Security},
    pages = {21–37},
    numpages = {17},
}

@InProceedings{ivan2012malicious,
    author="Damg{\aa}rd, Ivan
    and Pastro, Valerio
    and Smart, Nigel
    and Zakarias, Sarah",
    title="Multiparty Computation from Somewhat Homomorphic Encryption",
    booktitle="Advances in Cryptology -- CRYPTO 2012",
    year="2012",
    pages="643--662",
}

@InProceedings{bendlin2011malicious,
    author="Bendlin, Rikke
    and Damg{\aa}rd, Ivan
    and Orlandi, Claudio
    and Zakarias, Sarah",
    title="Semi-homomorphic Encryption and Multiparty Computation",
    booktitle="Advances in Cryptology -- EUROCRYPT 2011",
    year="2011",
    pages="169--188",
}

@inproceedings{goldreich1987how,
    author = {Goldreich, O. and Micali, S. and Wigderson, A.},
    title = {How to play ANY mental game},
    year = {1987},
    booktitle = {Proceedings of the Nineteenth Annual ACM Symposium on Theory of Computing},
    pages = {218–229},
    numpages = {12},
}

@article{qardaji2013understanding,
    author = {Qardaji, Wahbeh and Yang, Weining and Li, Ninghui},
    title = {Understanding hierarchical methods for differentially private histograms},
    year = {2013},
    issue_date = {September 2013},
    publisher = {VLDB Endowment},
    volume = {6},
    number = {14},
    issn = {2150-8097},
    url = {https://doi.org/10.14778/2556549.2556576},
    doi = {10.14778/2556549.2556576},
    journal = {Proc. VLDB Endow.},
    month = {sep},
    pages = {1954–1965},
    numpages = {12}
}

@article{hay2010hierarchical,
    author = {Hay, Michael and Rastogi, Vibhor and Miklau, Gerome and Suciu, Dan},
    title = {Boosting the accuracy of differentially private histograms through consistency},
    year = {2010},
    issue_date = {September 2010},
    publisher = {VLDB Endowment},
    volume = {3},
    number = {1–2},
    issn = {2150-8097},
    url = {https://doi.org/10.14778/1920841.1920970},
    doi = {10.14778/1920841.1920970},
    journal = {Proc. VLDB Endow.},
    month = {sep},
    pages = {1021–1032},
    numpages = {12}
}

@inproceedings{barreno2006can,
    author = {Barreno, Marco and Nelson, Blaine and Sears, Russell and Joseph, Anthony D. and Tygar, J. D.},
    title = {Can machine learning be secure?},
    year = {2006},
    booktitle = {Proceedings of the 2006 ACM Symposium on Information, Computer and Communications Security},
    pages = {16–25},
    numpages = {10},
}

@inproceedings{battista2012poisoning,
    author = {Biggio, Battista and Nelson, Blaine and Laskov, Pavel},
    title = {Poisoning attacks against support vector machines},
    year = {2012},
    booktitle = {Proceedings of the 29th International Coference on International Conference on Machine Learning},
    pages = {1467–1474},
    numpages = {8},
}

@inproceedings {fang2020local,
    author = {Minghong Fang and Xiaoyu Cao and Jinyuan Jia and Neil Gong},
    title = {Local Model Poisoning Attacks to {Byzantine-Robust} Federated Learning},
    booktitle = {29th USENIX Security Symposium (USENIX Security 20)},
    year = {2020},
    pages = {1605--1622},
}

@inproceedings{tolpegin2020data,
  title={Data poisoning attacks against federated learning systems},
  author={Tolpegin, Vale and Truex, Stacey and Gursoy, Mehmet Emre and Liu, Ling},
  booktitle={Computer security--ESORICs 2020: 25th European symposium on research in computer security, ESORICs 2020, guildford, UK, September 14--18, 2020, proceedings, part i 25},
  pages={480--501},
  year={2020},
  organization={Springer}
}

@misc{kaggle_sf_salaries_2014,
  title = {San Francisco City Employee Salary Data},
  author = {{Kaggle}},
  year = {2014},
  howpublished = {\url{https://www.kaggle.com/datasets/kaggle/sf-salaries}},
  note = {Accessed: 2025-05-18}
}

@misc{kaggle_brazil_salary_2020,
  title = {Monthly Salary of Public Worker in Brazil},
  author = {{Kaggle}},
  year = {2020},
  howpublished = {\url{https://www.kaggle.com/datasets/gustavomodelli/monthly-salary-of-public-worker-in-brazil}},
  note = {Accessed: 2025-05-18}
}

@inproceedings{chaum1983blind,
  title={Blind signatures for untraceable payments},
  author={Chaum, David},
  booktitle={Advances in Cryptology: Proceedings of Crypto 82},
  pages={199--203},
  year={1983},
  organization={Springer}
}

@article{rm2,
    author = {Luo, Qiyao and Yu, Jianzhe and Dong, Wei and Xu, Quanqing and Yang, Chuanhui and Yi, Ke},
    title = {RM2: Answer Counting Queries Efficiently under Shuffle Differential Privacy},
    year = {2025},
    volume = {3},
    number = {3},
    journal = {Proc. ACM Manag. Data},
    articleno = {210},
    numpages = {24},
}

@inproceedings{dong2022r2t,
author = {Dong, Wei and Fang, Juanru and Yi, Ke and Tao, Yuchao and Machanavajjhala, Ashwin},
title = {R2T: Instance-optimal Truncation for Differentially Private Query Evaluation with Foreign Keys},
year = {2022},
isbn = {9781450392495},
publisher = {Association for Computing Machinery},
address = {New York, NY, USA},
url = {https://doi.org/10.1145/3514221.3517844},
doi = {10.1145/3514221.3517844},
booktitle = {Proceedings of the 2022 International Conference on Management of Data},
pages = {759–772},
numpages = {14},
location = {Philadelphia, PA, USA},
series = {SIGMOD '22}
}

@inproceedings{kifer2011no,
  title={No free lunch in data privacy},
  author={Kifer, Daniel and Machanavajjhala, Ashwin},
  booktitle={Proceedings of the 2011 ACM SIGMOD International Conference on Management of data},
  pages={193--204},
  year={2011}
}

@article{zhang2017privbayes,
  title={Privbayes: Private data release via bayesian networks},
  author={Zhang, Jun and Cormode, Graham and Procopiuc, Cecilia M and Srivastava, Divesh and Xiao, Xiaokui},
  journal={ACM Transactions on Database Systems (TODS)},
  volume={42},
  number={4},
  pages={1--41},
  year={2017},
  publisher={ACM New York, NY, USA}
}

@inproceedings{cormode2018marginal,
  title={Marginal release under local differential privacy},
  author={Cormode, Graham and Kulkarni, Tejas and Srivastava, Divesh},
  booktitle={Proceedings of the 2018 International Conference on Management of Data},
  pages={131--146},
  year={2018}
}

@article{li2024local,
  title={Local differentially private heavy hitter detection in data streams with bounded memory},
  author={Li, Xiaochen and Liu, Weiran and Lou, Jian and Hong, Yuan and Zhang, Lei and Qin, Zhan and Ren, Kui},
  journal={Proceedings of the ACM on Management of Data},
  volume={2},
  number={1},
  pages={1--27},
  year={2024},
  publisher={ACM New York, NY, USA}
}

@article{zhang2025federated,
  title={Federated Heavy Hitter Analytics with Local Differential Privacy},
  author={Zhang, Yuemin and Ye, Qingqing and Hu, Haibo},
  journal={Proceedings of the ACM on Management of Data},
  volume={3},
  number={1},
  pages={1--27},
  year={2025},
  publisher={ACM New York, NY, USA}
}

@article{he2024common,
  title={Common Neighborhood Estimation over Bipartite Graphs under Local Differential Privacy},
  author={He, Yizhang and Wang, Kai and Zhang, Wenjie and Lin, Xuemin and Zhang, Ying},
  journal={Proceedings of the ACM on Management of Data},
  volume={2},
  number={6},
  pages={1--26},
  year={2024},
  publisher={ACM New York, NY, USA}
}

@inproceedings{kulkarni2019answering,
  title={Answering range queries under local differential privacy},
  author={Kulkarni, Tejas},
  booktitle={Proceedings of the 2019 International Conference on Management of Data},
  pages={1832--1834},
  year={2019}
}

@inproceedings{li2020estimating,
  title={Estimating numerical distributions under local differential privacy},
  author={Li, Zitao and Wang, Tianhao and Lopuha{\"a}-Zwakenberg, Milan and Li, Ninghui and {\v{S}}koric, Boris},
  booktitle={Proceedings of the 2020 ACM SIGMOD International Conference on Management of Data},
  pages={621--635},
  year={2020}
}

@inproceedings{ren2022ldp,
  title={LDP-IDS: Local differential privacy for infinite data streams},
  author={Ren, Xuebin and Shi, Liang and Yu, Weiren and Yang, Shusen and Zhao, Cong and Xu, Zongben},
  booktitle={Proceedings of the 2022 international conference on management of data},
  pages={1064--1077},
  year={2022}
}

@article{he2025robust,
  title={Robust Privacy-Preserving Triangle Counting under Edge Local Differential Privacy},
  author={He, Yizhang and Wang, Kai and Zhang, Wenjie and Lin, Xuemin and Zhang, Ying and Ni, Wei},
  journal={Proceedings of the ACM on Management of Data},
  volume={3},
  number={3},
  pages={1--26},
  year={2025},
  publisher={ACM New York, NY, USA}
}

@article{yu2025privrm,
  title={PrivRM: A Framework for Range Mean Estimation under Local Differential Privacy},
  author={Yu, Liantong and Ye, Qingqing and Du, Rong},
  journal={Proceedings of the ACM on Management of Data},
  volume={3},
  number={3},
  pages={1--26},
  year={2025},
  publisher={ACM New York, NY, USA}
}

@article{10.14778/3424573.3424576,
author = {Wang, Tianhao and Ding, Bolin and Xu, Min and Huang, Zhicong and Hong, Cheng and Zhou, Jingren and Li, Ninghui and Jha, Somesh},
title = {Improving utility and security of the shuffler-based differential privacy},
year = {2020},
issue_date = {September 2020},
publisher = {VLDB Endowment},
volume = {13},
number = {13},
issn = {2150-8097},
url = {https://doi.org/10.14778/3424573.3424576},
doi = {10.14778/3424573.3424576},
journal = {Proc. VLDB Endow.},
month = sep,
pages = {3545–3558},
numpages = {14}
}

@inproceedings{mcsherry2009privacy,
  title={Privacy integrated queries: an extensible platform for privacy-preserving data analysis},
  author={McSherry, Frank D},
  booktitle={Proceedings of the 2009 ACM SIGMOD International Conference on Management of data},
  pages={19--30},
  year={2009}
}

@inproceedings{hay2016principled,
  title={Principled evaluation of differentially private algorithms using dpbench},
  author={Hay, Michael and Machanavajjhala, Ashwin and Miklau, Gerome and Chen, Yan and Zhang, Dan},
  booktitle={Proceedings of the 2016 International Conference on Management of Data},
  pages={139--154},
  year={2016}
}

@article{dwork2014algorithmic,
  title={The algorithmic foundations of differential privacy},
  author={Dwork, Cynthia and Roth, Aaron and others},
  journal={Foundations and trends{\textregistered} in theoretical computer science},
  volume={9},
  number={3--4},
  pages={211--407},
  year={2014},
  publisher={Now Publishers, Inc.}
}

\end{document}